\documentclass{article}
\usepackage{amsmath, amsthm}
\usepackage{amssymb}
\usepackage{color}
\usepackage{epsfig}
\usepackage{pstricks}

\input{psfig}

\newtheorem{proposition}{Proposition}
\newtheorem{remark}{Remark}
\newtheorem{corollary}[proposition]{Corollary}
\newtheorem{theorem}[proposition]{Theorem}
\newtheorem{lemma}[proposition]{Lemma}
\newtheorem{maintheorem}{Main Theorem}

\usepackage[margin=1in,dvips]{geometry}

\newcommand{\nabb}{{\nabla} \mkern-13mu /\,}
\newcommand{\lapp}{{\Delta} \mkern-13mu /\,}

\title{Improved decay for solutions to the linear wave equation on a Schwarzschild black hole}
\author{Jonathan Luk}
\date{\today}

\begin{document}
\maketitle

\begin{abstract}
We prove that sufficiently regular solutions to the wave equation $\Box_g\phi=0$ on the exterior of the Schwarzschild black hole obey the estimates $|\phi|\leq C_\delta v_+^{-\frac{3}{2}+\delta}$ and $|\partial_t\phi|\leq C_{\delta} v_+^{-2+\delta}$ on a compact region of $r$ and along the event horizon. This is proved with the help of a new vector field commutator that is analogous to the scaling vector field on Minkowski spacetime. This result improves the known decay rates in the region of finite $r$ and along the event horizon.
\end{abstract}

\tableofcontents

\section{Introduction}

A major open problem in general relativity is that of the nonlinear stability of Kerr spacetimes. These spacetimes are stationary axisymmetric asymptotically flat black hole solutions to the vacuum Einstein equations
$$R_{\mu\nu}=0$$
in $3+1$ dimensions. They are parametrized by two parameters $\left(M, a\right)$, representing respectively the mass and the angular momentum of a black hole. It is conjectured that Kerr spacetimes are stable. In the framework of the initial value problem, the stability of Kerr would mean that for any solution to the vacuum Einstein equations with initial data close to the initial data of a Kerr spacetime, its maximal Cauchy development has an exterior region that approaches a nearby, but possibly different, Kerr spacetime.

Kerr spacetimes have a one-parameter subfamily of spacetimes known as Schwarzschild spacetimes for which $a=0$. The Schwarzschild metric in the so-called exterior region can be expressed as $$g=-\left(1-\frac{2M}{r}\right)dt^2+\left(1-\frac{2M}{r}\right)^{-1}dr^2+r^2d\sigma_{\mathbb S^2}, $$
where $d\sigma_{\mathbb S^2}$ denotes the standard metric on the unit sphere. In view of the nonlinear problem, it is conjectured that a spacetime that is close to Schwarzschild initially will approach a Kerr spacetime that is also close to Schwarzschild, i.e., $a \ll M$. In other words, we can consider the stability of Schwarzschild spacetimes within Kerr spacetimes. (Notice that the Schwarzschild family itself is not asymptotically stable since a Kerr spacetime with small $a$ can be considered as a small perturbation of a Schwarzschild spacetime.)

To tackle the nonlinear stability of Schwarzschild spacetimes within the Kerr family, it is important to first understand the linear waves
$$\Box_g \phi=0$$
on the exterior region of Schwarzschild spacetimes. This can be compared with the nonlinear stability of Minkowski spacetime whose proof requires a robust understanding of the quantitative decay of the solutions to the linear wave equation \cite{CK}, \cite{LR}.

The pointwise decay of the solutions to the linear wave equation on Schwarzschild background is proved in \cite{DRS}, \cite{BSt}. In particular, Dafermos-Rodnianski proved a decay rate of $|\phi|\leq C\left(\max\{1, v\}\right)^{-1}$ everywhere in the exterior region, including along the event horizon \cite{DRS}. The subject of this paper is to improve this decay rate. In particular, we will prove that for arbitrarily small $\delta >0$, $|\phi|\leq C_{\delta,R}\left(\max\{1, v\}\right)^{-\frac{3}{2}+\delta}$ in the region $\{r\leq R\}$ for any $R >2M$, including along the event horizon.

Our proof applies a new vector field commutator $S$ that is analogous to the scaling vector field in Minkowski spacetime. We will show that for solutions to $\Box_g\phi=0$, $\Box_g\left(S\phi\right)$ decays sufficiently towards spatial infinity and only grows mildly towards event horizon. We then prove energy estimates for $S\phi$ with the help of (a slightly modified version of) the energy estimates of $\phi$ in \cite{DRS}. This will enable us to prove the decay of $S\phi$. With this decay, we follow Klainerman-Sideris \cite{KS} to improve the decay rate for $\partial_t\phi$. We also introduce a novel method to improve the decay rates for $\phi$ and its spatial derivatives.

We hope that this improved decay will be relevant for nonlinear problems. We recall for example the wave map equation from $\mathbb R^{3,1}$ to $\mathbb S^2$ given by:
$$\Box_m\phi=\phi\left(\left(\partial_t\phi\right)^2-|\nabla\phi|^2\right).$$
To prove the global existence for small data for this equation, it is insufficient to have $|\partial\phi|\leq C\left(1+|t|\right)^{-1}$. One needs an improved decay $|\partial\phi|\leq C\left(1+|t+r|\right)^{-1}\left(1+|t-r|\right)^{-\delta}$. Moreover, one needs the nonlinearity to satisfy the so-called null condition (see \cite{Knull}). In a future work, we will use the improved decay rate we prove in this paper and study the global well-posedness of small data for a nonlinear wave equation satisfying a null condition on a fixed Schwarzschild background.

In Section 1.1 and 1.2 we will introduce the Schwarzschild spacetime and the class of solutions that we consider. This will introduce the terminologies necessary to state the main theorem in Section 1.3. We will motivate our proof with a comparison with the linear waves on Minkowski spacetime (Section 1.4). We then mention some known results on linear waves on Schwarzschild spacetime (Section 1.5). We especially discuss the work \cite{DRS} whose techniques are important for this paper. We will then provide some heuristics for our proof of the main theorem in the final subsection of the introduction (Section 1.6).

\subsection{Schwarzschild Spacetime}
Schwarzschild spacetime is the spherically symmetric asymptotically flat solution to the vacuum Einstein equations. The Schwarzschild metric in the exterior region is
$$g=-\left(1-\frac{2M}{r}\right)dt^2+\left(1-\frac{2M}{r}\right)^{-1}dr^2+r^2d\sigma_{\mathbb S^2}, $$
where $d\sigma_{\mathbb S^2}$ denotes the standard metric on the unit sphere. It is easy to observe from the metric that the vector field $\partial_t$ is Killing and it is orthogonal to the hypersurfaces $t=\mbox{constant}$. Spacetimes with this property are called static. It is also manifestly spherically symmetry and therefore has a basis of Killing vector fields $\Omega_i$ generating the symmetry. Moreover, Schwarzschild spacetimes are asymptotically flat. This means that the metric approaches the flat metric as we go to spatial infinity ($r\to\infty$).

Synge \cite{Synge} and Kruskal \cite{Kruskal} showed that the Schwarzschild metric can be extended past $r=2M$ as a solution to the vacuum Einstein equations. Its maximal development is usually described by a Penrose diagram, which depicts a conformal compactification of the 4 dimensional manifold quotiented out by spherical symmetry (Figure 1). In this diagram, the coordinate system $\left(t,r>2M, \omega\in\mathbb S^2\right)$ with the metric described above represents the region I, which we will call from now on the exterior region. In the nonlinear stability problem, it is this region that is conjectured to be stable. Extended beyond $r=2M$, the Schwarzschild spacetime contains a black hole (region II in the diagram). Physically, an observer outside the black hole region cannot receive signals emitted inside the black hole. The null hypersurface $r=2M$ separating the exterior region I and the black hole is known as the event horizon $\mathcal H^+$.

\begin{figure}[htbp]
\begin{center}
 
\input{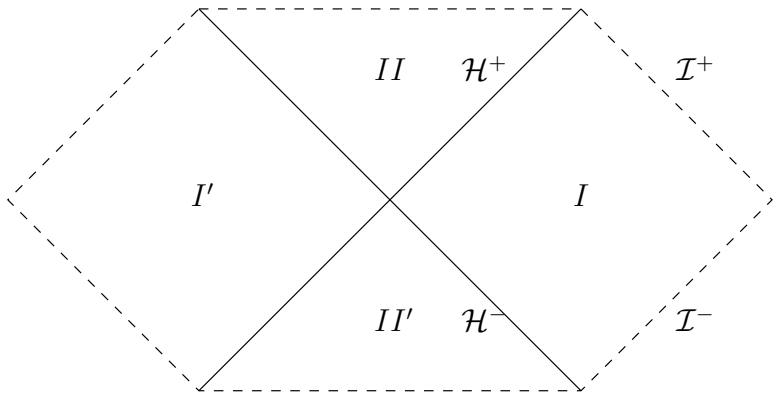}
 
\caption{Schwarzschild spacetime}
\end{center}
\end{figure}

We return to the discussion of the exterior region of the Schwarzschild black hole. For notational convenience, we let 
$$\mu = \frac{2M}{r}.$$
We denote as $r^*$ the Regge-Wheeler tortoise coordinate 
$$r^*=r+2M\log\left(r-2M\right)-3M-2M\log M.$$
In these coordinates, the Schwarzschild metric in the exterior region is given by
$$g=-\left(1-\mu\right)dt^2+\left(1-\mu\right)dr^{*2}+r^2d\sigma_{\mathbb S^2}. $$
Notice that in the above equation we have used both $r^*$ and $r$. Here, and below, we think of $r^*$ as the coordinate and $r$ as a function on $\mathcal{Q}$, with $r\left(q\right)=\sqrt{\frac{Area\left(q\right)}{4\pi}}$, i.e. the physical radius of the 2-sphere under which the metric is symmetric. The coordinate $r^*$ is $+\infty$ at spatial and null infinity; $-\infty$ at the event horizon and $0$ at $r=3M$. The set $\{r=3M\}$ is known as the photon sphere. On this set trapping occurs: there exist null geodesics that lie in this set. In particular, these geodesics neither cross the event horizon nor approach null infinity. This suggests, via geometrical optics considerations, that one has to lose derivatives while proving energy estimates. We will return to this point when we discuss the vector field $X$.\\
We notice that as in the coordinates $\left(t, r, \omega\right)$, $\partial_t$ and $\Omega$ are Killing in the $\left(t, r^*, \omega\right)$ coordinates.\\
We also define the retarded and advanced Eddington-Finkelstein coordinates $u$ and $v$ by
$$t=v+u, r^*=v-u.$$
At the event horizon $\mathcal H^+$, $u=+\infty$. At future null infinity $\mathcal I^+$, $v=+\infty$. Notice that in these coordinates, the metric is given by
$$-4\left(1- \mu \right)dudv+r^2d\sigma_{\mathbb S^2}.$$
In particular, this shows that $\partial_u$, $\partial_v$ are null.

\subsection{Wave Equation and Class of Solutions}
We would like to study the solutions to $\Box_g \phi=0$ in the exterior region of Schwarzschild. Written in local coordinates,
$$\Box_g \phi = -\left(1-\mu \right)^{-1}\partial_t^2\phi+\left(1-\mu \right)^{-1}r^{-2}\partial_{r^*}\left(r^2\partial_{r^*}\phi\right)+\lapp\phi,$$
where $\lapp$ denotes the Laplace-Beltrami operator on the standard 2-sphere with radius $r$.\\
We notice that $\Box_g$ commutes with Killing vector fields. In particular, $\Box_g \phi=0$ implies $\Box_g \partial_t \phi=0$ and $\Box_g \Omega \phi = 0$.\\
The decay result that we prove apply to solutions to the wave equation that is in some energy class initially. We define the energy classes using currents of vector fields. We will briefly introduce the relevant concepts here in order to present the energy classes. A more detailed description of the vector fields will be presented in the next section.\\
Define the energy-momentum tensor $$T_{\mu\nu}=\partial_{\mu}\phi\partial_{\nu}\phi-\frac{1}{2}g_{\mu\nu}\partial^\alpha\phi\partial_\alpha\phi.$$
Given a vector field $V^\mu$, we define the associated current
$$J^V_{\mu}\left(\phi\right)=V^{\nu}T_{\mu\nu}\left(\phi\right)$$
and the modified current $$J^{V,w}_{\mu}\left(\phi\right)=J^{V}_{\mu}\left(\phi\right)+\frac{1}{8}\left(w\partial_{\mu}\phi^2-\partial_{\mu}w\phi^2\right).$$
To define the energy classes we need two vector fields:
$$N=\partial_t + \frac{y_1\left(r^*\right)}{1-\mu}\partial_u +y_2\left(r^*\right)\partial_v, $$
$$Z=u^2\partial_u+v^2\partial_v, $$
where $y_1, y_2 > 0$ are supported near the event horizon with $y_1=1$, $y_2=0$ at the event horizon. The precise form of $y_1, y_2$ will be defined later. 
We also define a modifying function for the associated current of $Z$:
$$w^Z =\frac{2tr^*\left(1-\mu \right)}{r}.$$
We note here that 
$$\int  J^N_\mu\left(\phi\right)n^{\mu}_{t_0}dVol_{t_0} \sim \int_{-\infty}^{\infty} \int_{\mathbb S^2} \left(\frac{\left(\partial_u \phi\right)^2}{1-\mu }+\left(\partial_v \phi\right)^2 + |\nabb \phi|^2\right)r^2 dAdr^*,$$
\begin{equation*}
\begin{split}
&\int  J^{Z,w^Z}_\mu\left(\phi\right)n^{\mu}_{t_0}dVol_{t_0}\\
\sim &\int_{-\infty}^{\infty} \int_{\mathbb S^2} \left(u^2\left(\partial_u \phi\right)^2+v^2\left(\partial_v \phi\right)^2 + \left(1-\mu \right)\left(\left(u^2+v^2\right)|\nabb \phi|^2+ \frac{t^2+\left(r^*\right)^2}{r^2}\phi^2\right)\right)r^2 dAdr^*,
\end{split}
\end{equation*}
where $dVol_{t_0}$ is the volume form of the slice $\{t=t_0\}$ (See Section 2).\\
Let $S = t\partial_t+r^*\partial_{r^*}$.\\
Define
\begin{equation*}
\begin{split}
E_0\left(\phi\right)=&\int \left(\displaystyle\sum_{k=0}^3 J^N_\mu\left(\Omega^k \phi\right)n^{\mu}_{t_0} + \displaystyle\sum_{k=0}^2 J^{Z,w^Z}_\mu\left(\Omega^k \phi\right)n^{\mu}_{t_0}\right) dVol_{t_0},\\
E_1\left(\phi \right)=&E_0\left(S\phi \right)+ \sum_{m=0}^1 \sum_{k=0}^{4-m} E_0\left(\partial_t^m \Omega^k \phi \right),\\
E_2\left(\phi \right)=&\sum_{m=0}^2 \sum_{k=0}^{2-m} E_1\left(\partial_t^m \Omega^k \phi\right),\\
E_3\left(\phi \right)=&\sum_{m=0}^1 \sum_{k=0}^1 E_0\left(\partial_t^m \Omega^k \phi \right)+E_1\left(\phi \right),\\
E_4\left(\phi \right)=&\sum_{m=0}^2 \sum_{k=0}^{4-m} E_0\left(\partial_t^m \Omega^k\phi \right)+ \sum_{m=0}^{2} \sum_{k=0}^{2-m}  E_1\left(\partial_t\Omega^k\phi \right).
\end{split}
\end{equation*}
We notice that the boundedness of these quantities should be thought of as requirements of regularity and decay. In the above, $E_0$, $E_1$, $E_2$, $E_3$ and $E_4$ requires 4, 8, 10, 8 and 10 derivatives respectively. In terms of spatial decay, all the energy classes require decay of $\phi$ at spatial infinity. However, we note that $\phi$ is not required to decay toward the bifurcate sphere ($\mathcal H^+ \cap\mathcal H^-$ in Figure 1). In the following, we will work as if $\phi$ is smooth and supported away from spatial infinity. This assumption can be removed by a standard approximation argument.
\subsection{Statement of the Main Theorem}
We prove both pointwise decay and energy decay for solutions of $\Box_g\phi=0$. From this point onwards, we assume $t_* >1$, $v_* >1$.
\begin{maintheorem}
Suppose $\phi$ is a solution to the wave equation on the exterior region of Schwarzschild, i.e., $\Box_g\phi=0$. Then
for any $\delta>0$ and any $2M\le R<\infty$ 
\begin{enumerate}
\item Pointwise Decay of $\phi$

$$|\phi\left(v_*,u\right)|\leq C_{\delta, R} v_*^{-\frac{3}{2}+\delta} E_2^{\frac{1}{2}}\left(\phi\right) \quad\mbox{ for } r \leq R.$$
\item Pointwise Decay of Derivatives of $\phi$

$$|\nabla\phi\left(v_*,u\right)|\leq C_{\delta, R} v_*^{-\frac{3}{2}+\delta}\displaystyle\sum_{k+m\leq 1} E_2^{\frac{1}{2}}\left(\partial_t^m\Omega^k\phi\right)\quad\mbox{ for } r \leq R,$$  where $\nabla$ denotes any derivatives.

\item Decay of Non-degenerate Energy in the Region $r\leq r_2$

\begin{equation*}
\begin{split}
&\int_{r_1^*}^{r_2^*} \left(\nabla \phi\left(t_*\right)\right)^2 dVol_{t_*} +\int_{\{r^*\leq r_3^*\}} \left(\frac{1}{1-\mu }\left(\partial_u \phi\right)^2+\left(1-\mu\right)|\nabb\phi|^2\right) dA du_{\{v=v_*\}}\\
\leq &C_{\delta,r_1^*,r_2^*} \min{\{t_*, v_* \}}^{-3+\delta}\displaystyle\sum_{k+m \leq 1} E_1\left(\partial_t^m\Omega^k\phi\right),
\end{split}
\end{equation*}
for any $r_1^*$, $r_2^*$, $r_3^*$.

\end{enumerate} 
\end{maintheorem}
\begin{remark}
The integral in statement 3 represents the part of non-degenerate energy restricted to the region $r\le r_2$ (See Section 3.5). It should be compared with the corresponding part of the (degenerate) energy generated by the vectorfield $T=\frac \partial{\partial t}$
$$\int_{r_1^*}^{r_2^*} \left(\nabla \phi\left(t_*\right)\right)^2 dVol_{t_*} +\int_{\{r^*\leq r_3^*\}} \left(\left(\partial_u \phi\right)^2+\left(1-\mu\right)|\nabb\phi|^2\right) dA du_{\{v=v_*\}}.$$
The non-degenerate energy provides strong control of the solution near the event horizon ${\mathcal H}^+$. 
Such control also easily allows one to extend the obtained decay results past the event horizon.
\end{remark}
For the time derivatives, we have better decay estimates:
\begin{maintheorem}
Suppose $\phi$ is a solution to the wave equation on the exterior region of Schwarzschild, i.e., $\Box_g\phi=0$. Then for any $\delta >0$
\begin{enumerate}
\item Pointwise Decay of $\partial_t\phi$

$$|\partial_t\phi\left(v^*\right)|\leq C_\delta v_*^{-2+\delta}E_4^{\frac{1}{2}}\left(\phi \right) \quad\mbox{ for } r^*\leq  \frac{t_*}{2}.$$

\item Decay of Non-degenerate Energy of $\partial_t\phi$ in the Region $r^* \leq \frac{t_*}{2}$

\begin{equation*}
\begin{split}
&\int_{{r_1}^*}^{\frac{t_*}{2}} \left(\nabla\partial_t \phi\left(t_*\right)\right)^2 dVol_{t_*}+\int_{\{r^*\leq r_2^*\}} \left(\frac{1}{1-\mu }\left(\partial_u \partial_t\phi\right)^2+\left(1-\mu\right)|\nabb\partial_t\phi|^2\right) dA du_{\{v=v_*\}}\\
\leq &C_{\delta, r^*} \min\{t_*, v_*\}^{-4+\delta}E_3\left(\phi \right).
\end{split}
\end{equation*}
for any $r_1^*$, $r_2^*$.

\end{enumerate}
\end{maintheorem}
We would like to point out that the pointwise decay rates in both theorems apply to region of finite $r$ and along the event horizon.

\subsection{The Case of Minkowski Spacetime}
At this point, we would like to discuss some decay results for the linear wave equation on Minkowski spacetimes. We would like to especially highlight techniques that are relevant to our result. In Minkowski space $\mathbb R^{3,1}$, the solutions to the wave equation with initial conditions $\phi\left(t=0,x\right)=\phi_0$ and $\partial_t\phi\left(t=0,x\right)=\phi_1$ can be written as
\begin{equation}\label{funsol}
\phi\left(t,x\right)=\frac{1}{4\pi t^2} \left(\partial_t \int_{\mathbb S^2} t\phi_0\left(x+ty\right)dA\left(y\right)+\int_{\mathbb S^2} t\phi_1\left(x+ty\right)dA\left(y\right)\right).
\end{equation} 
This formula implies immediately that $|\phi\left(t,x\right)|\leq \frac{C}{t_+}$, where $t=\max\{t,1\}$. This decay is optimal in the variable $t$. However, improved decay can be seen in the null coordinates $v=\frac{1}{2}\left(t+r\right)$ and $u=\frac{1}{2}\left(t-r\right)$, where $r^2=\displaystyle\sum_{i=0}^3 x_i^2$. In particular, (\ref{funsol}) implies the strong Huygens' Principle, asserting that $\phi$ with compactly supported initial data is compactly supported in the variable $u$. Therefore, denoting $v_+=\max\{v, 1\}$, $u_+=\max\{u, 1\}$, we have in particular
$$|\phi|\leq \frac{C_N}{v_+ u_+^N}, \quad \forall N\geq0.$$
If we just focus on the region $\{r\leq\frac{t}{2}\}$, where $t\sim v\sim u$, the decay can be written as
$$|\phi|\leq \frac{C_N}{t_+^N}, \quad \forall N\geq0.$$

However, the use of the representation formula (\ref{funsol}) is not available on perturbations of the Minkowski spacetime. In \cite{CK}, \cite{LR}, a more robust understanding of the decay of the linear waves was necessary. This was achieved by the vector field method. Let $\phi$ be a solution to the linear wave equation on Minkowski spacetime, $\Box_m\phi=0$. Define the energy momentum tensor
$$T_{\mu\nu}=\partial_{\mu}\phi\partial_{\nu}\phi-\frac{1}{2}m_{\mu\nu}\partial^\alpha\phi\partial_\alpha\phi.$$
Notice that the wave equation implies that the energy momentum tensor is divergence free, i.e.,
$$\nabla^\mu T_{\mu\nu}=0.$$
Given a vector field $V^\mu$, we define the associated currents
$$J^V_{\mu}\left(\phi\right)=V^{\nu}T_{\mu\nu}\left(\phi\right),$$
$$K^V\left(\phi\right)=\frac{1}{2}T_{\mu\nu}\left(\nabla^\mu V^\nu+\nabla^\nu V^\mu\right);$$ 
and the modified currents $$J^{V,w^V}_{\mu}\left(\phi\right)=J^{V}_{\mu}\left(\phi\right)+\frac{1}{8}\left(w^V\partial_{\mu}\phi^2-\partial_{\mu}w^V\phi^2\right),$$
$$K^{V,w^V}\left(\phi\right)=K^V\left(\phi\right)+\frac{1}{4}w^V\partial^{\nu}\phi\partial_{\nu}\phi-\frac{1}{8}\Box_g w^V\phi^2,$$
where $w^V$ is some scalar function associated to the vector field $V$.
Since the energy momentum tensor is divergence free, it is easy to check that
$$\nabla^{\mu}J^{V}_{\mu}\left(\phi\right)=K^{V}\left(\phi\right),$$
$$\nabla^{\mu}J^{V,w^V}_{\mu}\left(\phi\right)=K^{V,w^V}\left(\phi\right).$$
Notice that $K^V\left(\phi\right)=0$ whenever $V$ is Killing. In this case $J^{V}_{\mu}\left(\phi\right)$ is divergence free. Therefore, for any solution $\phi$ and Killing vector field $V$, there is a conservation law:
$$\int_{t=t_1} J^{V}_0\left(\phi\right) dx_{t_1}=\int_{t=t_0} J^{V}_0\left(\phi\right) dx_{t_0}.$$
This is a manifestation of Noether's Theorem, which states that a differentiable one-parameter family of symmetries gives rise to a conservation law. We call the vector field $V$ in this application a multiplier because we ``multiply'' it to the energy momentum tensor. An example of this is to take the Killing vector field $\partial_t$ and derive the energy conservation law
\begin{equation}\label{energyconservation}
\int \left(\left(\partial_t\phi\right)^2+\sum_{i=1}^3\left(\partial_{x_i}\phi\right)\right)dx_t=\int \left(\left(\partial_t\phi\right)^2+\sum_{i=1}^3\left(\partial_{x_i}\phi\right)^2\right)dx_{t_0}.
\end{equation}
Besides being multipliers, vector fields can also be used as commutators. This means that we commute the vector fields with $\Box_m$. For example, since $\partial\in\{\partial_t, \partial_{x_i}\}$ is Killing, $[\Box_m,\partial]=0$ and therefore $\Box_m \left(\partial\phi\right)=0$. Then the energy conservation law (\ref{energyconservation}) can be applied to $\partial\phi$ and we can control the $L^2$ norm of the derivatives of $\phi$ of orders 1 and 2. Then using a Sobolev-type inequality $||\phi||_{L^\infty\left(\mathbb R^3\right)} \leq C ||\phi||^{\frac{1}{2}}_{\dot{H}^1\left(\mathbb R^3\right)}||\phi||^{\frac{1}{2}}_{\dot{H}^2\left(\mathbb R^3\right)}$ (which holds for compactly supported functions), uniform boundedness of the solutions to the wave equation can be proved. The Killing vector fields $\Omega_i$ generating the spherical symmetry can also be used as commutators. This is especially useful because compared to the angular derivatives, $\Omega_i$ has an extra factor of $r$, i.e., $\Omega\sim r\nabb$. This allows one to prove in \cite{Klext} that for $\phi$ decaying sufficiently fast at spatial infinity: 
$$|\phi|\leq \frac{C}{r}\sum_{m=1}^2\sum_{k=0}^{2}||\partial_r^m \Omega^k \phi||_{L^2\left(\mathbb R^3\right)},$$
which implies a decay in the region $\{r>\frac{t}{2}\}$:
$$|\phi| \leq \frac{C}{r} \leq \frac{C}{v_+}$$ 
after applying (\ref{energyconservation}) to $\Omega^k\phi$. 

To achieve decay of $\phi$ in $\{r\leq\frac{t}{2}\}$, one can use the conformally Killing vector field $Z=\left(t^2+r^2\right)\partial_t+2tr\partial_r$ introduced by Morawetz \cite{Morawetz}. In this case, $K^Z\left(\phi\right)\not=0$. Nevertheless, by defining $w^Z=2t$, $K^{Z,w^Z}\left(\phi\right)=0$ and therefore $\int J_0^{Z,w^Z}\left(\phi\right) dx_t$ is a conserved quantity. Moreover, some algebraic manipulation would show
$$\int J_0^{Z,w^Z}\left(\phi\right)dx_t \geq c \int \left(v^2\left(\partial_v\phi\right)^2+u^2\left(\partial_u\phi\right)^2+\left(v^2+u^2\right)\left(\frac{\phi^2}{r^2}+|\nabb\phi|^2\right)\right) dx,$$
where $\nabb$ denotes the angular derivatives. The conserved nonnegative quantity $\int J^{Z,w^Z}_0\left(\phi\right) dx_t$ is known as the conformal energy. For the region $\{r\leq\frac{t}{2}\}$, notice that the boundedness of the conformal energy implies a local energy decay
$$\int_{\{r\leq\frac{t}{2}\}} \left(\frac{\phi^2}{r^2}+\left(\partial_v\phi\right)^2+\left(\partial_u\phi\right)+|\nabb\phi|^2\right)dx_t \leq \frac{C}{t_+}.$$
After considering the equations $\Box_m\left(\partial^k\phi\right)=0$, Sobolev embedding would imply the pointwise decay $|\phi|\leq \frac{C}{t_+},$ for  $r\leq\frac{t}{2}$. Notice that in this region $t_+$ is comparable to $v_+$. Therefore, we have in the whole of Minkowski spacetime
$$|\phi|\leq \frac{C}{v_+}.$$

Klainerman-Sideris \cite{KS} showed that more decay can be achieved in the interior region $\{r\leq\frac{t}{2}\}$ for the derivatives of $\phi$. They used the scaling vector field $S=t\partial_t+r\partial_r$ as a commutator. Notice that $S$ is conformally Killing and $[\Box_m,S]=2\Box_m$. In particular, if one has $\Box_m\phi=0$, then $\Box_m\left(S\phi\right)=S\Box_m\phi+2\Box_m\phi=0$. Therefore, any decay results that hold for $\phi$ also hold for $S\phi$. Klainerman-Sideris \cite{KS} showed that
$$\sum_{\partial\in\{\partial_t,\partial_{x_i}\}}||u_+\partial\partial_t\phi||_{L^2\left(\mathbb R^3\right)}\leq C\sum_{\partial\in\{\partial_t,\partial_{x_i}\}}\left(||\partial S\phi||_{L^2\left(\mathbb R^3\right)}+||\partial\phi||_{L^2\left(\mathbb R^3\right)}+||\partial^2\phi||_{L^2\left(\mathbb R^3\right)}+||\partial\Omega\phi||_{L^2\left(\mathbb R^3\right)}\right).$$
By cutting off appropriately and using the local energy decay estimates,
$$||\partial\partial_t\phi||_{L^2\left(\{r\leq\frac{t}{2}\}\right)}\leq \frac{C}{t_+^2} \quad 
\mbox{in } \{r\leq\frac{t}{2}\}$$ since $\frac{1}{u_+}\leq\frac{C}{t_+}$ in this region. Again, using the Sobolev-type inequality above, one shows that $|\partial_t\phi|\leq \frac{C}{t_+^2}$ in $\{r\leq\frac{t}{2}\}$. The other derivatives can also be estimated first by elliptic estimates and then the Sobolev inequality, since $||u_+\partial_t^2\phi||_{L^2\left(\mathbb R^2\right)}=||u_+\Delta\phi||_{L^2\left(\mathbb R^2\right)}$ by the linear wave equation. Therefore
$$|\partial\phi|\leq \frac{C}{t_+^2} \quad \mbox{in }\{r\leq \frac{t}{2}\}.$$ We remark that in \cite{KS}, the improved decay in $\{r\leq\frac{t}{2}\}$ can also be proved for the function $\phi$ itself by inverting the Laplacian. As we proceed to prove the analogous decay on Schwarzschild spacetimes, we will avoid doing so. This is because on Schwarzschild spacetimes, it is impossible to invert the Laplacian for functions that do not vanish on the bifurcate sphere ($\mathcal H^+ \cap \mathcal H^-$ in Figure 1).

\subsection{Some Known Results on the Wave Equation on Schwarzschild Spacetimes}

We now turn to the corresponding problem for linear waves on Schwarzschild spacetimes. The problem of the uniform boundedness of solutions to $\Box_g \phi=0$ on the exterior of Schwarzschild occupied the physics community for some time. The first mathematically rigorous result was obtained by Wald \cite{W} for solutions vanishing on the bifurcate sphere ($\mathcal H^+ \cap \mathcal H^-$ in Figure 1). Kay-Wald \cite{KW} later removed this restriction and proved the uniform boundedness of a more general class of solutions. They used the energy conservation law given by using $\partial_t$ as a multiplier as well as the Killing fields $\{\partial_t, \Omega_i\}$ as commutators. The decay rates
\begin{equation}\label{phidecay}
|\phi|\leq C v_+^{-1},\quad \forall r\geq 2M,
\end{equation}
\begin{equation*}
|r\phi|\leq C_R u_+^{-\frac{1}{2}}, \quad \forall r\geq R,
\end{equation*}
where $v_+=\max\{v, 1\}$, $u_+=\max\{u, 1\}$ and $C_R$ depends only on an appropriate norm of the initial data, for sufficiently regular solutions to $\Box_g \phi=0$, were proved by Dafermos-Rodnianski in \cite{DRS}. We note that the decay rate (\ref{phidecay}) holds in the entire exterior region of Schwarzschild spacetimes, including along the event horizon. In addition to the vector fields in \cite{W},\cite{KW}, their approach employed several other (non-Killing!) vector fields. One is an analog of the Morawetz vector field $Z$ in Minkowski spacetime. It has an associated nonnegative quantity which we will call the conformal energy. It has weights similar to that of the conformal energy on Minkowski spacetime so that its boundedness would imply a local energy decay. Another is a vector field of the form $X=f\left(r^*\right)\partial_{r^*}$. The construction of this vector field was motivated by \cite{LSo}. Unlike other multipliers, $X$ is constructed so that $K^{X,w^X}\left(\phi\right)$ (instead of $J^{X,w^X}\left(\phi\right)$) can be controlled. This is used to estimate some energy quantity integrated over spacetime, in particular error terms from the ``conservation law'' of the conformal energy. The estimates of $X$ are iterated together with that of $Z$ to achieve the boundedness of the conformal energy. This then implies the decay of $\phi$ away from the event horizon. The estimate associated to $X$ can be thought of as an \emph{integrated in time} local energy decay. It was extensively studied in \cite{A},\cite{BSX},\cite{BSt},\cite{DRN},\cite{DRS},\cite{MMTT}. 

In addition, \cite{DRS} introduced a new - red shift vector field - which takes advantage of the geometry of the event horizon and is used crucially in proving the decay rate close to and along the event horizon. This vector field is one of the few stable features of the Schwarzschild spacetime. In particular, it can be used to give a more robust proof of boundedness of the solutions to the linear wave equation on Schwarzschild spacetimes. It also plays key roles in the boundedness results for the linear wave equation on small axisymmetric stationary perturbations of Schwarzschild spacetimes and in the decay result for the linear wave equation on slowly rotating Kerr spacetimes \cite{DRK}, \cite{DRL}. As we will see later, it will make a crucial appearance in this article to achieve the improved decay rate along the event horizon.

The study of pointwise decay was carried out independently by Blue-Sterbenz \cite{BSt}. They showed a similar quantitative decay result for initial data vanishing on the bifurcate sphere, with a decay rate that is weaker than \cite{DRS} along the event horizon. In the proof they used analogues of the vector fields $Z$ and $X$ but not the vector field $Y$. Strichartz estimates for solutions of the wave equation on Schwarzschild background were shown in \cite{MMTT}. We refer the readers to Sections 3 and 4 in \cite{DRL} for further references on this problem.

Considerable attention has also been given to the problem of decay of solutions of the wave equations on the Schwarzschild spacetime restricted to a {\it fixed} spherical harmonic $\phi_\ell$ arising in the decomposition $\phi(t,r,\omega)=\displaystyle\sum_\ell \phi_\ell(r,t)Y_\ell(\omega)$, $\omega\in{\Bbb S}^2$. Such results for a fixed spherical harmonic have been obtained in \cite{DRPL}, \cite{DSS}, \cite{Kr}, \cite{MS}. We refer the readers to Section 4.6 in \cite{DRL} for a more detailed discussion.

\subsection{Outline of the Proof}

Our proof uses ideas from Dafermos-Rodnianski \cite{DRS} and Klainerman-Sideris \cite{KS}. In addition to the arguments used in \cite{DRS}, we introduce a vector field $S=t\partial_t+r^*\partial_{r^*}$ which is analogous to the scaling vector field in Minkowski spacetime. Since Schwarzschild spacetimes are asymptotically flat, $S$ is still an ``asymptotic conformal symmetry'' generating an ``asymptotic almost conservation law''. However, the error terms away from spacelike infinity are in general large. To see this more concretely, we recall that on Minkowski spacetimes, $\Box_m\phi=0$ implies $\Box_m\left(S\phi\right)=0$. This does not hold in Schwarzschild spacetimes. Nevertheless, for $\Box_g \phi=0$, we still have a (schematic) equation $\Box_g\left(S\phi\right)= h(r)\left(\nabla\phi+\nabla^2\phi\right)$ with $h\to 0$ as $r\to\infty$. The strategy is then to go through the argument in Dafermos-Rodnianski \cite{DRS} and control the error terms that arise from $\Box_g\left(S\phi\right)\not = 0$. To do so, we use a slightly modified version of the energy estimates that are available from the proof in \cite{DRS}.

As in later parts of the paper, we define $\psi=S\phi$. We would like to prove energy estimates for $\psi$ similar to those for $\phi$ that are established in \cite{DRS}, except for a loss of an arbitrarily small power of $t$. A key estimate that will be used to prove the main theorem is:
\begin{equation}\label{Xest}
\int_{t_1}^{t_2}\int_{-\frac{t}{2}}^{\frac{t}{2}} \int_{\mathbb S^2} \left(\psi^2+ \left(\nabla\psi\right)^2\right)\chi\left(r^*\right)dVol \leq C_\delta t_1^{-2+\delta},
\end{equation}
where $\chi$ is some weight and $t_1\leq t_2\leq\left(1.1\right)t_1$. A similar estimate is available with $\psi$ replaced by $\phi$ from \cite{DRS} using the $X$ vector field. In order to prove this, we argue in a similar fashion. We want to show, using the vector field $X$, that for $t_1$, $t_2$ as above
$$\int_{t_1}^{t_2}\int_{-\frac{t}{2}}^{\frac{t}{2}} \int_{\mathbb S^2} \left(\psi^2+ \left(\nabla\psi\right)^2\right)\chi\left(r^*\right)dVol \leq C_\delta t_1^{-2+\delta}\{\mbox{conf. energy}(\psi)\},$$
where the conformal energy is the current of the vector field $Z$ on the boundary $\{t=t_i\}$. We then hope to show
$$\{\mbox{conf. energy}(\psi)\mbox{ at } t_2\}\leq C \left(\{\mbox{conf. energy}(\psi)\mbox{  at } t_1\}+\int_{t_1}^{t_2}\int_{-\frac{t}{2}}^{\frac{t}{2}} \int_{\mathbb S^2} \left(\psi^2+ \left(\nabla\psi\right)^2\right)\chi\left(r^*\right)dVol\right).$$
We then iterate two inequalities to obtain (\ref{Xest}) as in \cite{DRS}.

The main difficulty in actually carrying out the above procedure is that each step is only true modulo some error terms that need not be small. These are error terms arising from the fact that $\psi$ does not satisfy the homogeneous wave equation, but only satisfies an inhomogeneous wave equation, which schematically can be thought of as $\Box_g \psi=h(r^*)\left(\nabla\phi+\nabla^2\phi\right)$. If one applies the vector field method to this equation, one would generate an error term of the form 
\begin{equation}\label{error}
\int_{t_1}^{t_2}\int_{-\infty}^{\infty} \int_{\mathbb S^2} V^\mu\partial_\mu \psi h(r^*)\left(\nabla\phi+\nabla^2\phi\right) dVol,
\end{equation}
for the vector fields $V\in\{\partial_t, X=f(r)\partial_{r^*}, Z=\left(t^2+\left(r^*\right)^2\right)\partial_t+2tr^*\partial_r^*\}$. (In practice there is still another error term if one uses the modified current, but since it can be controlled similarly, we omit the technicalities here.) Applying Cauchy-Schwarz, we can control (\ref{error}) by
\begin{equation}\label{errorCS}
\left(\int_{t_1}^{t_2}\int_{-\infty}^{\infty} \int_{\mathbb S^2} \left(\nabla\psi\right)^2 dVol\right)^{\frac{1}{2}}\left(\int_{t_1}^{t_2}\int_{-\infty}^{\infty} \int_{\mathbb S^2} \tilde{h}\left(r^*\right)\left((\nabla\phi)^2+(\nabla^2\phi)^2\right) dVol\right)^{\frac{1}{2}}
\end{equation}
We control the first factor by some energy quantities of $\psi$ which we are in the process of proving. They are set up so that we can estimate them with a bootstrap argument. In order that the bootstrap can close, we would need to show that the second factor decays or does not grow as $t_1, t_2 \to \infty$ (for example with $t_2=(1.1)t_1$). The precise rate of decay that is necessary depends on the vector field $V$ under consideration and is ultimately dictated by what the bootstrap argument requires. To achieve this, we recall the energy estimates derived from the $X$ vector field in \cite{DRS}. In particular, we have
\begin{equation}\label{Xphiwhole}
\int_{t_1}^{t_2}\int_{-\infty}^{\infty} \int_{\mathbb S^2}  \left(\nabla\phi\right)^2\chi\left(r^*\right)dVol \leq C,
\end{equation}
\begin{equation}\label{Xphimiddle}
\int_{t_1}^{t_2}\int_{-\frac{t}{2}}^{\frac{t}{2}} \int_{\mathbb S^2}  \left(\nabla\phi\right)^2\chi\left(r^*\right)dVol \leq C t_1^{-2},
\end{equation}
where $\chi$ is a weight that decays at spatial infinity. (\ref{Xphimiddle}) gives good control for the second factor in (\ref{errorCS}) for the region $\{-\frac{t}{2}\leq r^*\leq\frac{t}{2}\}$ as long as $\tilde{h}$ and $\chi$ behaves appropriately. We will slightly improve the weight $\chi$ from \cite{DRS} so that we have, loosely speaking, $\tilde{h}\left(r^*\right)\leq C\left(1+|r^*|\right)^{-2}\chi\left(r^*\right)$. This would give control for the second factor in (\ref{errorCS}) for the region $\{-\frac{t}{2}\leq r^*\leq\frac{t}{2}\}$. For the regions $\{r^*\leq-\frac{t}{2}\}$ and $\{r^*\geq\frac{t}{2}\}$, $\tilde{h}\left(r^*\right)\leq C\left(1+|r^*|\right)^{-2}\chi\left(r^*\right)\leq C\left(1+t\right)^{-2}\chi\left(r^*\right)$. Then we can control the second factor in (\ref{errorCS}) in this region with (\ref{Xphiwhole}) and the extra factor of $(1+t)^{-2}$. The reader should keep in mind that these are only heuristics and are not true if directly applied. The actual estimates for these error terms are slightly more involved considering firstly that $V^\mu$ might grow $t$; and secondly that we do not have energy estimates that control every derivatives of $\psi$; and thirdly that some error terms would tend to infinity as $r$ approaches the event horizon. The relevant estimates will be proved in Section 5.

In \cite{KS}, the estimates for $\psi$ are used to prove the decay for $\partial_t\phi$ in Minkowski spacetime. We show that it is possible to argue similarly to prove the decay for $\partial_t\phi$ in Schwarzschild spacetimes (Section 7.2). Recall that in \cite{KS}, one then proceeds with elliptic estimates to prove the decay for other derivatives. However, on Schwarzschild spacetimes, if we are to prove an $L^2$ elliptic estimate, we are bound to have some lower order terms involving only one derivative of $\phi$. These terms cannot be controlled by the estimates of $\psi$ and therefore we are unable to use a similar method to prove the decay of the the spatial derivatives of $\phi$.

Therefore, we introduce in this paper a new method, based on a novel application of $S$,to prove the decay for the function $\phi$ as well as its derivatives in spatial directions (Section 7.1). We notice that by (\ref{Xphimiddle}), 
$$\int_{t_1}^{t_2} \int_{r_1^*}^{r_2^*} \left(\phi^2+\left(\partial_{r^*}\phi\right)^2\right) dVol \leq C t_1^{-2},$$
for $t_1 \leq t_2 \leq \left(1.1\right)t_1$. Therefore, there exists a time $\tilde{t}\in[t_1,t_2]$ such that 
$$\int_{r_1^*}^{r_2^*} \left(\phi^2+\left(\partial_{r^*}\phi\right)^2\right) dVol_{\tilde{t}} \leq C \tilde{t}^{-3}.$$
In order to show that the same holds for any $t$, we note that $S$ is strictly timelike on a compact set of $r^*$. Therefore we can integrate in the direction of $S$ from the slice $\tilde{t}$ to a generic slice $t$. This integration would not give an extra factor of $t$ precisely because we already have the estimates for $\psi=S\phi$. After controlling the spacetime terms by (\ref{Xest}) and (\ref{Xphimiddle}), we show that for any $t$, 
$$\int_{r_1^*}^{r_2^*} \left(\phi^2+\left(\partial_{r^*}\phi\right)^2\right) dVol_{t} \leq C_\delta t^{-3+\delta}.$$
We use Sobolev Embedding to get the pointwise decay estimate for $\phi$ and its derivatives (for $r_1^*\leq r^* \leq r_2^*$) after commuting with an appropriate number of Killing vector fields. We note in particular that in this proof, it is unnecessary to invert the Laplacian on Schwarzschild spacetime to prove the decay of $\phi$.

The argument above gives the decay of $\phi$ and its derivatives in a compact region of $r^*$, i.e. a compact region of space that is also away from the event horizon. (Recall $r^*$ is defined so that $r^*=-\infty$ at the event horizon.) In order to prove that $\phi$ also decays along the event horizon, we use the red-shift vector field introduced in \cite{DRS}. This vector field was used in \cite{DRS} to show that in some (explicitly identified) neighborhood of the event horizon, some energy quantity on an initial slice can control some similar energy quantity in a spacetime slab provided that the error terms that are supported in a compact region of $r^*$ can be controlled. It is then used to propagate the decay of $\phi$ from a compact region of $r^*$ to the event horizon. In this article, we show along these lines that any decay estimate proved on a compact region of $r^*$ can be propagated to the event horizon, giving rise to a decay estimate of the same rate. This will be carried out in Section 6 and will give the full improved decay result.

\section{Notations}
Before proceeding, we would like to first define the notations used for the coordinates and volume form.\\ \\
For the $r, r^*$ coordinates, we always use $^*$ to denote the Regge-Wheeler tortoise coordinate of the same point.\\ \\
For the $t$ coordinates:\\
$t_0$ denotes the time slice on which the initial data is posed.\\
$t_*$ denotes the time slice on which we would like to control the solution.\\
$t_i$ denotes dyadic time slices (which will be defined in Section 4).\\
$t$ denotes a generic time slice.\\
We assume $t_0, t_*, t_i, t >0$. \\ \\
For volume forms:\\
$dVol$ denotes the spacetime volume form, $dVol=r^2\left(1-\mu \right) dA dr^* dt$.\\
$dVol_{t}$ denotes the volume form on a time slice, $dVol_{t}=r^2\sqrt{1-\mu } dA dr^*$.\\
$dVol_{v}$ denotes a volume form on a $v$ slice, $dVol_{v}=r^2\sqrt{1-\mu } dA du$.\footnote{
Most of the time it is clear from context whether we are integrating over a $t$ or $v$ slice. We will specify in the case of possible ambiguity, for example $dVol_{\{t=2v-r^*\}}$ is the volume form on a $t$ slice, where $t$ has the specified value.}\\
$dA$ denotes the volume form on the standard sphere of radius 1.\\
Whenever we write $\int$ without integration limits, it denotes the integration over ``whole space'' that is appropriate for the volume form.
\section{Vector Fields}
\subsection{Conservation Laws}
We consider the conservation laws for $\phi$ satisfying $\Box_{g}\phi=0$.
Define the energy-momentum tensor $$T_{\mu\nu}=\partial_{\mu}\phi\partial_{\nu}\phi-\frac{1}{2}g_{\mu\nu}\partial^\alpha\phi\partial_\alpha\phi.$$
We note that $T_{\mu\nu}$ is symmetric and the wave equation implies that $$\nabla^{\mu}T_{\mu\nu}=0.$$
Given a vector field $V^\mu$, we define the associated currents 
$$J^V_{\mu}\left(\phi\right)=V^{\nu}T_{\mu\nu}\left(\phi\right),$$
$$K^V\left(\phi\right)=\pi^V_{\mu\nu}T^{\mu\nu}\left(\phi\right),$$
where $\pi^V_{\mu\nu}$ is the deformation tensor defined by
$$\pi^V_{\mu\nu}=\frac{1}{2}\left(\nabla_{\mu}V_{\nu}+\nabla_{\nu}V_{\mu}\right).$$
In particular, $K^V\left(\phi\right)=\pi^V_{\mu\nu}=0$ if $V$ is Killing.
Since the energy-momentum tensor is divergence-free, 
$$\nabla^{\mu}J^{V}_{\mu}\left(\phi\right)=K^V\left(\phi\right).$$
We also define the modified current $$J^{V,w}_{\mu}\left(\phi\right)=J^{V}_{\mu}\left(\phi\right)+\frac{1}{8}\left(w\partial_{\mu}\phi^2-\partial_{\mu}w\phi^2\right).$$
Define $K^{V,w}\left(\phi\right)=K^V\left(\phi\right)+\frac{1}{4}w\partial^{\nu}\phi\partial_{\nu}\phi-\frac{1}{8}\Box_g w\phi^2$.\\
Then $$\nabla^{\mu}J^{V,w}_{\mu}\left(\phi\right)=K^{V,w}\left(\phi\right).$$
We integrate by parts with this in a region $\mathcal{B}$ bounded to the future by $\Sigma_1$ and to the past by $\Sigma_0$. The region $\mathcal{B}$ should have no other boundary. Denoting the future-directed normal to $\Sigma_0$ and $\Sigma_1$ by $n_{\Sigma_0}^{\mu}$ and $n_{\Sigma_1}^{\mu}$ respectively, we have
\begin{proposition}
$$\int_{\Sigma_1} J^{V}_{\mu}\left(\phi\right)n_{\Sigma_1}^{\mu}dVol_{\Sigma_1}+\int_\mathcal{B} K^V\left(\phi\right) dVol=\int_{\Sigma_0} J^{V}_{\mu}\left(\phi\right)n_{\Sigma_0}^{\mu} dVol_{\Sigma_0}.$$
$$\int_{\Sigma_1} J^{V, w}_{\mu}\left(\phi\right)n_{\Sigma_1}^{\mu}dVol_{\Sigma_1}+\int_\mathcal{B} K^{V,w}\left(\phi\right)dVol=\int_{\Sigma_0} J^{V, w}_{\mu}\left(\phi\right)n_{\Sigma_0}^{\mu}dVol_{\Sigma_0}.$$
\end{proposition}
In this paper, there are two choices of $\Sigma_i$ that we will use. The first is to choose $\Sigma_i$ to be $t=constant$ slices. The second choice is for estimates near the event horizon. In this case, $\mathcal{B}=\{v_0\leq v \leq v_1, t \geq t_0\}$, $\Sigma_0=\{v=v_0, t \geq t_0\} \cup \{v_0\leq v \leq v_1, t = t_0\}$ and $\Sigma_1=\{v=v_1, t \geq t_0\} \cup \{v_0\leq v \leq v_1, u = \infty\}$ (See Figure).\\
\begin{figure}[htbp]
\begin{center}
 
\input{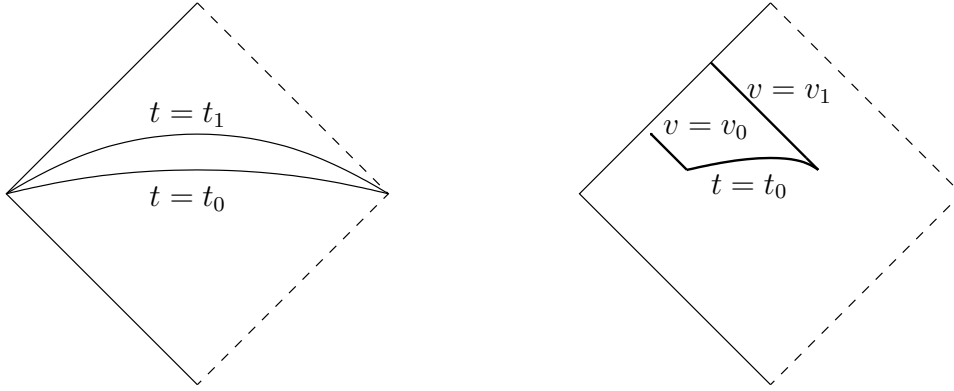}
 
\caption{Regions of integration}
\end{center}
\end{figure}

One can similarly define the above quantities for the inhomogeneous wave equation $\Box_{g}\psi=F$. In this case, the energy-momentum is no longer divergence free. Instead, we have
$$\nabla^{\mu}T_{\mu\nu}=F\partial_\nu\psi.$$
In this case,
$$\nabla^{\mu}J^{V}_{\mu}\left(\psi\right)=K^V\left(\psi\right)+FV^{\nu}\partial_{\nu}\psi.$$
For the modified current, $$\nabla^{\mu}J^{V,w}_{\mu}\left(\psi\right)=K^{V,w}\left(\psi\right)-\frac{1}{4}Fw\psi+FV^{\nu}\partial_{\nu}\psi.$$
\begin{proposition}
$$\int_{\Sigma_1} J^{V}_{\mu}\left(\psi\right)n_{\Sigma_1}^{\mu}dVol_{\Sigma_1}+\int_B K^V\left(\psi\right)dVol+\int_B FV^{\nu}\partial_{\nu}\psi=\int_{\Sigma_0} J^{V}_{\mu}\left(\psi\right)n_{\Sigma_0}^{\mu}dVol_{\Sigma_0}.$$
$$\int_{\Sigma_1} J^{V,w}_{\mu}\left(\psi\right)n_{\Sigma_1}^{\mu}dVol_{\Sigma_1}+\int_B K^{V,w}\left(\psi\right)dVol+\int_B \left(-\frac{1}{4}Fw\psi+FV^{\nu}\partial_{\nu}\psi\right)dVol=\int_{\Sigma_0} J^{V,w}_{\mu}\left(\psi\right)n_{\Sigma_0}^{\mu}dVol_{\Sigma_0}.$$
\end{proposition}
In the case of wave equation on Schwarzschild background, we can compute the energy momentum tensor explicitly in local coordinates $\left(t,r^*,x^A,x^B\right)$ or equivalently $\left(u,v,x^A,x^B\right)$, where $x^A,x^B$ is an orthonormal basis on $\mathbb S^2$.
$$T_{uu}\left(\phi\right)=\left(\partial_u\phi\right)^2,$$
$$T_{vv}\left(\phi\right)=\left(\partial_v\phi\right)^2,$$
$$T_{uv}\left(\phi\right)=\left(1-\mu \right)|\nabb\phi|^2,$$
$$T_{AA}\left(\phi\right)+T_{BB}\left(\phi\right)=|\nabb\phi|^2-\partial^{\alpha}\phi\partial_{\alpha}\phi.$$
As a result,
\begin{equation*}
\begin{split}
K^V\left(\phi\right)=&\frac{1}{4\left(1-\mu \right)}\left(\left(\partial_u\phi \right)^2\partial_v \left(V_v\left(1-\mu \right)^{-1}\right)+\left(\partial_v\phi \right)^2\partial_u \left(V_u\left(1-\mu \right)^{-1}\right)\right.\\
&\left.+|\nabb\phi|^2\left(\partial_uV_v+\partial_vV_u\right) \right)-\frac{1}{2r}\left(V_u-V_v\right)\left(|\nabb\phi|^2-\partial^{\alpha}\phi\partial_{\alpha}\phi \right).
\end{split}
\end{equation*}
\subsection{Vector Field Multiplier $T$}
Define $T=\partial_t$.
Recall that $T$ is Killing.
Therefore,
$$K^T\left(\phi\right)=0.$$
In the following, we will consider this current on a constant t-slice.\\
One computes that in local coordinates 
$$J^T_{\mu}\left(\phi\right) n_t^{\mu}=\frac{1}{2\sqrt{1-\mu}}\left(\left(\partial_t \phi\right)^2+\left(\partial_{r^*} \phi\right)^2+\left(1-\mu \right)|\nabb\phi|^2\right),$$
where $n_t^{\mu}$ is the normal to a $t$-slice.\\
\subsection{Vector Field Multiplier $X$}
Define $X=f\left(r^*\right)\partial_{r^*}$. In the following we will use different functions $f$.
One computes that $$K^X\left(\phi\right)=\frac{f'\left(r^*\right)\left(\partial_{r^*}\phi\right)^2}{1-\mu}+\frac{1}{2}|\nabb\phi|^2\left(\frac{2-3\mu}{r}\right)f\left(r^*\right)-\frac{1}{4}\left(2f'\left(r^*\right)+\frac{4\left(1-\mu \right)}{r}f\left(r^*\right)\right)\partial^{\alpha}\phi\partial_{\alpha}\phi.$$
We consider the modified current using $w^X=2f'\left(r^*\right)+\frac{4\left(1-\mu \right)}{r}f\left(r^*\right)$.
Then 
\begin{equation}\label{Xid}
\begin{split}
K^{X,w^X}=&\frac{f'\left(r^*\right) \left( \partial_{r^*}\phi \right)^2}{1-\mu }+\frac{1}{2}|\nabb\phi|^2 \left(\frac{2-3\mu }{r}\right)f\left(r^*\right)-\frac{1}{8}\Box_g w^X \phi^2\\
=&\frac{f'\left(r^*\right)\left(\partial_{r^*}\phi \right)^2}{1-\mu}+\frac{1}{2}|\nabb\phi|^2 \left(\frac{2-3\mu}{r} \right)f\left(r^*\right)\\
&-\frac{1}{4}\left(\frac{1}{1-\mu}f'''\left(r^*\right)+\frac{4}{r}f''\left(r^*\right)+\frac{\mu}{r^2}f'\left(r^*\right)-\frac{2\mu}{r^3}\left(3-4\mu \right) f\left(r^*\right)\right)\phi^2,\\
J^{X,w^X}_{\mu} n_t^{\mu}=&\frac{1}{\sqrt{1-\mu}}f\left(r^*\right)\partial_t\phi\partial_{r^*}\phi+\frac{1}{2 \sqrt{1-\mu}}\left(f'\left(r^*\right)+\frac{2\left(1-\mu \right)}{r}f\left(r^*\right)\right)\left(\partial_t\phi\right)\phi,
\end{split}
\end{equation}
where $n_t^{\mu}$ is the normal to a $t$-slice.\\
The vector field $X$ is constructed to control a spacetime integral by the boundary terms, i.e., one hopes to control the integral of $K^{X,w^X}\left(\phi\right)$ by the integral of $J^{X,w^X}_{\mu}\left(\phi\right)n_t^{\mu}$. In order for this to be useful, we need $K^X\left(\phi\right)$ to be everywhere positive. Such vector fields are constructed in \cite{DRS} using spherical harmonic decomposition.
In particular, it was shown in \cite{DRS} that there exists a family of vector fields $X_l=f_l\left(r^*\right)\partial_{r^*}$ for $l \ge 0$ such that for any function $\phi$ (not necessarily satisfying the wave equation), if we write out the spherical harmonic decomposition $\phi=\displaystyle\sum_{l=0}^{\infty} \phi_l$, $K^{X_l, w^{X_l}}\left(\phi_l\right)\geq 0$.\\
Moreover, one has 
\begin{equation} \label{Xl}
\int_{\mathbb S^2}\left(\frac{\left(\partial_{r^*}\phi_l\right)^2}{\left(1+|r^*|\right)^2\left(1-\mu \right)}+\frac{\phi_l^2}{\left(1+|r^*|\right)^4 \left(1-\mu \right)}\right)dA \leq C \int_{\mathbb S^2} K^{X_l, w^{X_l}}\left(\phi_l\right) dA
\end{equation}
for $l \ge 1$, where $C$ can be picked to be independent of $l$, and 
\begin{equation*}
\int_{\mathbb S^2} \frac{\left(\partial_{r^*}\phi_0\right)^2}{\left(1+|r^*|\right)^{1+\delta}r^2\left(1-\mu \right)}dA\leq C \int_{\mathbb S^2} K^{X_l, w^{X_l}}\left(\phi_0\right) dA.
\end{equation*}
Moreover for this choice of $X_l$, the boundary terms are also controllable as shown in \cite{DRS}:
\begin{equation*}
\int K^{X_l, w^{X_l}}\left(\phi_l\right) dVol \leq C \int J^{T}_{\mu}\left(\phi\right)n_{t_0}^{\mu}dVol_{t_0}.
\end{equation*}
\begin{remark}
We note that although $K^{X_l,w^{X_l}}\left(\phi\right)$ is shown to be nonnegative everywhere, it has a weight in front of $|\nabb\phi|^2$ that degenerates at $r=3M$. Therefore, we cannot directly estimate the integral of $|\nabb\phi|^2$ by that of $K^{X_l,w^{X_l}}\left(\phi\right)$. Instead, we will consider the equation $\Box_g\left(\Omega\phi\right)=0$ and estimate the relevant quantities with $\int K^{X_l,w^{X_l}}\left(\Omega\phi\right) dVol$. This loss of derivative is related to the trapping phenomenon that we mentioned in Section 1.1.
\end{remark}
In section 3 we will construct two more vector fields of this form. One will be a modified $X_0$ to control a weighted $L^2$-norm of the zeroth spherical harmonic and the other will be used to control the behavior at infinity.
\subsection{Vector Field Multiplier $Z$}
Define $Z=u^2\partial_u+v^2\partial_v$.
This is the analogue of the conformal vector field in Minkowski spacetime. Like the case in Minkowski spacetime, it is used to show decay for the solution to the wave equation.
One computes that 
$$K^Z=-t|\nabb\phi|^2 \left(\frac{1}{2}+\frac{\mu r^*}{4r}-\frac{r^*\left(1-\mu \right)}{2r}\right)-\frac{1}{4}\frac{2tr^*\left(1-\mu \right)}{r}\partial^{\alpha}\phi\partial_{\alpha}\phi.$$
We consider the modified current using $w^Z=\frac{2tr^*\left(1-\mu \right)}{r}$.
Then
\begin{equation*}
\begin{split}
K^{Z,w^Z}=&-t|\nabb\phi|^2 \left(\frac{1}{2}+\frac{\mu r^*}{4r}-\frac{r^*\left(1-\mu \right)}{2r}\right)-\frac{1}{8}\Box_g w^Z \phi^2\\
=&-t|\nabb\phi|^2 \left(\frac{1}{2}+\frac{\mu r^*}{4r}-\frac{r^*\left(1-\mu \right)}{2r}\right)-\frac{t}{4}\mu r^{-2}\phi^2 \left(2+\frac{r^* \left(4\mu -3\right)}{r}\right),\\
J^{Z,w^Z}_{\mu} n_t^{\mu}=&\frac{1}{4\sqrt{1-\mu }}\left(u^2\left(\partial_u\phi \right)^2+v^2\left(\partial_v\phi \right)^2+\left(1-\mu \right)\left(u^2+v^2\right)|\nabb\phi|^2+\frac{2tr^*\left(1-\mu \right)}{r}\phi\partial_t\phi-\frac{r^*\left(1-\mu \right)}{r}\phi^2 \right).
\end{split}
\end{equation*}
where $n_t^{\mu}$ is the normal to a $t$-slice.\\
It is shown in \cite{DRS} that there exist $r_1^*, r_2^*$ such that for $r^* \leq r_1^*$ or $r^* \geq r_2^*$, $K^{Z,w^Z}\geq 0$.\\
Moreover, it is shown that $\int J^{Z,w^Z}_{\mu} n_t^{\mu} dVol_{t}$ is everywhere non-negative.\\
More specifically, if we define $S=u\partial_u+v\partial_v$ and $\underline{S}=-u\partial_u+v\partial_v$,
\begin{equation*}
\begin{split}
&\int J^{Z,w^Z}_{\mu} n_t^{\mu} dVol_t\\
=&\int \frac{1}{8\sqrt{1-\mu}}\left(\mu\left(\left(S\phi\right)^2+\left(\underline{S}\phi\right)^2\right)+\left(1-\mu \right)\left(\left(S\phi+\frac{r^*}{r}\phi\right)^2+\left(\underline{S}\phi+\frac{t}{r}\phi\right)^2\right)+2\left(1-\mu \right)\left(u^2+v^2\right)|\nabb\phi|^2\right) dVol_t.
\end{split}
\end{equation*}
\subsection{Vector Field Multipliers $Y'$, $Y$ and $N$}
Define $Y'=\frac{y_1\left(r^*\right)}{1-\mu }\partial_u+y_2\left(r^*\right)\partial_v$, where $y_1, y_2 > 0$ are supported in $r \leq \left(1.2\right)r_0$, with $y_1=1$, $y_2=0$ at the event horizon and $y_1'\left(r^*\right)\sim y_2\left(r^*\right)\sim C\left(1+|r^*|\right)^{-1-\delta}$ for $2M\leq r \leq r_0$.\\
Here we want to choose $r_0$ small enough so that \\
1. $Y'$ is supported on $r<3M$ (i.e. $\left(1.2\right)r_0< 3M$),\\
2. $K^{Y'}\left(\phi\right)\geq 0$ on $2M\leq r \leq r_0$,\\
3. $C K^{Y'}\left(\phi\right) \geq \frac{1}{\sqrt{1-\mu }}J^{Y'}_{\mu}\left(\phi\right) n^{\mu}_{\{v=const.\}}$ on $2M\leq r \leq r_0$.\\
The vector field $Y'$ is designed to capture the red-shift effect at the event horizon \cite{DRS}. Using the current $J^{Y'}$, we will not only produce estimates on constant $t$-slices, but also on constant $v$- slices. We will therefore record here all the relevant computations.\\
We have
$$K^{Y'}=\frac{\left(\partial_u\phi\right)^2}{2\left(1-\mu \right)^2}\left(\frac{y_1\mu }{r}-y_1' \right)+\frac{\left(\partial_v\phi\right)^2}{2\left(1-\mu \right)}y_2'+\frac{1}{2}|\nabb\phi|^2\left(\frac{y_1'}{1-\mu }-\frac{\left(y_2\left(1-\mu \right)\right)'}{1-\mu }\right)-\frac{1}{r}\left(\frac{y_1}{1-\mu }-y_2\right)\partial_u\phi\partial_v\phi,$$
$$J^{Y'}_{\mu}\left(\phi\right) n^{\mu}_{\{v=const.\}}=\frac{1}{2\sqrt{1-\mu }}\left(\frac{y_1}{1-\mu }\left(\partial_u\phi\right)^2+\left(1-\mu \right)y_2|\nabb\phi|^2\right),$$
$$J^{Y'}_{\mu}\left(\phi\right) n^{\mu}_{\{t=const.\}}=\frac{1}{2\sqrt{1-\mu }}\left(\frac{y_1}{1-\mu }\left(\partial_u\phi\right)^2 + y_2\left(\partial_v\phi\right)^2 + \left(y_1+\left(1-\mu \right)y_2\right)|\nabb\phi|^2\right).$$
From this we see that if $r_0$ is chosen to be close enough to $2M$, requirements 2 and 3 can be satisfied.\\
We modify this vector field so that it has better bounds on constant $t$- slices.\\
Define $Y=Y'+\chi\left(r\right)T$, where $\chi\left(r\right)$ is a cutoff function with $\chi\left(r\right)=\left\{\begin{array}{clcr}1&r\le r_0\\0&r\ge \left(1.2\right)r_0\end{array}\right.$. $Y$ has the following properties:\\
1. $Y$ is supported on $r<\left(1.2\right)r_0$,\\
2. $K^Y=K^{Y'}$ on $r<r_0$,\\
3. $C K^{Y}\left(\phi\right) \geq \frac{1}{\sqrt{1-\mu }}J^{Y}_{\mu}\left(\phi\right) n^{\mu}_{\{v=const.\}}$ on $2M\leq r \leq r_0$.\\
On the region $2M \leq r \leq r_0$, we have
$$J^{Y}_{\mu}\left(\phi\right) n^{\mu}_{\{v=const.\}}=\frac{1}{2\sqrt{1-\mu }}\left(\left(\frac{y_1}{1-\mu }+\frac{1}{2}\right)\left(\partial_u\phi\right)^2+\left(1-\mu \right)\left(y_2+1\right)|\nabb\phi|^2\right),$$
$$J^{Y}_{\mu}\left(\phi\right) n^{\mu}_{\{t=const.\}}=\frac{1}{2\sqrt{1-\mu }}(\left(\frac{y_1}{1-\mu }+\frac{1}{2}\right)\left(\partial_u\phi\right)^2 + \left(y_2+\frac{1}{2}\right)\left(\partial_v\phi\right)^2 + \left(y_1+\left(1-\mu \right)\left(y_2+1)\right)|\nabb\phi|^2\right).$$
We argue without computation that for $r_0 \leq r \leq \left(1.2\right)r_0$,
$$|K^Y| \leq C \frac{1}{\sqrt{1-\mu }}J^T_{\mu }\left(\phi\right) n^{\mu}_{\{t=const.\}}, $$
$$J^Y_{\mu}\left(\phi\right) n^{\mu}_{\{v=const.\}}\leq C J^T_{\mu}\left(\phi\right) n^{\mu}_{\{v=const.\}},$$
$$J^Y_{\mu }\left(\phi\right) n^{\mu}_{\{t=const.\}} \leq C J^T_{\mu }\left(\phi\right) n^{\mu}_{\{t=const.\}}. $$
This is true because $J^T_{\mu }\left(\phi\right) n^{\mu}_{\{t=const.\}}$ controls every derivative of $\phi$ while the terms in $J^T_{\mu}\left(\phi\right) n^{\mu}_{\{v=const.\}}$ and $J^{Y'}_{\mu }\left(\phi\right) n^{\mu}_{\{t=const.\}}$ contain only derivatives $\partial_u$ and $\nabb$. Thus the only difference is the weights, which are functions of $r$ and are harmless since $r$ is bounded on this region.\\
Define $N=T+Y$. $N$ is clearly causal, thus $J^N_{\mu}\left(\phi\right) n^{\mu}_{\{t=const.\}}\geq 0$. Away from the horizon, namely when $r \geq 1.2r_0$, $J^N_{\mu}\left(\phi\right) n^{\mu}_{\{t=const.\}}=J^T_{\mu}\left(\phi\right) n^{\mu}_{\{t=const.\}}$. However, as we approach the horizon, $J^N_{\mu}\left(\phi\right) n^{\mu}_{\{t=const.\}} \sim J^Y_{\mu}\left(\phi\right) n^{\mu}_{\{t=const.\}}$ and thus $J^N_{\mu}\left(\phi\right) n^{\mu}_{\{t=const.\}}$ gives a much stronger bound. We assume for our energy classes that the integral of $J^N_{\mu}\left(\phi\right) n^{\mu}_{\{t=const.\}}$ is bounded initially and this clearly implies the boundedness for the corresponding integrals for $J^T$ and $J^Y$ initially. The flux corresponding to $J^N$ should be thought of as a non-degenerate energy, which does not degenerate at the event horizon. This allows us to prove decay results along the event horizon.

Before introducing the vector field commutator $S$, we end this part on vector field multipliers by explicitly noting what each of the positive quantities bounds. Most of these are direct consequences of the expressions of the currents, except that for $J^{Z,w^Z}$, which requires some manipulation and is proved in \cite{DRS}.
\begin{proposition}
\begin{enumerate}
\item $\frac{1}{\sqrt{1-\mu }}\left(\left(\partial_{r^*}\phi \right)^2+\left(\partial_t\phi \right)^2+ \left(1-\mu \right)|\nabb\phi|^2\right) \leq C J^T_{\mu}\left(\phi\right) n^{\mu}_t$,
\item $\int_{-\infty}^{\infty}\int_{\mathbb S^2}\frac{1}{\sqrt{1-\mu }}\left(u^2\left(\partial_u\phi \right)^2+v^2\left(\partial_v\phi \right)^2 +\left(1-\mu \right)\left(u^2+v^2\right)|\nabb\phi|^2\right)dVol_t \leq C \int_{-\infty}^{\infty}\int_{\mathbb S^2}J^{Z,w^Z}_{\mu}\left(\phi\right) n^{\mu}_t dVol_t$,
\item $\int_{-\infty}^{\infty}\int_{\mathbb S^2}\frac{1}{\sqrt{1-\mu }}\left(\left(1-\mu \right)\frac{\left(r^*\right)^2+t^2}{r^2} \phi^2\right) dVol_t\leq C \int_{-\infty}^{\infty}\int_{\mathbb S^2} J^{Z,w^Z}_{\mu}\left(\phi\right) n^{\mu}_t dVol_t$,
\item $\int_{\mathbb S^2}\frac{\left(\partial_{r^*}\phi\right)^2}{\left(1+|r^*|\right)^2 r^{1+\delta}\left(1-\mu \right)} dA \leq C \int_{\mathbb S^2}\displaystyle\sum_l K^{X_l,w^{X_l}}\left(\phi\right)dA$,
\item $\int_{\mathbb S^2}\frac{|\nabb\phi|^2}{\left(1+|r^*|\right)^4 \left(1-\mu \right)} dA\leq C \int_{\mathbb S^2}\displaystyle\sum_l K^{X_l,w^{X_l}}\left(\Omega\phi\right)dA$.
\end{enumerate}
\end{proposition}
\subsection{Vector Field Commutator $S$}
Define $S=t \partial_t+r^* \partial_{r^*}=v\partial_v+u\partial_u.$\\
This vector field, together with the usual Killing fields, will be commuted with $\Box_g$. We note that the vector field $t \partial_t+r \partial_r$ is conformally Killing on Minkowski with $[\Box_m,t \partial_t+r \partial_r]=2\Box_m$. Therefore, the commutator $[\Box_g,S]$ is expected to approach $2\Box_g$ towards spatial infinity, where the spacetime approaches Minkowski. \\
We set $\psi=S\phi$ and derive an equation for $\psi$.
\begin{proposition}
\begin{enumerate}
\item $[\Box_g,S]=\left(2+\frac{r^*\mu}{r}\right)\Box_g +\frac{2}{r}\left(\frac{r^*}{r}-1-\frac{2r^*\mu}{r}\right)\partial_{r^*}+2\left(\left(\frac{r^*}{r}-1\right)-\frac{3r^*\mu}{2r}\right)\lapp\phi$.
\item $\Box_g \psi=g_1\left(r^*\right)\partial_{r^*}\phi+g_2\left(r^*\right)\lapp\phi$, where $|g_1\left(r^*\right)|,\frac{|g_2\left(r^*\right)|}{r}\sim\left\{\begin{array}{clcr}\frac{\left(\log r\right)_+}{r^2}& r >>2M\\|r^*|& r\sim 2M \end{array}\right.$,\\ $\left(\log r\right)_+=\max \{ \log r, 1 \}$.
\end{enumerate}
\end{proposition}
\begin{remark}
Equivalently, we write $|g_1\left(r^*\right)|,\frac{|g_2\left(r^*\right)|}{r}\sim \frac{\left(1+|r^*|\right)\left(\log r\right)_+}{r^3}.$
\end{remark}
\begin{proof}
\begin{equation*}
\begin{split}
[-\left(1-\mu \right)^{-1}\partial_t^2,S ]=&-2\left(1-\mu \right)^{-1}\partial_t^2+r^*\partial_{r^*}\left(1-\mu \right)^{-1}\partial_t^2\\
=&-2\left(1-\mu \right)^{-1}\partial_t^2-\frac{r^*\mu}{r\left(1-\mu \right)}\partial_t^2,
\end{split}
\end{equation*}
\begin{equation*}
\begin{split}
[\left(1-\mu \right)^{-1}\partial_{r^*}^2,S ]=&2\left(1-\mu \right)^{-1}\partial_{r^*}^2-r^*\partial_{r^*}\left(1-\mu \right)^{-1}\partial_{r^*}^2\\
=&2\left(1-\mu \right)^{-1}\partial_{r^*}^2+\frac{r^*\mu}{r\left(1-\mu \right)}\partial_{r^*}^2,
\end{split}
\end{equation*}
\begin{equation*}
\begin{split}
[\frac{2}{r}\partial_{r^*},S ]=&\frac{2}{r}\partial_{r^*}+\frac{2r^*\left(1-\mu \right)}{r^2}\partial_{r^*}\\
=&\frac{4}{r}\partial_{r^*}+\left(\frac{2r^*\left(1-\mu \right)}{r^2}-\frac{2}{r}\right)\partial_{r^*},
\end{split}
\end{equation*}
\begin{equation*}
\begin{split}
[\lapp,S ]=&\frac{2r^*\left(1-\mu \right)}{r}\lapp,
\end{split}
\end{equation*}
\begin{equation*}
\begin{split}
[\Box_g, S]=&\left(2+\frac{r^*\mu}{r}\right)\Box_g+\left(\frac{2r^*\left(1-\mu \right)}{r^2}-\frac{2}{r}-\frac{2r^*\mu}{r^2}\right)\partial_{r^*}+\left(\frac{2r^*\left(1-\mu \right)}{r}-2-\frac{r^*\mu}{r}\right)\lapp\\
=&\left(2+\frac{r^*\mu }{r}\right)\Box_g+\frac{2}{r}\left(\frac{r^*}{r}-1-\frac{2r^*\mu}{r}\right)\partial_{r^*}+2\left(\left(\frac{r^*}{r}-1\right)-\frac{3r^*\mu}{2r}\right)\lapp\phi.
\end{split}
\end{equation*}
2. is immediate from 1. if we let $g_1\left(r^*\right)=\frac{2}{r}\left(\frac{r^*}{r}-1-\frac{2r^*\mu}{r}\right)$ and $g_2\left(r^*\right)=2\left(\left(\frac{r^*}{r}-1\right)-\frac{3r^*\mu}{2r}\right)$.
\end{proof}
\section{Estimates for $\phi$}
The following has been proved in \cite{DRS} and is collected for later use.
\begin{theorem}[Dafermos-Rodnianski]
\begin{enumerate}
\item $\int J^T_\mu\left(\phi\right) n^{\mu}_{t_*} dVol_{t_*}=\int J^T_\mu\left(\phi\right) n^{\mu}_{t_0} dVol_{t_0}$,
\item $\int \displaystyle\sum_l K^{X_l,w^{X_l}} \left(\phi\right) dVol \leq C\int J^T_\mu\left(\phi\right) n^{\mu}_{t_0} dVol_{t_0}$,
\item $\int J^Z_\mu\left(\phi\right) n^{\mu}_{t_*} dVol_{t_*} \leq C E_0\left(\phi\right)$,
\item $\int_{-\frac{t_*}{2}}^{\frac{t_*}{2}} J^T_\mu\left(\phi\right) n^{\mu}_{t_*} dVol_{t_*} \leq C E_0\left(\phi\right) t_*^{-2}$,
\item $\int_{t_1}^{t_2} \int_{-\frac{t}{2}}^{\frac{t}{2}} \displaystyle\sum_l K^{X_l,w^{X_l}} \left(\phi\right) dVol \leq C E_0\left(\phi\right) t_1^{-2},$ where $t_1 \leq t_2 \leq \left(1.1\right)t_1$.
\end{enumerate}
\end{theorem}
The following Hardy type inequality is also proved in \cite{DRS} and will be used throughout this paper.
\begin{lemma}
$$\int \left(1+|r^*|\right)^{-2}\phi^2\left(1-\mu \right)^{-\frac{1}{2}} dVol_{t_0}\leq C \int \left(\partial_{r^*}\phi\right)^2 \left(1-\mu \right)^{-\frac{1}{2}} dVol_{t_0}.$$
\end{lemma}
\begin{remark}
This can be written equivalently in local coordinates as
$$\int_{-\infty}^{\infty} \int_{\mathbb S^2} \frac{\phi^2}{\left(1+|r^*|\right)^{2}} r^2 dA dr^*\leq C \int_{-\infty}^{\infty} \int_{\mathbb S^2} \left(\partial_{r^*}\phi\right)^2 r^2 dA dr^*.$$
\end{remark}
We construct a vector field $X_0$ to control the spacetime integral of $\phi^2$ itself.
\begin{proposition}
$$\int \frac{\phi^2}{\left(1+|r^*|\right)^4} dVol\leq C\int J^{T}_{\mu}\left(\phi\right)n_{t_0}^{\mu}dVol_{t_0}.$$
\end{proposition}
\begin{proof}
We first notice that we already have control of a weighted $L^2$-norm of the non-zeroth spherical harmonics. This is because by (\ref{Xl}),
$$\int_{\mathbb S^2} \frac{\phi_l^2}{\left(1+|r^*|\right)^4 \left(1-\mu \right)} dA\leq C \int_{\mathbb S^2} K^{X_l, w^{X_l}}\left(\phi_l\right) dA$$ for $l \ge 1$. This together with Theorem 5.2 would give
$$\int \frac{\phi_l^2}{\left(1+|r^*|\right)^4 \left(1-\mu \right)} dVol\leq C \int J^T_\mu\left(\phi\right) n^{\mu}_{t_0} dVol_{t_0}$$ for $l \ge 1$.
So it suffices to consider the zeroth spherical harmonic.\\
Define $X_0=f_0 \partial_{r^*}$, with $f_0\left(r^*\right)=f_0\left(r\right)=-\frac{M^3}{\left(1+4\mu^{-2}\right)}=-\frac{\mu^3r^3}{8\left(1+4\mu^{-2}\right)}$.\\
Suppose we act with $X_0$ on the zeroth spherical harmonic of $\phi_0$.\\
Using (\ref{Xid}),
\begin{equation*}
\begin{split}
K^{X_0,w^{X_0}}\left(\phi_0\right)=&\frac{f_0'\left(r^*\right)\left(\partial_{r^*}\phi_0 \right)^2}{1-\mu}+\frac{1}{2}|\nabb\phi_0|^2 \left(\frac{2-3\mu}{r} \right)f_0\left(r^*\right)\\
&-\frac{1}{4}\left(\frac{1}{1-\mu}f_0'''\left(r^*\right)+\frac{4}{r}f_0''\left(r^*\right)+\frac{\mu}{r^2}f_0'\left(r^*\right)-\frac{2\mu}{r^3}\left(3-4\mu \right) f_0\left(r^*\right)\right)\phi_0^2\\
=&\frac{f_0'\left(r^*\right)\left(\partial_{r^*}\phi_0 \right)^2}{1-\mu}-\frac{1}{4}\left(\frac{1}{1-\mu}f_0'''\left(r^*\right)+\frac{4}{r}f_0''\left(r^*\right)+\frac{\mu}{r^2}f_0'\left(r^*\right)-\frac{2\mu}{r^3}\left(3-4\mu \right) f_0\left(r^*\right)\right)\phi_0^2,\\
J^{X_0,w^{X_0}}_{\mu}\left(\phi_0\right) n_t^{\mu}=&\frac{1}{\sqrt{1-\mu}}f_0\left(r^*\right)\partial_t\phi_0\partial_{r^*}\phi_0+\frac{1}{2 \sqrt{1-\mu}}\left(f_0'\left(r^*\right)+\frac{2\left(1-\mu \right)}{r}f_0\left(r^*\right)\right)\left(\partial_t\phi_0\right)\phi_0,
\end{split}
\end{equation*}
where we have used $\nabb\phi_0=0$.\\
We would have to show firstly that $K^{X_0,w^{X_0}}\left(\phi_0\right) \ge 0$ and controls $\phi^2$, and secondly that $J^{X_0,w^{X_0}}_{\mu}\left(\phi_0\right) n^{\mu}$ is controllable by $J^{T}_{\mu}\left(\phi_0\right) n^{\mu}$. We first compute the derivatives of $f_0$:
\begin{equation*}
\begin{split}
f'_0\left(r^*\right)=&\left(1-\mu \right)\partial_{r}f_0\left(r\right)\\
=&\frac{\mu r^2\left(1-\mu \right)}{\left(1+4\mu^{-2}\right)^2}\geq0\\
f_0''\left(r^*\right)=&\left(1-\mu \right)^2\partial_r^2 f_0+\frac{\mu\left(1-\mu\right)}{r}\partial_r f_0\\
=&\left(1-\mu \right)^2\left(-\frac{16r}{\mu \left(1+4\mu^{-2}\right)^3}+\frac{2\mu r}{\left(1+4\mu^{-2}\right)^2}\right)+\frac{\mu^2 r\left(1-\mu \right)}{\left(1+4\mu^{-2}\right)^2}\\
f_0'''\left(r^*\right)=&\left(1-\mu \right)^3\partial_r^3 f_0+\frac{3\mu \left(1-\mu \right)^2}{r}\partial_r^2 f_0+\frac{\mu^2 \left(1-\mu \right)}{r^2}\partial_r f_0-\frac{2\mu\left(1-\mu \right)^2}{r^2}\partial_r f_0\\
=&\left(1-\mu\right)^3 \left( \frac{384}{\mu^3 \left( 1+4\mu^{-2} \right)^4}-\frac{48}{\mu \left(1+4\mu^{-2}\right)^3}\right)+\frac{3\mu\left(1-\mu\right)^2}{r}\left(-\frac{16r}{\mu \left(1+4\mu^{-2}\right)^3}+\frac{\mu r}{\left(1+4\mu^{-2}\right)^2}\right)\\
&+\frac{\mu^2 \left(1-\mu\right)\left(3 \mu -2\right)}{\left(1+4\mu^{-2}\right)^2}.
\end{split}
\end{equation*}
A computation shows that 
\begin{equation*}
\begin{split}
&\frac{1}{1-\mu}f_0'''+\frac{4}{r}f_0''+\frac{\mu}{r^2}f_0'-\frac{2\mu}{r^3}\left(3-4\mu\right)f_0\\
=&-\frac{\mu^6\left(192+\mu \left(128+\mu \left(-784+\mu \left(464+\mu \left(-28+\mu \left(52+\mu \left(-3+4\mu \right)\right)\right)\right)\right)\right)\right)}{4\left(4+\mu^{2}\right)^4}.
\end{split}
\end{equation*}
We need to show that $192+\mu \left(128+\mu \left(-784+\mu \left(464+\mu \left(-28+\mu \left(52+\mu \left(-3+4\mu \right)\right)\right)\right)\right)\right)\geq 0$ for $0\leq \mu \leq 1$.
Case 1: $\frac{11}{20} \leq \mu \leq 1$\\
$192+128 \mu - 784 \mu^2 +464\mu^3 =16\left(-12-20 \mu +29 \mu^2 \right)\left(\mu -1\right)\geq 0.$\\
$52-3\mu +4 \mu^2$ reaches its minimum at $\frac{3}{8}$. Hence, $52-3\mu +4 \mu^2 \geq \frac{823}{16}.$\\
$-28+\mu \left(52-3\mu +4 \mu^2\right) \geq-28+\frac{11}{20}\frac{823}{16}\geq 0.$\\
Case 2: $0 \leq \mu \leq \frac{11}{20}$\\
$464-28\mu +\frac{823}{16}\mu^2$ has negative discriminant, hence $\geq 0$.\\
Also, for this range of $\mu$, $192+128 \mu-784\mu^2 \geq 0$.\\
Therefore, $K^{X_0,w^{X_0}}\left(\phi_0\right) \ge 0$.
Moreover, $\phi_0^2 \leq C K^{X_0,w^{X_0}}\left(\phi_0\right)$.\\
It now remains only to control the boundary terms. Using Lemma 6 and Cauchy-Schwarz,
\begin{equation*}
\begin{split}
&\int J^{X_0,w^{X_0}}_{\mu}\left(\phi_0\right) n^{\mu} dVol_{t}\\
=&\int \frac{1}{\sqrt{1-\mu}}f_0\left(r^*\right)\partial_t\phi_0\partial_{r^*}\phi_0+\frac{1}{2 \sqrt{1-\mu}}\left(f_0'\left(r^*\right)+\frac{2\left(1-\mu \right)}{r}f_0\left(r^*\right)\right)\left(\partial_t\phi_0\right)\phi_0 dVol_{t_0}\\
\leq &C \int \frac{1}{\sqrt{1-\mu}}\left(\left(\partial_t\phi_0\right)^2+\left(\partial_{r^*}\phi_0\right)^2+\frac{1}{\left(1+|r^*|\right)^2}\phi_0^2\right) dVol_{t_0}\\
\leq &C \int \frac{1}{\sqrt{1-\mu}}(\left(\partial_t\phi_0\right)^2+\left(\partial_{r^*}\phi_0)^2\right) dVol_{t_0}\\
\leq &C \int J^{T}_{\mu}\left(\phi_0\right) n_{t_0}^{\mu} dVol_{t_0}.
\end{split}
\end{equation*}
\end{proof}
We would like to construct a vector field $\tilde{X}=\tilde{f}\left(r^*\right)\partial_{r^*}$ so as to improve the weights in $r$ of  the spacetime integral that can be controlled. More precisely, we have the following:
\begin{proposition}
$$\int_{t_0}^{t_*}\int_1^{\infty}\int_{\mathbb S^2} \left(r^{-1-\delta}\left(\partial_{r^*}\phi\right)^2+r^{-3-\delta}\phi^2\right)dVol\leq  C\int J^T_\mu \left(\phi\right) n^{\mu}_{t_0} dVol_{t_0},$$ 
$$\int_{t_0}^{t_*}\int_1^{\infty}\int_{\mathbb S^2} r^{-1}|\nabb\phi|^2 dVol\leq  C \sum_{k=0}^1 \int J^T_\mu \left(\Omega^k\phi\right) n^{\mu}_{t_0} dVol_{t_0},$$ 
for $0<\delta<\frac{1}{2}$.
\end{proposition}
\begin{remark}
The loss of derivative above is unnecessary because we are considering only a subregion of $\{r^* > 0 \}$. One can construct yet another variant of the vector field $X$ to achieve the above estimate without any loss of derivatives. However, since this would not improve the regularity in our final result, it is not pursued here.
\end{remark}
\begin{proof}
Let $\tilde{X}=\tilde{f}\left(r^*\right)\partial_{r^*}$, where $\tilde{f}=\chi\left(r^*\right)\left(1-\frac{1}{\left(1+r^*\right)^{\delta}}\right)$ and $\chi$ is a cutoff function satisfying $$\chi=
\left\{\begin{array}{clcr}0&r^*\le 1\\1&r^*\ge \max\{100,100M\}\end{array}\right..$$
We recall (\ref{Xid}):
\begin{equation*}
\begin{split}
K^{\tilde{X},w^{\tilde{X}}}\left(\phi \right)=&\frac{\tilde{f}'\left(r^*\right)\left(\partial_{r^*}\phi \right)^2}{1-\mu}+\frac{1}{2}|\nabb\phi|^2 \left(\frac{2-3\mu}{r} \right)\tilde{f}\left(r^*\right)\\
&-\frac{1}{4}\left(\frac{1}{1-\mu}\tilde{f}'''\left(r^*\right)+\frac{4}{r}\tilde{f}''\left(r^*\right)+\frac{\mu}{r^2}\tilde{f}'\left(r^*\right)-\frac{2\mu}{r^3}\left(3-4\mu \right) \tilde{f}\left(r^*\right)\right)\phi^2,\\
J^{\tilde{X},w^{\tilde{X}}}_{\mu}\left(\phi \right) n_t^{\mu}=&\frac{1}{\sqrt{1-\mu}}\tilde{f}\left(r^*\right)\partial_t\phi\partial_{r^*}\phi+\frac{1}{2 \sqrt{1-\mu}}\left(\tilde{f}'\left(r^*\right)+\frac{2\left(1-\mu \right)}{r}\tilde{f}\left(r^*\right)\right)\left(\partial_t\phi\right)\phi,
\end{split}
\end{equation*}
Since we already have control of the spacetime integrals on a compact set using Theorem 5 and Proposition 7, we only have to show that $K^{\tilde{X},w^{\tilde{X}}}\left(\phi \right) \ge 0 $ for $r^* \geq \max\{100,100M\}$.
For $r^*\geq \max\{100,100M\}$, we have 
$$\tilde{f}'\left(r^*\right)=\frac{\delta}{\left(1+r^*\right)^{1+\delta}}$$
$$\tilde{f}''\left(r^*\right)=-\frac{\delta\left(1+\delta\right)}{\left(1+r^*\right)^{2+\delta}}$$
$$\tilde{f}'''\left(r^*\right)=\frac{\delta\left(1+\delta\right)\left(2+\delta\right)}{\left(1+r^*\right)^{3+\delta}}.$$
Clearly, the coefficient of $\left(\partial_{r^*}\phi \right)^2$ and $|\nabb\phi|^2$ in $K^{\tilde{X},w^{\tilde{X}}}\left(\phi \right)$ is positive for $r^*\geq \max\{100,100M\}$.
We now study the coefficient of $\phi^2$ in $K^{\tilde{X},w^{\tilde{X}}}\left(\phi \right)$ for $r^*\geq \max\{100,100M\}$:
\begin{equation*}
\begin{split}
&\frac{1}{1-\mu}\tilde{f}'''+\frac{4}{r}\tilde{f}''+\frac{\mu}{r^2}\tilde{f}'-\frac{2\mu}{r^3}\left(3-4\mu\right)\tilde{f}\\
=&\frac{1}{1-\mu}\frac{\delta\left(1+\delta\right)\left(2+\delta\right)}{\left(1+r^*\right)^{3+\delta}}-\frac{4\delta\left(1+\delta\right)}{r\left(1+r^*\right)^{2+\delta}}+\frac{2M\delta}{r^2\left(1+r^*\right)^{1+\delta}}\\
&-\frac{12M}{r^3\left(1+r^*\right)^\delta}+\frac{32M^2}{r^5\left(1+r^*\right)^\delta}-\frac{2\mu}{r^3}\left(3-4\mu\right)\\
\leq&\frac{3\delta\left(1+\delta\right)\left(2+\delta\right)}{2\left(1+r^*\right)^{3+\delta}}-\frac{4\delta\left(1+\delta\right)}{r\left(1+r^*\right)^{2+\delta}}+\frac{2M\delta}{r^2\left(1+r^*\right)^{1+\delta}}-\frac{12M}{r^3\left(1+r^*\right)^\delta}+\frac{32M^2}{r^5\left(1+r^*\right)^\delta}\\
\leq&\frac{\delta\left(1+\delta\right)}{r\left(1+r^*\right)^{2+\delta}}\left(\frac{3\delta}{2}-1\right)+\frac{M}{r^3\left(1+r^*\right)^\delta}\left(2\delta-12+\frac{32}{100}\right)\\
<&0.
\end{split}
\end{equation*}
Hence $K^{\tilde{X},w^{\tilde{X}}}\left(\phi \right) \ge 0 $ for $r^*\geq \max\{100,100M\}$.\\ Moreover, on this region of $r^*$, $\left(r^{-1-\delta}\left(\partial_{r^*}\phi\right)^2+r^{-3-\delta}\phi^2 + r^{-1}|\nabb\phi|^2\right) \leq C K^{\tilde{X},w^{\tilde{X}}}\left(\phi \right)$.\\
Finally, we have $\int J^{\tilde{X},w^{\tilde{X}}}_{\mu}\left(\phi \right) n^{\mu} dVol_{t}\leq C J^{T}_{\mu}\left(\phi\right) n^{\mu} dVol_{t}$ using Lemma 6 and Cauchy-Schwarz exactly as in Proposition 7.
\end{proof}
\begin{remark}
The weights in the Proposition are the same as those for Minkowski space. Since Schwarzschild is asymptotically flat, they are the expected weights.
\end{remark}
\begin{corollary}
In local coordinates, Theorem 5, Proposition 7 and 8 imply via Proposition 3 the following bounds:
\begin{enumerate}
\item $\int \frac{1}{\sqrt{1-\mu }}\left(\left(\partial_{r^*}\phi \right)^2+\left(\partial_t\phi \right)^2+ \left(1-\mu \right)|\nabb\phi|^2\right) dVol_{t_*}\leq C\int J^T_\mu\left(\phi\right) n^{\mu}_{t_0} dVol_{t_0}$,
\item $\int \frac{1}{\sqrt{1-\mu }} \left(u^2\left(\partial_u\phi \right)^2+v^2\left(\partial_v\phi \right)^2 +\left(1-\mu \right)|\nabb\phi|^2\right) dVol_{t_*}\leq C E_0\left(\phi\right)$,
\item $\int \sqrt{1-\mu }\frac{\left(r^*\right)^2+t^2}{r^2} \phi^2 dVol_{t_*}\leq C E_0\left(\phi\right) t_*^{-2}$,
\item $\int_{t_1}^{t_2}\int_{-\frac{t}{2}}^{\frac{t}{2}} \frac{r^{1-\delta}\left(\partial_{r^*}\phi\right)^2}{\left(1+|r^*|\right)^2\left(1-\mu \right)} + \frac{r^{1-\delta}\phi^2}{\left(1+|r^*|\right)^4} dVol \leq C E_0\left(\phi\right) t_1^{-2}$, where $t_1 \leq t_2 \leq \left(1.1\right)t_1$,
\item $\int_{t_1}^{t_2}\int_{-\frac{t}{2}}^{\frac{t}{2}} \frac{r^3|\nabb\phi|^2}{\left(1+|r^*|\right)^4 \left(1-\mu \right)} dVol \leq C \sum_{k=0}^1 E_0\left(\Omega^k\phi\right) t_1^{-2}$, where $t_1 \leq t_2 \leq \left(1.1\right)t_1$.
\end{enumerate}
\end{corollary}
\section{Estimates for $\psi$}
In this section, we would like to imitate \cite{DRS} and prove an analogue of Theorem 5. For technical reasons, however, we will need to lose an arbitrarily small power of $t$.
\begin{theorem}
\begin{enumerate}
\item $\int J^T_\mu\left(\psi\right) n^{\mu}_{t_*} dVol_{t_*}\leq C \left(\int J^T_\mu\left(\psi\right) n^{\mu}_{t_0} dVol_{t_0}+E_1\left(\phi\right)\right)$,
\item $\int J^Z_\mu\left(\psi\right) n^{\mu}_{t_*} dVol_{t_*} \leq C t_*^\delta E_1\left(\phi\right) $,
\item $\int_{\{-\frac{t_*}{2} \leq r^* \leq \frac{t_*}{2}\}} J^T_\mu\left(\psi\right) n^{\mu}_{t_*} dVol_{t_*} \leq C E_1\left(\phi\right) t_*^{-2+\delta}$,
\item $\int_{t_1}^{t_2} \int_{-\frac{t}{2}}^{\frac{t}{2}} \displaystyle\sum_l K^{X_l,w^{X_l}} \left(\psi\right) dVol \leq C E_1\left(\phi\right) t_1^{-2+\delta},$ where $t_1 \leq t_2 \leq \left(1.1\right)t_1$,
\end{enumerate}
\end{theorem}
The general strategy is as follows. We follow the argument in \cite{DRS} but now in the conservation law for each of the vector fields, there is an extra error term which is a spacetime integral that looks like $\int_{t_0}^{t_*} V^\mu\partial_{\mu}\psi \Box_g \psi dVol$ (as well as an extra term $-\frac{1}{4}\int_{t_0}^{t_*} w \psi \Box_g \psi dVol$ for the modified currents). Very often, we need to show that this integral decays (or does not grow) with $t_*$, thus we need to "produce" some decay in $t$. We do this by splitting the domain of integration into three regions and estimating them separately:
\begin{enumerate}
\item For the region $\{\frac{t}{2} \leq r^* \leq \infty\}$, we use the fact that $\Box_g \psi$ contains negative powers of $r^*$, (which is a consequence of the asymptotic flatness of Schwarzschild). In this region, negative powers of $r^*$ can be estimated by negative powers of $t$.
\item For the region $\{ -\frac{t}{2} \leq r^* \leq \frac{t}{2} \}$, we note that we have decay in the spacetime integral of $\phi$ for each dyadic slab by Corollary 9.4, 9.5. We therefore estimate the integral on this region by that of $K^X\left(\phi\right)$. Here, it is essential that we use the improved $X$ estimates given by Proposition 8.
\item For the region $\{-\infty \leq r^* \leq -\frac{t}{2}\}$, we make use of the fact that there is an extra factor of $\left(1-\mu \right)^{\frac{1}{2}}$ in the spacetime volume form compared to the volume form on a time-slice (see Section 1.5). From the definition of $r^*$, we have $\left(1-\mu \right) \leq C e^{cr^*}$, thus the factor of $\left(1-\mu \right)^\frac{1}{2}$ gives exponential decay in $r^*$, which translates to exponential decay in $t$ in this region. Therefore, on this region, we first estimate on each time slice, and then carry out the integration in $t$.
\end{enumerate}
Since we will often perform integration dyadically, we first set up the notation. We define a dyadic partition of $[t_0,t_*]$ by $t_0 \leq t_1 \leq ... \leq t_n =t_*$, where $t_i \leq \left(1.1\right) t_{i-1}$ and $n$ is the minimal integer such that this can be done. In particular, $\log\left(t_*-t_0\right)\sim n$.\\
We begin with the $T$ estimate.
\begin{proposition}
$$\int J^{T}_{\mu}\left(\psi\right)n_{t_*}^{\mu}dVol_{t_*}\leq C \int J^{T}_{\mu}\left(\psi\right)n_{t_0}^{\mu}dVol_{t_0}+C \sum_{k=0}^2 E_0\left(\Omega^k \phi\right).$$
\end{proposition}
\begin{proof}
The conservation law gives
$$\int J^{T}_{\mu}\left(\psi\right)n_{t_*}^{\mu}dVol_{t_*}=\int J^{T}_{\mu}\left(\psi\right)n_{t_0}^{\mu}dVol_{t_0}+ \int \partial_t\psi \Box_g \psi dVol.$$
We split the error term into three parts and estimate them separately.\\
By Corollary 9.1,
\begin{equation*}
\begin{split}
&|\int_{\frac{t}{2}}^{\infty} \partial_t\psi \Box_g \psi dVol| \\
\leq & C\int_{t_0}^{t_*} \left(\int_{-\infty}^{\infty} \int_{\mathbb S^2} J^T_\mu \left(\psi\right) n^\mu_{t} dVol_{t}\right)^{\frac{1}{2}}\left(\int_{\frac{t}{2}}^{\infty} \int_{\mathbb S^2} \frac{\left(\log r\right)_+^2}{r^4} \left(\left(\partial_{r^*}\phi \right)^2+|\nabb\Omega\phi|^2 \right)dVol_{t}\right)^{\frac{1}{2}}dt\\
\leq & C\sup_{t_0 \leq t \leq t_*} \left(\int_{-\infty}^{\infty} \int_{\mathbb S^2} J^T_\mu \left(\psi\right) n^\mu_{t} dVol_{t}\right)^{\frac{1}{2}}\left(\sum_{k=0}^1 \int_{-\infty}^{\infty} \int_{\mathbb S^2} J^T_\mu \left(\Omega^k\phi\right) n^\mu_{t_0} dVol_{t_0}\right)^{\frac{1}{2}} \int_{t_0}^{t_*} t^{-\frac{3}{2}} dt\\
\leq & C\sup_{t_0 \leq t \leq t_*} \left(\int_{-\infty}^{\infty} \int_{\mathbb S^2} J^T_\mu \left(\psi\right) n^\mu_{t} dVol_{t}\right)^{\frac{1}{2}}\left(\sum_{k=0}^1 \int_{-\infty}^{\infty} \int_{\mathbb S^2} J^T_\mu \left(\Omega^k\phi\right) n^\mu_{t_0} dVol_{t_0}\right)^{\frac{1}{2}}
\end{split}
\end{equation*}
For the middle region, we observe that 
By Corollary 9.4 and 9.5,
\begin{equation*}
\begin{split}
&|\int_{-\frac{t}{2}}^{\frac{t}{2}} \partial_t\psi \Box_g \psi dVol| \\
\leq & C\int_{t_0}^{t_*} \left(\int_{-\infty}^{\infty} \int_{\mathbb S^2} J^T_\mu \left(\psi\right) n^\mu_{t} dVol_{t}\right)^{\frac{1}{2}}\left(\int_{-\frac{t}{2}}^{\frac{t}{2}} \int_{\mathbb S^2} \frac{\left(1+|r^*|\right)^2\left(\log r\right)_+^2\left(1-\mu \right)^{\frac{3}{2}}}{r^6} \left(\left(\partial_{r^*}\phi \right)^2+|\nabb\Omega\phi|^2 \right)dVol_{t}\right)^{\frac{1}{2}}dt\\
\leq & C\sup_{t_0 \leq t \leq t_*} \left(\int_{-\infty}^{\infty} \int_{\mathbb S^2} J^T_\mu \left(\psi\right) n^\mu_{t} dVol_{t}\right)^{\frac{1}{2}}\left(\sum_{k=0}^2 \sum_{i=0}^{n-1} \int_{t_i}^{t_{i+1}} \left(\int_{-\frac{t}{2}}^{\frac{t}{2}} \sum_l K^{X_l, w^{X_l}}\left(\Omega^k\phi \right)\left(1- \mu\right) dVol_{t}\right)^\frac{1}{2} dt\right)\\
\leq & C\sup_{t_0 \leq t \leq t_*} \left(\int_{-\infty}^{\infty} \int_{\mathbb S^2} J^T_\mu \left(\psi\right) n^\mu_{t} dVol_{t}\right)^{\frac{1}{2}}\left(\sum_{k=0}^2 \sum_{i=0}^{n-1} t_i^{\frac{1}{2}} \left(\int_{t_i}^{t_{i+1}} \int_{-\frac{t}{2}}^{\frac{t}{2}} \sum_l K^{X_l, w^{X_l}}\left(\Omega^k\phi \right) dVol\right)^{\frac{1}{2}}\right)\\
\leq & C\sup_{t_0 \leq t \leq t_*} \left(\int_{-\infty}^{\infty} \int_{\mathbb S^2} J^T_\mu \left(\psi\right) n^\mu_{t} dVol_{t}\right)^{\frac{1}{2}}\left(\sum_{k=0}^2 E_0\left(\Omega^k \phi\right)\right)^{\frac{1}{2}} \left(\sum_{i=0}^{n-1} t_i^{-\frac{1}{2}}\right)\\
\leq & C\sup_{t_0 \leq t \leq t_*} \left(\int_{-\infty}^{\infty} \int_{\mathbb S^2} J^T_\mu \left(\psi\right) n^\mu_{t} dVol_{t}\right)^{\frac{1}{2}}\left(\sum_{k=0}^2 E_0\left(\Omega^k \phi\right)\right)^{\frac{1}{2}}
\end{split}
\end{equation*}
By Corollary 9.1,
\begin{equation*}
\begin{split}
&|\int_{-\infty}^{-\frac{t}{2}} \partial_t\psi \Box_g \psi dVol| \\
\leq & C\int_{t_0}^{t_*} \left(\int_{-\infty}^{\infty} \int_{\mathbb S^2} J^T_\mu \left(\psi \right) n^\mu_{t} dVol_{t}\right)^{\frac{1}{2}}\left(\int_{-\infty}^{-\frac{t}{2}} \int_{\mathbb S^2} \left(1+|r^*|\right)^2\left(1-\mu \right)^{\frac{3}{2}} \left(\left(\partial_{r^*}\phi \right)^2+|\nabb\Omega\phi|^2 \right)dVol_{t}\right)^{\frac{1}{2}}dt\\
\leq & C \sup_{t_0 \leq t \leq t_*} \left(\int_{-\infty}^{\infty} \int_{\mathbb S^2} J^T_\mu \left(\psi \right) n^\mu_{t} dVol_{t}\right)^{\frac{1}{2}}\left(\sum_{k=0}^1\int_{-\infty}^{\infty} \int_{\mathbb S^2} J^T_\mu \left(\Omega^k\psi\right) n^\mu_{t} dVol_{t}\right)^{\frac{1}{2}}\int_{t_0}^{t_*} e^{-ct} dt\\
\leq & C \sup_{t_0 \leq t \leq t_*} \left(\int_{-\infty}^{\infty} \int_{\mathbb S^2} J^T_\mu \left(\psi \right) n^\mu_{t} dVol_{t}\right)^{\frac{1}{2}}\left(\sum_{k=0}^1\int_{-\infty}^{\infty} \int_{\mathbb S^2} J^T_\mu \left(\Omega^k\psi\right) n^\mu_{t} dVol_{t}\right)^{\frac{1}{2}}\\
\end{split}
\end{equation*}
These together show that 
$$\int J^{T}_{\mu}\left(\psi\right)n_{t_*}^{\mu}dVol_{t_*}\leq \int J^{T}_{\mu}\left(\psi\right)n_{t_0}^{\mu}dVol_{t_0}+C\sup_{t_0 \leq t \leq t_*} \left(\int J^T_\mu \left(\psi\right) n^\mu_{t} dVol_{t}\right)^{\frac{1}{2}}\left(\sum_{k=0}^2 E_0\left(\Omega^k \phi\right)\right)^{\frac{1}{2}},$$
which implies the Proposition with the following Lemma, taking $h_1\left(t\right)=0$ and $h_2\left(t\right)=\displaystyle\sum_{k=0}^2 E_0\left(\Omega^k\phi\right)$.
\end{proof}
\begin{lemma}
Suppose $f\left(t\right)$ is continuous, $h_1\left(t\right)$, $h_2\left(t\right)$ are increasing and we have 
$$f\left(t_*\right)\leq C\left(f\left(t_0\right)+ h_1\left(t_*\right) + \sup_{t_0 \leq t \leq t_*} f\left(t\right)^{\frac{1}{2}} h_2\left(t_*\right)^{\frac{1}{2}}\right),$$
for all $t_* \geq t_0$.\\
Then $$f\left(t_*\right)\leq C(f\left(t_0\right)+ h_1\left(t_*\right) + h_2\left(t_*)\right).$$
\end{lemma}
\begin{proof}
Suppose $\displaystyle\sup_{t_0 \leq t \leq t_*} f\left(t\right)$ is achieved by $f\left(\tilde{t}\right)$ for some $t_0 \leq \tilde{t} \leq t_*$.
Then $$f\left(\tilde{t}\right)\leq C\left(f\left(t_0\right)+ h_1\left(\tilde{t}\right) + f\left(\tilde{t}\right)^{\frac{1}{2}} h_2\left(\tilde{t}\right)^{\frac{1}{2}}\right).$$
$h_1\left(t\right)$, $h_2\left(t\right)$ increasing implies,
$$f\left(\tilde{t}\right)\leq C\left(f\left(t_0\right)+ h_1\left(t_*\right) + f\left(\tilde{t}\right)^{\frac{1}{2}} h_2\left(t_*\right)^{\frac{1}{2}}\right).$$
Using Cauchy-Schwarz and subtracting $\frac{1}{2}f\left(\tilde{t}\right)$ from both sides,
$$f\left(\tilde{t}\right)\leq C(f\left(t_0\right)+ h_1\left(t_*\right) + h_2\left(t_*)\right).$$
Clearly, $f\left(t_*\right) \leq \displaystyle\sup_{t_0 \leq t \leq t_*} f\left(t\right)=f\left(\tilde{t}\right)$. Hence we have the lemma.
\end{proof}
We then derive an $X$ estimate. Here unlike in the case for $\phi$, in which the $\tilde{X}$ estimate was used to improve the already known estimate from $X_l$, we need to consider both of them at the same time.
\begin{proposition}
\begin{equation*}
\begin{split}
&\int_{t_0}^{t_*} |K^{\tilde{X},w^{\tilde{X}}}\left(\psi\right)| + \sum_l K^{X_l,w^{X_l}}\left(\psi_l\right)dVol\\
\leq &C\left(\int J^T_{\mu}\left(\psi \right)n_{t_*}^{\mu}dVol_{t_*}+\int J^T_{\mu}\left(\psi \right)n_{t_0}^{\mu}dVol_{t_0}\right)+C t_0^{-2+\delta}\displaystyle\sum_{k=0}^2 E_0\left(\Omega^k \phi \right)
\end{split}
\end{equation*}
\end{proposition}
\begin{remark}
The reader may ask why this Proposition gives decay for the error term while the statement of Proposition 11 does not. In fact, the proof of Proposition 11 is sufficient to show that the error term decays. However, we do not pursue this as it is unnecessary for later use.
\end{remark}
\begin{proof}
Decompose $\psi=\displaystyle\sum_l \psi_l$ into spherical harmonics.\\
Since Schwarzschild spacetimes are spherically symmetric, $\Box_g \psi_l=g_1\left(r^*\right)\partial_{r^*}\phi_l+g_2\left(r^*\right)\nabb\left(\Omega\phi_l\right)$.\\
Notice that $K^{\tilde{X},w^{\tilde{X}}}\left(\psi\right)$ is not everywhere positive. It is identically zero for $r^* \leq 1$ and as we have shown in the proof of Proposition 8, $K^{\tilde{X},w^{\tilde{X}}} \geq 0$ for $r^* \geq max\{100,100M\}$. On the remaining (not necessarily positive) region $1 \leq r^* \leq max\{100,100M\}$, we have $|K^{\tilde{X},w^{\tilde{X}}}\left(\psi_l\right)| \le C K^{X_l,w^{X_l}}\left(\psi_l\right)$. (Notice that we have avoided the region around $r=3M$ where this inequality is potentially problematic.)\\
In particular, applying Proposition 2 for the vector field $X$, we have
\begin{equation*}
\begin{split}
&\int |K^{\tilde{X},w^X}\left(\psi_l\right)|dVol + \int K^{X_l,w^{X_l}}\left(\psi_l\right)dVol\\
\leq &\int K^{\tilde{X},w^{\tilde{X}}}\left(\psi_l\right)dVol + \left(C+1\right) \int K^{X_l,w^{X_l}}\left(\psi_l\right)dVol\\
=&\int J^{\tilde{X},w^{\tilde{X}}}_{\mu}\left(\psi_l\right)n_{t_*}^{\mu}dVol_{t_*}-\int J^{\tilde{X},w^{\tilde{X}}}_{\mu}\left(\psi_l\right)n_{t_0}^{\mu}dVol_{t_0}\\
&+\left(C+1\right)\left(\int J^{X_l,w^{X_l}}_{\mu}\left(\psi_l\right)n_{t_*}^{\mu}dVol_{t_*}-\int J^{X_l,w^{X_l}}_{\mu}\left(\psi_l\right)n_{t_0}^{\mu}dVol_{t_0}\right)\\
&+\frac{1}{4}\int \left(\tilde{f}'+\frac{2\left(1-\mu \right)}{r}\tilde{f}\right)\psi_l\Box\psi_l dVol-\int \left(\tilde{f}\partial_{r^*}\psi_l \right)\Box\psi_l dVol\\
&+\left(C+1\right)\left(\frac{1}{4}\int \left(f'_l+\frac{2\left(1-\mu \right)}{r}f_l\right)\psi_l\Box\psi_l dVol-\int \left(f_l\partial_{r^*}\psi_l \right)\Box\psi_l dVol\right)\\
\leq &C\left(\int J^T_{\mu}\left(\psi_l \right)n_{t_*}^{\mu}dVol_{t_*}+\int J^T_{\mu}\left(\psi_l \right)n_{t_0}^{\mu}dVol_{t_0}+\int |\left(r^{-1}\psi_l+\partial_{r^*}\psi_l \right)\Box\psi_l| dVol\right).
\end{split}
\end{equation*}
We split the last term into three integrals and estimate them separately.\\
By Theorem 5.2,
\begin{equation*}
\begin{split}
&\int_{t_0}^{t_*}\int_{\frac{t}{2}}^{\infty}\int_{\mathbb S^2} |\left(r^{-1}\psi_l+\partial_{r^*}\psi_l \right)\Box\psi_l| dVol\\
\leq &C \int_{t_0}^{t_*}\int_{\frac{t}{2}}^{\infty}\int_{\mathbb S^2} r^{-1+\frac{\delta}{2}}\left(\log r\right)_{+}\left(r^{-\frac{3}{2}-\frac{\delta}{4}}|\psi_l|+r^{-\frac{1}{2}-\frac{\delta}{4}}|\partial_{r^*}\psi_l| \right)\left(r^{-\frac{1}{2}-\frac{\delta}{4}}\left(|\partial_{r^*}\phi|+|\nabb\Omega\phi| \right)\right)dVol\\
\leq &C t_0^{-1+\frac{\delta}{2}}\left(\int_{t_0}^{t_*}\int_{\frac{t}{2}}^{\infty}\int_{\mathbb S^2} |K^{\tilde{X},w^{\tilde{X}}}\left(\psi_l \right)| dVol\right)^{\frac{1}{2}}\left(\displaystyle\sum_{k=0}^1 \int_{t_0}^{t_*}\int_{\frac{t}{2}}^{\infty}\int_{\mathbb S^2} |K^{\tilde{X},w^{\tilde{X}}}\left(\Omega^k \phi_l \right)| dVol\right)^{\frac{1}{2}}\\
\leq &C t_0^{-1+\frac{\delta}{2}}\left(\int_{t_0}^{t_*}\int_{-\infty}^{\infty}\int_{\mathbb S^2} |K^{\tilde{X},w^{\tilde{X}}}\left(\psi_l \right)| dVol\right)^{\frac{1}{2}}\left(\displaystyle\sum_{k=0}^1 \int_{t_0}^{t_*}\int_{-\infty}^{\infty}\int_{\mathbb S^2} |K^{\tilde{X},w^{\tilde{X}}}\left(\Omega^k \phi_l \right)| dVol\right)^{\frac{1}{2}}\\
\leq &\frac{1}{4}\int_{t_0}^{t_*}\int_{-\infty}^{\infty}\int_{\mathbb S^2} |K^{\tilde{X},w^{\tilde{X}}}\left(\psi_l \right)| dVol+ C t_0^{-2+\delta}\displaystyle\sum_{k=0}^1 \int_{-\infty}^{\infty}\int_{\mathbb S^2} J^T_\mu\left(\Omega^k \phi_l \right) n^{\mu }_0 dVol_{t_0}.\\
\end{split}
\end{equation*}
By Theorem 5.5,
\begin{equation*}
\begin{split}
&\int_{t_0}^{t_*}\int_{-\frac{t}{2}}^{\frac{t}{2}}\int_{\mathbb S^2} |\left(r^{-1}\psi_l+\partial_{r^*}\psi_l \right)\Box\psi_l| dVol\\
\leq &C \int_{t_0}^{t_*}\int_{-\frac{t}{2}}^{\frac{t}{2}}\int_{\mathbb S^2} \left(r^{-\frac{3}{2}-\frac{\delta}{4}}|\psi_l|+r^{-\frac{1}{2}-\frac{\delta}{4}}|\partial_{r^*}\psi_l| \right)\left(r^{-\frac{1}{2}-\frac{\delta}{4}}\left(|\partial_{r^*}\phi|+|\nabb\Omega\phi| \right)\right)dVol\\
\leq &C \left(\int_{t_0}^{t_*}\int_{-\frac{t}{2}}^{\frac{t}{2}}\int_{\mathbb S^2} \left(|K^{\tilde{X},w^{\tilde{X}}}\left(\psi_l \right)|+K^{X_l,w^{X_l}}\left(\psi_l \right)\right) dVol\right)^{\frac{1}{2}}\left(\displaystyle\sum_{k=0}^2 \int_{t_0}^{t_*}\int_{-\frac{t}{2}}^{\frac{t}{2}}\int_{\mathbb S^2} |K^{\tilde{X},w^{\tilde{X}}}\left(\Omega^k \phi_l \right)| dVol\right)^{\frac{1}{2}}\\
\leq& \frac{1}{4}\int_{t_0}^{t_*}\int_{-\infty}^{\infty}\int_{\mathbb S^2} \left(|K^{\tilde{X},w^{\tilde{X}}}\left(\psi_l \right)|+K^{X_l,w^{X_l}}\left(\psi_l \right)\right) dVol+C\displaystyle\sum_{i=0}^{n-1}\displaystyle\sum_{k=0}^2 t_i^{-2}E_0\left(\Omega^k \phi_l \right)\\
\leq& \frac{1}{4}\int_{t_0}^{t_*}\int_{-\infty}^{\infty}\int_{\mathbb S^2} \left(|K^{\tilde{X},w^{\tilde{X}}}\left(\psi_l \right)|+K^{X_l,w^{X_l}}\left(\psi_l \right)\right) dVol+C t_0^{-2}\displaystyle\sum_{k=0}^1 E_0\left(\Omega^k \phi_l \right).\\
\end{split}
\end{equation*}
By Theorem 5.1, Proposition 11 and Lemma 6,
\begin{equation*}
\begin{split}
&\int_{t_0}^{t_*}\int_{-\infty}^{-\frac{t}{2}}\int_{\mathbb S^2} |\left(r^{-1}\psi_l+\partial_{r^*}\psi_l \right)\Box\psi_l| dVol\\
\leq &C \int_{t_0}^{t_*}\int_{-\infty}^{-\frac{t}{2}}\int_{\mathbb S^2} |r^*|\left(|\psi_l|+|\partial_{r^*}\psi_l| \right)\left(|\partial_{r^*}\phi_l|+|\nabb\Omega\phi_l| \right)dVol\\
\leq &C \int_{t_0}^{t_*} \left(\int_{-\infty}^{\infty}\left(\frac{1}{\left(1+|r^*|\right)^2}\psi_l^2 + \left(\partial_{r^*}\psi_l \right)^2\right) dA dr^*\right)^{\frac{1}{2}}\\
&\times\left(\int_{-\infty}^{-\frac{t}{2}} \left(r^*\right)^4\left(1-\mu \right) \left(\left(\partial_{r^*}\phi_l\right)^2 + \left(1-\mu \right)|\nabb\Omega\phi_l|^2\right) dA dr^*\right)^{\frac{1}{2}} dt\\
\leq &C \int_{t_0}^{t_*} \left(\int_{-\infty}^{\infty} \left(\partial_{r^*}\psi_l \right)^2 dA dr^*\right)^{\frac{1}{2}}\left(\int_{-\infty}^{\infty} \left(\left(\partial_{r^*}\phi_l\right)^2 + \left(1-\mu \right)|\nabb\Omega\phi_l|^2\right) dA dr^*\right)^{\frac{1}{2}} e^{-ct} dt\\
\leq &C \sup_{t_0 \leq t \leq t_*}\left(\int_{-\infty}^{\infty} J^T_\mu\left(\psi_l \right) n^\mu_{t} dVol_{t}\right)^{\frac{1}{2}}\left(\sum_{k=0}^2 \int_{-\infty}^{\infty} J^T_\mu\left(\Omega^k\phi_l \right) n^\mu_{t} dVol_{t_0}\right)^{\frac{1}{2}} e^{-ct_0}\\
\leq &C \int J^{T}_{\mu}\left(\psi_l\right)n_{t_0}^{\mu}dVol_{t_0}+C t_0^{-2} \sum_{k=0}^2 E_0\left(\Omega^k \phi_l\right),
\end{split}
\end{equation*}
Subtract the terms with $K$ from both sides and get
\begin{equation*}
\begin{split}
&\int_{t_0}^{t_*} |K^{\tilde{X},w^{\tilde{X}}}\left(\psi_l\right)| + K^{X_l,w^X}\left(\psi_l\right)dVol\\
\leq &C\left(\int J^T_{\mu}\left(\psi_l \right)n_{t_*}^{\mu}dVol_{t_*}+\int J^T_{\mu}\left(\psi_l \right)n_{t_0}^{\mu}dVol_{t_0}\right)+C t_0^{-2+\delta}\displaystyle\sum_{k=0}^2 E_0\left(\Omega^k \phi_l \right).
\end{split}
\end{equation*}
Sum over $l \ge 0$ to get the Proposition.
\end{proof}
We localize the estimates in the above Proposition to obtain decay as in \cite{DRS}.
\begin{proposition}
Let $t_0 \leq t_1 \leq \left(1.1\right)t_0$, $|r_1^*|+|r_2^*|\leq\frac{t_0}{2}$. Then 
\begin{equation*}
\begin{split}
&\int_{\mathcal{P}} |K^{\tilde{X},w^{\tilde{X}}}\left(\psi\right)|+\displaystyle\sum_l K^{X_l,w^{X_l}}\left(\psi_l\right) dVol\\
\leq &C\left(t_0^{-2}\int J^Z_\mu\left(\psi\right) n^\mu_0 dVol_{t_0}+t_0^{-2+\delta}\displaystyle\sum_{k=0}^{2}\displaystyle\sum_{m=0}^1 E_0\left(\partial_t^m\Omega^k\phi \right)\right),
\end{split}
\end{equation*}
where $\mathcal{P}=\{t_0\leq t\leq t_1, -\frac{t}{2}\leq r^* \leq \frac{t}{2}\}$ or $\mathcal{P}=\{t_0\leq t\leq t_1, r_1^*-\left(t_1-t\right) \leq r^* \leq r_2^* +\left(t_1-t\right)\}$.
\end{proposition}
\begin{proof}
Let $\chi=\left\{\begin{array}{clcr}1&|x|\le 1\\0&|x|\ge 1.1\end{array}\right.$.
On $t=t_0$, let $\tilde{\phi}=\chi\left(\frac{r^*}{0.65t_0}\right)\phi$, $\partial_t\tilde{\phi}=\chi\left(\frac{r^*}{0.65t_0}\right)\partial_t\phi$ and solve for $\Box_g\tilde{\phi}=0$ for $t \geq t_0$.\\
Following \cite{DRS}, we have 
$$\int_{-0.715t_0}^{0.715t_0} \frac{1}{\sqrt{1-\mu }}\phi^2 dVol_{t_0} \leq \int J^Z_{\mu}\left(\phi \right)n_{t_0}^{\mu}dVol_{t_0}.$$
This is true because of Proposition 3.2, 3.3 and an elementary one-dimensional estimate:
$$\int_{-a}^{a} |f\left(x\right)|^2 dx \leq C a^2 \left(\int_{-a}^a |\partial_x f\left(x\right)|^2 +\int_{-1}^1 |f\left(x\right)|^2 dx\right),$$
for $a \geq 1$.\footnote{One can prove this one-dimensional estimate by first considering $g=0$ on $[-\frac{1}{2}, \frac{1}{2}]$ and the trivial bound $\int_{-a}^a |g\left(x\right)| dx \leq C a \int_{-a}^a |\partial_xg\left(x\right)| dx.$ Then one sets $g\left(x\right)=f\left(x\right)^2$ and use Cauchy-Schwarz to get $\int_{-a}^a |f\left(x\right)|^2 dx \leq C a^2 \int_{-a}^a |\partial_xf\left(x\right)|^2 dx.$. Finally, one cuts off $f\left(x\right)$ to be identically zero in $[-\frac{1}{2}, \frac{1}{2}]$.}\\
Using this, we can estimate the current of $\tilde{\phi}$:
\begin{equation*}
\begin{split}
&\int J^T_{\mu}\left(\tilde{\phi} \right)n_{t_0}^{\mu}dVol_{t_0}\\
\leq &\int_{-0.715t_0}^{0.715t_0} J^T_{\mu}\left(\phi \right)n_{t_0}^{\mu}dVol_{t_0} + C t_0^{-2}\int_{-0.715t_0}^{0.715t_0} \frac{1}{\sqrt{1-\mu }}\phi^2 dVol_{t_0}\\
\leq & C t_0^{-2} \int J^Z_{\mu}\left(\phi \right)n_{t_0}^{\mu}dVol_{t_0}.
\end{split}
\end{equation*}
Similarly,
$$\int J^T_{\mu}\left(\Omega\tilde{\phi} \right)n_{t_0}^{\mu}dVol_{t_0} \leq C t_0^{-2} \int J^Z_{\mu}\left(\Omega\phi \right)n_{t_0}^{\mu}dVol_{t_0}.$$
Define $\tilde{\psi}=S\tilde{\phi}$.
By Proposition 13, 
\begin{equation*}
\begin{split}
&\int_{\mathcal{P}} |K^{\tilde{X},w^{\tilde{X}}}\left(\tilde{\psi}\right)|+\displaystyle\sum_l K^{X_l,w^{X_l}}\left(\tilde{\psi_l}\right) dVol\\
\leq &C\left(\int J^T_{\mu}\left(\tilde{\psi } \right)n_{t_1}^{\mu}dVol_{t_1}+\int J^T_{\mu}\left(\tilde{\psi} \right)n_{t_0}^{\mu}dVol_{t_0}\right)+C t_0^{-2+\delta}\displaystyle\sum_{k=0}^2 E_0\left(\Omega^k \tilde{\phi} \right).
\end{split}
\end{equation*}
The left hand side equals $\int_{\mathcal{P}} |K^{\tilde{X},w^{\tilde{X}}}\left(\psi\right)|+\displaystyle\sum_l K^{X_l,w^{X_l}}\left(\psi_l\right) dVol$ by finite speed of propagation.\\
$\int J^T_{\mu}\left(\tilde{\psi} \right)n_{t_0}^{\mu}dVol_{t_0}$ can be estimated in a similar way as $\int J^T_{\mu}\left(\tilde{\phi} \right)n_{t_0}^{\mu}dVol_{t_0}$. More specifically, we claim that 
$$\int J^T_{\mu}\left(\tilde{\psi} \right)n_{t_0}^{\mu}dVol_{t_0} \leq C t_0^{-2} \left(\int J^Z_{\mu}\left(\psi \right)n_{t_0}^{\mu}dVol_{t_0}+E_0\left(\phi\right)\right).$$
To see this, we first note that
$$\tilde{\psi}=\chi\left(\frac{r^*}{0.65t_0}\right)\psi+\frac{r^*}{0.65t_0}\chi'\left(\frac{r^*}{0.65t_0}\right)\phi.$$
Also note that on the support of $\tilde{\psi}$, $\frac{|r^*|}{t_0}\leq C$.
Therefore,
\begin{equation*}
\begin{split}
&\int J^T_{\mu}\left(\tilde{\psi} \right)n_{t_0}^{\mu}dVol_{t_0}\\
\leq &\int_{-0.715t_0}^{0.715t_0} J^T_{\mu}\left(\psi \right)n_{t_0}^{\mu}dVol_{t_0} + C t_0^{-2}\int_{-0.715t_0}^{0.715t_0} \frac{1}{\sqrt{1-\mu }}\psi^2 dVol_{t_0} +C\int_{-0.715t_0}^{0.715t_0} J^T_{\mu}\left(\phi \right)n_{t_0}^{\mu}dVol_{t_0}\\
&+ C t_0^{-2}\int_{-0.715t_0}^{0.715t_0} \frac{1}{\sqrt{1-\mu }}\phi^2 dVol_{t_0}\\
\leq &C t_0^{-2} \left(\int J^Z_{\mu}\left(\psi \right)n_{t_0}^{\mu}dVol_{t_0}+\int J^Z_{\mu}\left(\phi \right)n_{t_0}^{\mu}dVol_{t_0}\right)\\
\leq &C t_0^{-2} \left(\int J^Z_{\mu}\left(\psi \right)n_{t_0}^{\mu}dVol_{t_0}+E_0\left(\phi\right)\right).
\end{split}
\end{equation*}
We would now want to control $\int J^T_{\mu}\left(\tilde{\psi } \right)n_{t_1}^{\mu}dVol_{t_1}$. Using the conservation law for $T$ and an integration by parts in $t$,
\begin{equation*}
\begin{split}
&\int J^T_{\mu}\left(\tilde{\psi } \right)n_{t_1}^{\mu}dVol_{t_1}\\
=&\int J^T_{\mu}\left(\tilde{\psi } \right)n_{t_0}^{\mu}dVol_{t_0}-\int \partial_t\tilde{\psi}\Box_g\tilde{\psi} dVol\\
\leq &\int J^T_{\mu}\left(\tilde{\psi } \right)n_{t_0}^{\mu}dVol_{t_0}+|\int \tilde{\psi}\Box_g\left(\partial_t\tilde{\psi}\right)dVol|+|\int \tilde{\psi}\Box_g \tilde{\psi}\sqrt{1-\mu } dVol_{t_0}|\\
&+|\int \tilde{\psi}\Box_g \tilde{\psi}\sqrt{1-\mu } dVol_{t_1}|.
\end{split}
\end{equation*}
We first estimate the spacetime error term in this expression. Using Proposition 3, 7, and 8,
\begin{equation*}
\begin{split}
&|\int \tilde{\psi}\Box_g \left(\partial_t\tilde{\psi}\right)dVol|\\
\leq &\int |\tilde{\psi}\Box_g \left(S \left(\partial_t\tilde{\phi}\right)\right)|dVol+\int |\tilde{\psi}\Box_g \left(\partial_t\tilde{\phi}\right)|dVol\\
\leq &C\int \frac{\left(1+|r^*|\right)\left(\log r\right)_+}{r^3}|\tilde{\psi}|\left(|\partial_{r^*}\partial_t\tilde{\phi}|+|\nabb \left(\partial_t\Omega\tilde{\phi}\right)|\right)dVol\\
\leq &C\left(\int r^{\delta} \frac{r^{1-\frac{\delta}{4}}\tilde{\psi}^2}{\left(1+|r^*|\right)^{4}} dVol\right)^{\frac{1}{2}}\left(\int \frac{\left(1+|r^*|\right)^6\left(\left(\partial_t\partial_{r^*}\tilde{\phi}\right)^2+|\nabb \left(\Omega\partial_t\tilde{\phi}\right)|^2\right)}{r^{7+\frac{\delta}{4}}}dVol\right)^{\frac{1}{2}}\\
\leq &C t_0^{\frac{\delta}{2}}\left(\int |K^{\tilde{X},w^{\tilde{X}}}\left(\tilde{\psi}\right)|+\displaystyle\sum_l K^{X_l}\left(\tilde{\psi_l} \right) dVol\right)^{\frac{1}{2}}\\
&\times\left(\displaystyle\sum_{k=0}^2 \int |K^{\tilde{X},w^{\tilde{X}}}\left(\partial_t\Omega^k\tilde{\phi} \right)|+\displaystyle\sum_l K^{X_l,w^{X_l}}\left(\partial_t\Omega^k\tilde{\phi_l} \right) dVol\right)^{\frac{1}{2}}\\
\leq &C t_0^{\frac{\delta}{2}}\left(\int J^T_{\mu}\left(\tilde{\psi } \right)n_{t_1}^{\mu}dVol_{t_1}+\int J^T_{\mu}\left(\tilde{\psi } \right)n_{t_0}^{\mu}dVol_{t_0}+t_0^{-2+\delta}\sum_{k=0}^2 E_0\left(\Omega^k\tilde{\phi}\right)\right)^{\frac{1}{2}}\\
&\times\left(\sum_{k=0}^2\int J^T_{\mu}\left(\partial_t\Omega^k\tilde{\phi } \right)n_{t_0}^{\mu}dVol_{t_0}\right)^{\frac{1}{2}}\\
\leq &\frac{1}{4}\int J^T_{\mu}\left(\tilde{\psi } \right)n_{t_1}^{\mu}dVol_{t_1}+C\int J^T_{\mu}\left(\tilde{\psi } \right)n_{t_0}^{\mu}dVol_{t_0}+C t_0^{-2+\delta}\displaystyle\sum_{k=0}^2 E_0\left(\Omega^k \tilde{\phi} \right)\\
&+ Ct_0^{-2+\delta}\left(\sum_{k=0}^2\int J^{Z,w^Z}_{\mu}\left(\partial_t\Omega^k\tilde{\phi } \right)n_{t_0}^{\mu}dVol_{t_0}\right)\\
\leq &\frac{1}{4}\int J^T_{\mu}\left(\tilde{\psi } \right)n_{t_1}^{\mu}dVol_{t_1}+C\int J^T_{\mu}\left(\tilde{\psi } \right)n_{t_0}^{\mu}dVol_{t_0}+C t_0^{-2+\delta}\displaystyle\sum_{m=0}^1\sum_{k=0}^2 E_0\left(\partial_t^m\Omega^k \tilde{\phi} \right)\\
\end{split}
\end{equation*}
where at the third to last step we again used Proposition 13.\\
We estimate the boundary terms using Lemma 6 and Corollary 9.1
\begin{equation*}
\begin{split}
&|\int \tilde{\psi}\Box_g \tilde{\psi}\sqrt{1-\mu } dVol_{t_0}|\\
\leq& C \int \left(1+|r^*|\right)r^{-3}\left(\log r\right)_+ \sqrt{1-\mu }|\tilde{\psi}|\left(|\partial_{r^*}\tilde{\phi}|+|\nabb\Omega\tilde{\phi}|\right) dVol_{t_0}\\
\leq &C \left(\int \left(1+|r^*|\right)^{-2}\tilde{\psi}^2\left(1-\mu \right)^{-\frac{1}{2}} dVol_{t_0}\right)^{\frac{1}{2}}\left(\int \left(\left(\partial_{r^*}\tilde{\phi}\right)^2+|\nabb\Omega\tilde{\phi}|^2\right) \sqrt{1-\mu }dVol_{t_0}\right)^{\frac{1}{2}}\\
\leq & C \int J^T_{\mu}\left(\tilde{\psi}\right)n_{t_0}^{\mu}dVol_{t_0}+C \int \left(J^T_{\mu}\left(\tilde{\phi}\right)+J^T_{\mu}\left(\Omega\tilde{\phi}\right)\right)n_{t_0}^{\mu}dVol_{t_0}.
\end{split}
\end{equation*}
The term for $t=t_1$ is done analogously, but with a more careful choice of constant.
\begin{equation*}
\begin{split}
&|\int \tilde{\psi}\Box_g \tilde{\psi}\sqrt{1-\mu } dVol_{t_1}|\\
\leq& C \int \left( 1+|r^*|\right) r^{-3}\left( \log r \right)_+ \sqrt{1-\mu }|\tilde{\psi}|\left(|\partial_{r^*}\tilde{\phi}|+|\nabb\Omega\tilde{\phi}|\right) dVol_{t_1}\\
\leq &C \left(\int \left( 1+|r^*|\right)^{-2}\tilde{\psi}^2\left(1-\mu \right)^{-\frac{1}{2}} dVol_{t_1}\right)^{\frac{1}{2}}\left(\int \left(\left(\partial_{r^*}\tilde{\phi}\right)^2+|\nabb\Omega\tilde{\phi}|^2\right) \sqrt{1-\mu }dVol_{t_1}\right)^{\frac{1}{2}}\\
\leq & \frac{1}{4} \int J^T_{\mu}\left(\tilde{\psi}\right)n_{t_1}^{\mu}dVol_{t_1}+C \int \left(J^T_{\mu}\left( \tilde{\phi}\right)+J^T_{\mu}\left(\Omega\tilde{\phi}\right)\right)n_{t_1}^{\mu}dVol_{t_1}\\
=& \frac{1}{4} \int J^T_{\mu}\left(\tilde{\psi}\right)n_{t_1}^{\mu}dVol_{t_1}+C \int \left(J^T_{\mu}\left(\tilde{\phi}\right)+J^T_{\mu}\left(\Omega\tilde{\phi}\right)\right)n_{t_0}^{\mu}dVol_{t_0}.
\end{split}
\end{equation*}
Combining these estimates and subtracting $\frac{1}{2} \int J^T_{\mu}\left(\tilde{\psi}\right)n_{t_1}^{\mu}dVol_{t_1}$ on both sides, we get
\begin{equation*}
\begin{split}
&\int J^T_{\mu}\left(\tilde{\psi } \right)n_{t_1}^{\mu}dVol_{t_1}\\
\leq &Ct_0^{-2}\int J^Z_{\mu}\left(\psi \right)n_{t_0}^{\mu}dVol_{t_0}+ C t_0^{-2+\delta}\sum_{m=0}^1 \sum_{k=0}^{2}  E_0\left(\partial_t^m\Omega^k \tilde{\phi} \right).
\end{split}
\end{equation*}
It remains to control $\displaystyle\sum_{m=0}^1 \sum_{k=0}^{2} E_0\left(\Omega^k \partial_t^m \tilde{\phi} \right)$.
\begin{equation*}
\begin{split}
&\sum_{m=0}^1 \sum_{k=0}^{2} E_0\left(\partial_t^m \Omega^k \tilde{\phi} \right)\\
\leq &C \int_{-0.715t_0}^{0.715t_0}\left(\sum_{m=0}^1 \sum_{k=0}^5 J^N_\mu\left(\partial_t^m\Omega^k\tilde{\phi} \right)n_{t_0}^\mu+\sum_{m=0}^1 \sum_{k=0}^4J^Z_\mu\left(\partial_t^m\Omega^k\tilde{\phi} \right)n_{t_0}^\mu\right) dVol_{ \{ t=t_0\}}\\
\leq &C \int_{-0.715t_0}^{0.715t_0} \left(\sum_{m=0}^1\sum_{k=0}^5 J^N_\mu\left(\partial_t^m\Omega^k\phi \right)n_{t_0}^\mu+\sum_{m=0}^1 \sum_{k=0}^4 J^Z_\mu\left(\partial_t^m\Omega^k\phi \right)n_{t_0}^\mu\right) dVol_{ \{ t=t_0\}}\\ 
&+ Ct_0^{-2} \int_{-0.715t_0}^{0.715t_0} \left(\sum_{m=0}^1 \sum_{k=0}^5\left(\partial_t^m\Omega^k\phi\right)^2 +\sum_{m=0}^1 \sum_{k=0}^4 t_0^2 \left(\partial_t^m\Omega^k\phi \right)^2 \right) r^2 dAdr^*\\
\leq &C \sum_{m=0}^1 \sum_{k=0}^2 E_0\left(\partial_t^m\Omega^k \phi\right).
\end{split}
\end{equation*}
\end{proof}
After establishing the $X$ estimates, we turn to the $Z$ estimates for $\psi$.
\begin{proposition}
\begin{equation*}
\begin{split}
&\int J^{Z,w^Z}_\mu \left( \psi \right) n^\mu_{t_*} dVol_{t_*}\\
\leq &C \int J^{Z,w^Z}_\mu \left( \psi \right) n^\mu_{t_0} dVol_{t_0}+C\displaystyle\sum_{k=0}^1\int_{t_0}^{t_*}\int_{r_1^*}^{r_2^*}t \displaystyle\sum_l K^{X_l,w^{X_l}}\left( \Omega^k\psi_l \right) dVol\\
&+ C\left( \int_{t_0}^{t_*}\int_{-\frac{t}{2}}^{\frac{t}{2}} t^{2+2\delta }\left( |K^{\tilde{X},w^{\tilde{X}}}\left( \psi \right) |+\displaystyle\sum_l K^{X_l,w^{X_l}}\left( \psi_l \right)\right) dVol\right)^{\frac{1}{2}}\left( \sum_{m=0}^1 \sum_{k=0}^2 E_0\left(\partial_t^m \Omega^k \phi \right)\right)^{\frac{1}{2}}\\
& + Ct_*^\delta \sum_{k=0}^2 E_0\left( \Omega^k\phi \right).
\end{split}
\end{equation*}
\end{proposition}
\begin{proof}
By Proposition 2 applied to the vector field $Z$,
\begin{equation*}
\begin{split}
&\int J^{Z,w^Z}_{\mu}\left(\psi\right) n_{t_*}^{\mu} dVol_{\{ t=t_*\} }\\
=&\int J^{Z,w^Z}_{\mu}\left(\psi \right) n_{t_0}^{\mu} dVol_{\{ t=t_0\} } +\int K^{Z,w^Z}\left(\psi \right) dVol-\int \frac{tr^*\left(1-\mu \right)}{2r}\Box_g\psi dVol\\
&+\int \left(u^2\partial_u\psi+v^2\partial_v\psi \right)\Box_g\psi dVol.
\end{split}
\end{equation*}
As remarked before, there exists $r_1^*, r_2^*$ with $r_1^* < r_2^*$ such that $K^{Z,w^Z}\left(\psi \right)$ is non-positive for $r^* \leq r_1^*$ or $r^* \geq r_2^*$.
Therefore,
\begin{equation*}
\begin{split}
\int K^{Z,w^Z}\left(\psi \right) dVol \leq &\int_{r_1^*}^{r_2^*} K^{Z,w^Z}\left(\psi \right) dVol\\
\leq & C\int_{t_0}^{t_*}\int_{r_1^*}^{r_2^*} t\left(\psi^2+|\nabb\psi|^2\right) dVol\\
\leq & C\displaystyle\sum_{k=0}^1\int_{t_0}^{t_*}\int_{r_1^*}^{r_2^*}t\displaystyle\sum_l K^{X_l,w^{X_l}}\left(\Omega^k\psi_l \right) dVol.
\end{split}
\end{equation*}
For the first error term, we again estimate by looking at three separate regions.\\
By Proposition 3.3 and Corollary 9.1,
\begin{equation*}
\begin{split}
&|\int_{t_0}^{t_*}\int_{\frac{t}{2}}^{\infty} \frac{tr^*\left( 1-\mu \right)}{2r}\psi\Box_g\psi dVol|\\
\leq &C \int_{t_0}^{t_*} \left(\int_{\frac{t}{2}}^{\infty}\int_{\mathbb S^2} \frac{t^2}{r^2} \psi^2 r^2 \left( 1-\mu \right) dA dr^*\right)^{\frac{1}{2}}\left(\int_{\frac{t}{2}}^{\infty}\int_{\mathbb S^2} \frac{\left(\log r\right)_{+}^2}{r^2}\left(\left(\partial_{r^*}\phi \right)^2+|\nabb\Omega\phi|^2 \right) r^2 dAdr^* \right)^{\frac{1}{2}}dt\\
\leq &C \sup_{t_0 \leq t \leq t_*} \left(\int J^{Z,w^Z}_{\mu}\left(\psi \right) n_{t}^{\mu} dVol_{t}\right)^{\frac{1}{2}}\left(\sum_{k=0}^1 \int J^T_{\mu}\left(\Omega^k\phi \right) n_{t_0}^{\mu} dVol_{\{ t=t_0\} }\right)^{\frac{1}{2}}\int_{t_0}^{t_*} t^{-1+\frac{\delta}{2}} dt\\
\leq &C t_*^{\frac{\delta}{2}}\sup_{t_0 \leq t \leq t_*} \left(\int J^{Z,w^Z}_{\mu}\left(\psi \right) n_{t}^{\mu} dVol_{t}\right)^{\frac{1}{2}}\left(\sum_{k=0}^1 E_0\left(\Omega^k\phi\right)\right)^{\frac{1}{2}}.
\end{split}
\end{equation*}
By Proposition 3.3, Corollary 9.4 and Corollary 9.5,
\begin{equation*}
\begin{split}
&|\int_{t_0}^{t_*}\int_{-\frac{t}{2}}^{\frac{t}{2}} \frac{tr^*\left( 1-\mu \right)}{2r}\psi\Box_g\psi dVol|\\
\leq &C \int_{t_0}^{t_*} \left(\int_{-\frac{t}{2}}^{\frac{t}{2}}\int_{\mathbb S^2} \frac{t^2}{r^2} \psi^2 r^2 \left( 1-\mu\right) dA dr^*\right)^{\frac{1}{2}}\\
&\times\left( \int_{-\frac{t}{2}}^{\frac{t}{2}}\int_{\mathbb S^2} \frac{\left(1+|r^*|\right)^2\left(\log r \right)_{+}^2\left(1-\mu \right)^2}{r^4}\left( \left(\partial_{r^*}\phi \right)^2+|\nabb\Omega\phi|^2 \right) r^2 dAdr^*\right)^{\frac{1}{2}}dt\\
\leq &C \sup_{t_0 \leq t \leq t_*} \left(\int J^{Z,w^Z}_{\mu}\left(\psi \right) n_{t}^{\mu} dVol_{t}\right)^{\frac{1}{2}}\int_{t_0}^{t_*} \left(\int_{-\frac{t}{2}}^{\frac{t}{2}}\int_{\mathbb S^2} \frac{\left(\log r \right)_{+}^2}{r^2}\left(\left(\partial_{r^*}\phi \right)^2+|\nabb\Omega\phi|^2 \right) r^2\left(1-\mu \right) dAdr^*\right)^{\frac{1}{2}}dt\\
\leq &C \sup_{t_0 \leq t \leq t_*} \left(\int J^{Z,w^Z}_{\mu}\left(\psi \right) n_{t}^{\mu} dVol_{t}\right)^{\frac{1}{2}}\\
&\times\left(\sum_{i=0}^{n-1} t_i^{\frac{1}{2}}\left(\int_{t_i}^{t_{i+1}}\int_{-\frac{t}{2}}^{\frac{t}{2}} \frac{\left(\log r\right)_{+}^2}{r^2}\left(\left(\partial_{r^*}\phi \right)^2+|\nabb\Omega\phi|^2 \right) r^2\left(1-\mu \right) dAdr^*dt\right)^{\frac{1}{2}}\right) \\
\leq &C \sup_{t_0 \leq t \leq t_*} \left(\int J^{Z,w^Z}_{\mu}\left(\psi \right) n_{t}^{\mu} dVol_{t}\right)^{\frac{1}{2}}\left(\sum_{i=0}^{n-1}\sum_{k=0}^2 t_i^{-\frac{1}{2}}E_0\left(\Omega^k\phi \right) \right)^{\frac{1}{2}}\\
\leq &C \sup_{t_0 \leq t \leq t_*} \left(\int J^{Z,w^Z}_{\mu}\left(\psi \right) n_{t}^{\mu} dVol_{t}\right)^{\frac{1}{2}}\left(\sum_{k=0}^2 E_0\left(\Omega^k\phi \right) \right)^{\frac{1}{2}}.
\end{split}
\end{equation*}
By Proposition 3.3 and Corollary 9.1 and using the fact that $\left(1-\mu \right)\leq Ce^{cr^*}$,
\begin{equation*}
\begin{split}
&|\int_{t_0}^{t_*}\int_{-\infty}^{-\frac{t}{2}} \frac{tr^*\left( 1-\mu \right)}{2r}\psi\Box_g\psi dVol|\\
\leq &C \int_{t_0}^{t_*} \left( \int_{-\infty}^{-\frac{t}{2}}\int_{\mathbb S^2} \frac{t^2}{r^2} \psi^2 r^2 \left( 1-\mu \right) dA dr^*\right)^{\frac{1}{2}}\\
&\times\left( \int_{-\infty}^{-\frac{t}{2}}\int_{\mathbb S^2} \left( r^*\right)^4 \left( \left( \partial_{r^*}\phi \right)^2+\left( 1-\mu \right)|\nabb\Omega\phi|^2 \right) r^2 \left( 1-\mu \right)dAdr^*\right)^{\frac{1}{2}}dt\\
\leq &C \sup_{t_0 \leq t \leq t_*} \left(\int J^{Z,w^Z}_{\mu}\left(\psi \right) n_{t}^{\mu} dVol_{t}\right)^{\frac{1}{2}}\int_{t_0}^{t_*}e^{-ct}\left(\sum_{k=0}^1\int J^T_{\mu}\left(\Omega^k\phi \right) n^{\mu}_{t} dVol_{t}\right)^{\frac{1}{2}}dt\\
\leq &C \sup_{t_0 \leq t \leq t_*} \left(\int J^{Z,w^Z}_{\mu}\left(\psi \right) n_{t}^{\mu} dVol_{t}\right)^{\frac{1}{2}}\left(\sum_{k=0}^1E_0\left(\Omega^k\phi \right) \right)^{\frac{1}{2}}.
\end{split}
\end{equation*}
The estimation of the second error term is slightly more involved because there is a factor of $t^2$ in the integrand. In particular, even near spacelike infinity, one needs to use estimates for the spacetime integral for $\phi$.
We intend to estimate this term separately in three regions as above. However, for technical reasons, we will divide the regions slightly differently. Divide as usual the interval into $t_0 \leq t_1 \leq \dots \leq t_n=t_*$. We then set the three regions to be $\displaystyle\bigcup_{i=0}^{n-1} \{t_i \leq t \leq t_{i+1}, r^* > \frac{t_i}{2}\}$, $\displaystyle\bigcup_{i=0}^{n-1} \{t_i \leq t \leq t_{i+1}, -\frac{t_i}{2} \leq r^* \leq \frac{t_i}{2} \}$, $\displaystyle\bigcup_{i=0}^{n-1} \{t_i \leq t \leq t_{i+1}, r^* < -\frac{t_i}{2}\}$.\\
In the region $\displaystyle\bigcup_{i=0}^{n-1} \{t_i \leq t \leq t_{i+1}, r^* > \frac{t_i}{2}\}$, we estimate one power of $t$ by that in $J^{Z, w^Z}\left(\phi\right)$ and the other is canceled with the decay in $r$. To achieve this we use Proposition 3.3, Theorem 5.1, 5.2 and Proposition 8,
\begin{equation*}
\begin{split}
&|\sum_{i=0}^{n-1}\int_{t_i}^{t_{i+1}}\int_{\frac{t_i}{2}}^{\infty} \left(u^2\partial_u\psi+v^2\partial_v\psi \right)\Box_g\psi dVol|\\
\leq &C\sum_{i=0}^{n-1}\int_{t_i}^{t_{i+1}} \left(\int_{\frac{t_i}{2}}^{\infty}\int_{\mathbb S^2} \left( u^2\left(\partial_u\psi  \right)^2+v^2\left(\partial_v\psi \right)^2\right)r^2 dA dr^*\right)^{\frac{1}{2}}\\
&\times\left(\int_{\frac{t}{2}}^{\infty}\int_{\mathbb S^2} \frac{\left(\log r\right)_{+}^2}{r^2}\left(\left(\partial_{r^*}\phi \right)^2+|\nabb\Omega\phi|^2 \right) r^2 dAdr^*\right)^{\frac{1}{2}}dt\\
\leq &C \sup_{t_0 \leq t \leq t_*} \left(\int J^{Z,w^Z}_{\mu}\left(\psi \right) n_{t}^{\mu} dVol_{t}\right)^{\frac{1}{2}}\left(\displaystyle\sum_{i=0}^{n-1} \int_{t_i}^{t_{i+1}}\left(\int_{\frac{t_i}{2}}^{\infty} \int_{\mathbb S^2} \frac{\left( \log r\right)_{+}^2}{r^2}\left( \left( \partial_{r^*}\phi \right)^2+|\nabb\Omega\phi|^2 \right) r^2 dA dr^*\right)^{\frac{1}{2}} dt \right)\\
\leq &C \sup_{t_0 \leq t \leq t_*} \left( \int J^{Z,w^Z}_{\mu}\left( \psi \right) n_{t}^{\mu} dVol_{t}\right)^{\frac{1}{2}}\left(\displaystyle\sum_{i=0}^{n-1} t_i^{\frac{1}{2}}\int_{t_i}^{t_{i+1}}\int_{\frac{t_i}{2}}^{\infty} \int_{\mathbb S^2} \frac{1}{r^{2-\frac{\delta}{2}}} \left( \left( \partial_{r^*}\phi \right)^2+|\nabb\Omega\phi|^2 \right) r^2 dA dr^* dt\right)^{\frac{1}{2}}\\
\leq &C \sup_{t_0 \leq t \leq t_*} \left(\int J^{Z,w^Z}_{\mu}\left( \psi \right) n_{t}^{\mu} dVol_{t}\right)^{\frac{1}{2}}\left(\displaystyle\sum_{i=0}^{n-1} t_i^{\frac{1}{2}}t_i^{-\frac{1}{2}+\frac{\delta}{2}} \int_{t_i}^{t_{i+1}} \int_{\frac{t_i}{2}}^{\infty} \int_{\mathbb S^2} \frac{1}{r^{1+\frac{\delta}{2}}} \left(\left( \partial_{r^*}\phi \right)^2+|\nabb\Omega\phi|^2 \right) r^2 dA dr^* dt\right)^{\frac{1}{2}}\\
\leq &C \sup_{t_0 \leq t \leq t_*} \left(\int J^{Z,w^Z}_{\mu}\left( \psi \right) n_{t}^{\mu} dVol_{t}\right)^{\frac{1}{2}}\left(\displaystyle\sum_{i=0}^{n-1} t_i^{\frac{\delta}{2}} \left( \int J^T_\mu \left( \phi\right) n^\mu_{t_i} dVol_{t_i}+\int J^T_\mu \left( \phi\right) n^\mu_{t_{i+1}} dVol_{t_{i+1}}\right)^{\frac{1}{2}}\right)\\
\leq &C \sup_{t_0 \leq t \leq t_*} \left(\int J^{Z,w^Z}_{\mu}\left(\psi \right) n_{t}^{\mu} dVol_{t}\right)^{\frac{1}{2}}\left(\int J^T_\mu\left(\phi\right) n^\mu_{t_0} dVol_{t_0}\right)^{\frac{1}{2}}\left(\displaystyle\sum_{i=0}^{n-1} t_i^{\frac{\delta}{2}} \right)\\
\leq &C t_*^{\frac{\delta}{2}}\sup_{t_0 \leq t \leq t_*} \left(\int J^{Z,w^Z}_{\mu}\left(\psi \right) n_{t}^{\mu} dVol_{t}\right)^{\frac{1}{2}}\left(\int J^T_\mu\left(\phi\right) n^\mu_{t_0} dVol_{t_0}\right)^{\frac{1}{2}}.
\end{split}
\end{equation*}
For the region $\displaystyle\bigcup_{i=0}^{n-1} \{t_i \leq t \leq t_{i+1}, -\frac{t_i}{2} \leq r^* \leq \frac{t_i}{2} \}$, we first rewrite into $\left(t, r^*\right)$-coordinates and then perform an integration by parts in $t$. It is to avoid extra boundary terms during this integration by parts that we have divided our regions differently from before. The reason that we perform this integration by parts is that instead of a spacetime integral term with $\partial_t\psi$, we would prefer a term with $\psi$, which can then be controlled by the integral of $|K^{ \tilde{X}, w^{\tilde{X}}}| +\displaystyle\sum_l K^{X_l, w^{X_l}}$.
\begin{equation*}
\begin{split}
&|\sum_{i=0}^{n-1}\int_{t_i}^{t_{i+1}}\int_{-\frac{t_i}{2}}^{\frac{t_i}{2}} \left( u^2\partial_u\psi+ v^2\partial_v\psi \right)\Box_g\psi dVol|\\
\leq &C \sum_{i=0}^{n-1}\int_{t_i}^{t_{i+1}}\int_{-\frac{t_i}{2}}^{\frac{t_i}{2}}\int_{\mathbb S^2} |tr^*\partial_{r^*}\psi\Box_g\psi| r^2\left(1-\mu \right) dA dr^* dt \\
&+C|\sum_{i=0}^{n-1}\int_{t_i}^{t_{i+1}}\int_{-\frac{t_i}{2}}^{\frac{t_i}{2}}\int_{\mathbb S^2} \left( t^2+\left( r^* \right)^2 \right)\partial_t\psi \Box_g\psi r^2\left( 1-\mu \right) dA dr^* dt|\\
\leq &C \sum_{i=0}^{n-1}\int_{t_i}^{t_{i+1}}\int_{-\frac{t_i}{2}}^{\frac{t_i}{2}}\int_{\mathbb S^2} |tr^*\partial_{r^*}\psi\Box_g\psi| r^2\left( 1-\mu \right) dA dr^* dt\\
&+C\sum_{i=0}^{n-1}\int_{t_i}^{t_{i+1}}\int_{-\frac{t_i}{2}}^{\frac{t_i}{2}}\int_{\mathbb S^2} t^2|\psi \Box_g\left( \partial_t\psi\right)| r^2\left( 1-\mu \right) dA dr^* dt+C\sum_{i=0}^{n-1}\int_{t_i}^{t_{i+1}}\int_{-\frac{t_i}{2}}^{\frac{t_i}{2}}\int_{\mathbb S^2} t|\psi \Box_g\psi| r^2\left( 1-\mu \right) dA dr^* dt\\
&+C\sum_{i=0}^{n-1}\int_{-\frac{t_i}{2}}^{\frac{t_i}{2}} t^2|\psi \Box_g\psi| \sqrt{1-\mu } dVol_{t_i}+C\sum_{i=0}^{n-1}\int_{-\frac{t_{i}}{2}}^{\frac{t_{i}}{2}} t^2|\psi \Box_g\psi| \sqrt{1-\mu } dVol_{t_{i+1}}\\
\leq &C \int_{t_0}^{t_*}\int_{-\frac{t}{2}}^{\frac{t}{2}}\int_{\mathbb S^2} |tr^*\partial_{r^*}\psi\Box_g\psi| r^2\left( 1-\mu \right) dA dr^* dt+C\int_{t_0}^{t_*}\int_{-\frac{t}{2}}^{\frac{t}{2}}\int_{\mathbb S^2} t^2|\psi \Box_g\left( \partial_t\psi \right)| r^2\left( 1-\mu \right) dA dr^* dt\\
&+C\int_{t_0}^{t_*}\int_{-\frac{t}{2}}^{\frac{t}{2}}\int_{\mathbb S^2} t|\psi \Box_g\psi| r^2\left( 1-\mu \right) dA dr^* dt+C\sum_{i=0}^{n}\int_{-\frac{t_i}{2}}^{\frac{t_i}{2}} t_i^2|\psi \Box_g\psi| \sqrt{1-\mu } dVol_{t_i}.\\
\end{split}
\end{equation*}
We now group this into three parts: firstly, the spacetime term that grows like $t^2$; secondly, the spacetime terms that grow like $t$; and finally, the boundary terms.\\
By Proposition 7, 8, Proposition 3.4, 3.5, Theorem 5.1 and 5.2,
\begin{equation*}
\begin{split}
&\int_{t_0}^{t_*}\int_{-\frac{t}{2}}^{\frac{t}{2}}\int_{\mathbb S^2} t^2|\psi \Box_g\left( \partial_t\psi\right)| r^2\left( 1-\mu \right) dA dr^* dt\\
\leq &C \int_{t_0}^{t_*}\int_{-\frac{t}{2}}^{\frac{t}{2}} \frac{t^2 \left( \log r\right)_+ \left( 1+|r^*|\right)|\psi|\left( |\partial_{r^*}\phi_t|+|\nabb\Omega\phi_t| \right)}{r^3}\\
\leq &C \left( \int_{t_0}^{t_*}\int_{-\frac{t}{2}}^{\frac{t}{2}} t^{2+\delta}r^{\delta}\frac{r^{1-\frac{\delta}{4}}\psi^2}{\left( 1+|r^*|\right)^4} dVol\right)^{\frac{1}{2}}\\
&\times \left( \sum_{i=0}^{n-1} t_i^{1-\frac{\delta}{2}} \left( \int_{t_i}^{t_{i+1}}\int_{-\frac{t}{2}}^{\frac{t}{2}} \frac{\left( 1+|r^*|\right)^6\left( \left( \partial_{r^*}\phi_t \right)^2+|\nabb\Omega\phi_t|^2\right) }{r^{7+\frac{\delta}{4}}} dVol\right)^{\frac{1}{2}}\right)\\
\leq &C \left( \int_{t_0}^{t_*}\int_{-\frac{t}{2}}^{\frac{t}{2}}t^{2+\delta}r^\delta \left(|K^{\tilde{X},w^{\tilde{X}}}\left( \psi \right)|+\sum_l K^{X_l,w^{X_l}}\left( \psi_l \right) \right) dVol\right)^{\frac{1}{2}}\left(\sum_{i=0}^{n-1} t_i^{-\frac{\delta}{2}}\right)\left(\sum_{m=0}^1 \sum_{k=0}^2  E_0\left( \partial_t^m \Omega^k \phi \right) \right)^{\frac{1}{2}}\\
\leq &C \left(\int_{t_0}^{t_*}\int_{-\frac{t}{2}}^{\frac{t}{2}}t^{2+2\delta}\left( |K^{\tilde{X},w^{\tilde{X}}}\left( \psi \right)|+\sum_l K^{X_l,w^{X_l}}\left( \psi_l \right)\right) dVol\right)^{\frac{1}{2}}\left( \sum_{m=0}^1 \sum_{k=0}^2 E_0\left(\partial_t^m \Omega^k \phi \right) \right)^{\frac{1}{2}}.
\end{split}
\end{equation*}
By Proposition 7, 8, Proposition 3.4, 3.5, Corollary 9.4 and 9.5,
\begin{equation*}
\begin{split}
&\int_{t_0}^{t_*}\int_{-\frac{t}{2}}^{\frac{t}{2}}\int_{\mathbb S^2} |tr^*\partial_{r^*}\psi\Box_g\psi| r^2\left( 1-\mu \right) dA dr^* dt
+\int_{t_0}^{t_*}\int_{-\frac{t}{2}}^{\frac{t}{2}}\int t|\psi \Box_g\psi| r^2\left( 1-\mu \right) dA dr^* dt\\
\leq &C\left(\int_{t_0}^{t_*}\int_{-\frac{t}{2}}^{\frac{t}{2}} \frac{t\left(\log r\right)_+\left( 1+|r^*|\right)\left( |r^*\partial_{r^*}\psi|+|\psi|\right)\left(|\partial_{r^*}\phi|+|\nabb\Omega\phi|\right)}{r^3}dVol\right)\\
\leq &C \left( \int_{t_0}^{t_*}\int_{-\frac{t}{2}}^{\frac{t}{2}}t^{\delta}r^{\delta}\left( \frac{r^{1-\frac{\delta}{4}}\psi^2}{\left( 1+|r^*|\right)^4}+\frac{\left( \partial_{r^*}\psi\right)^2}{r^{1+\frac{\delta}{4}}}\right) dVol\right)^{\frac{1}{2}}\\
&\times \left(\sum_{i=0}^{n-1} t_i^{1-\frac{\delta}{2}} \left(\int_{t_i}^{t_{i+1}}\int_{-\frac{t}{2}}^{\frac{t}{2}}\frac{\left( 1+|r^*|\right)^6\left( \left(\partial_{r^*}\phi \right)^2+|\nabb\Omega\phi|^2\right) }{r^{7+\frac{\delta}{4}}} dVol\right)^{\frac{1}{2}}\right)\\
\leq &C \left(\int_{t_0}^{t_*}\int_{-\frac{t}{2}}^{\frac{t}{2}}t^{\delta}r^{\delta}\left( |K^{\tilde{X},w^{\tilde{X}}}\left( \psi \right)|+\sum_l K^{X_l,w^{X_l}}\left( \psi_l \right)\right) dVol\right)^{\frac{1}{2}}\left( \sum_{k=0}^2E_0\left(\Omega^k \phi \right)\right)^{\frac{1}{2}}\left( \sum_{i=0}^{n-1} t_i^{-\frac{\delta}{2}}\right)\\
\leq &C \left(\int_{t_0}^{t_*}\int_{-\frac{t}{2}}^{\frac{t}{2}}t^{2\delta}\left( |K^{\tilde{X},w^{\tilde{X}}}\left( \psi \right)|+\displaystyle\sum_l K^{X_l,w^{X_l}}\left( \psi_l \right)\right) dVol\right)^{\frac{1}{2}}\left( \sum_{k=0}^2E_0\left(\Omega^k \phi \right)\right)^{\frac{1}{2}}.\\
\end{split}
\end{equation*}
By Proposition 3.1, 3.3 and Theorem 5.4,
\begin{equation*}
\begin{split}
&\sum_{i=0}^{n}\int_{-\frac{t_i}{2}}^{\frac{t_i}{2}} t_i^2|\psi \Box_g\psi| \sqrt{1-\mu } dVol_{t_i}\\
\leq &C \left(\sum_{i=0}^{n}t_i \left( \int J^{Z,w^Z}_\mu\left( \psi\right) n^\mu_{t_i} dVol_{t_i}\right)^{\frac{1}{2}}\left(\sum_{k=0}^1 \int_{-\frac{t_i}{2}}^{\frac{t_i}{2}} J^T_\mu\left( \Omega^k\phi\right) n^\mu_{t_i} dVol_{t_i}\right)^{\frac{1}{2}}\right)\\
\leq &C\sup_{t_0 \leq t \leq t_*} \left( \int J^{Z,w^Z}_{\mu}\left( \psi \right) n_{t}^{\mu} dVol_{t}\right)^{\frac{1}{2}}\left( \sum_{k=0}^1 E_0\left(\Omega^k\phi \right) \right)^{\frac{1}{2}}\left(\sum_{i=0}^{n} 1\right)\\
\leq &Ct_*^{\frac{\delta}{2}}\sup_{t_0 \leq t \leq t_*} \left(\int J^{Z,w^Z}_{\mu}\left(\psi \right) n_{t}^{\mu} dVol_{t}\right)^{\frac{1}{2}}\left(\sum_{k=0}^1 E_0\left(\Omega^k\phi \right) \right)^{\frac{1}{2}}.
\end{split}
\end{equation*}
These together give:
\begin{equation*}
\begin{split}
&|\int_{t_0}^{t_*}\int_{-\frac{t}{2}}^{\frac{t}{2}} \left(u^2\partial_u\psi+ v^2\partial_v\psi \right)\Box\psi dVol|\\
\leq &C \left(\int_{t_0}^{t_*}t^{2+2\delta}\left( |K^{\tilde{X},w^{\tilde{X}}}\left( \psi \right)|+\sum_l K^{X_l,w^{X_l}}\left(\psi_l \right)\right) dVol\right)^{\frac{1}{2}}\sum_{m=0}^1 \sum_{k=0}^2 E_0\left(\partial_t^m \Omega^k \phi \right)^{\frac{1}{2}}\\
&+Ct_*^{\frac{\delta}{2}}\sup_{t_0 \leq t \leq t_*} \left(\int J^{Z,w^Z}_{\mu}\left(\psi \right) n_{t}^{\mu} dVol_{t}\right)^{\frac{1}{2}}\left(\sum_{k=0}^1 E_0\left(\Omega^k\phi \right) \right)^{\frac{1}{2}}.
\end{split}
\end{equation*}
We finally look at the third region, $\displaystyle\bigcup_{i=0}^{n-1} \{t_i \leq t \leq t_{i+1}, r^* < -\frac{t_i}{2}\}$, for the second error term. By Proposition 3.1, 3.3 and Theorem 5.1.
\begin{equation*}
\begin{split}
&|\int_{t_0}^{t_*}\int_{-\infty}^{-\frac{t}{2}} \left(u^2\partial_u\psi+ v^2\partial_v\psi \right)\Box\psi dVol|\\
\leq &C\int_{t_0}^{t_*} \left(\int_{-\infty}^{-\frac{t}{2}}\left(u^2\left( \partial_u\psi \right)^2 + v^2\left(\partial_v\psi \right)^2\right) \left(1-\mu \right)^{-\frac{1}{2}}dVol_{t}\right)\\
&\times\left( \int_{-\infty}^{-\frac{t}{2}}\left( r^*\right)^2 \left( 1-\mu \right)^\frac{3}{2}\left(\left( \partial_{r^*}\phi \right)^2+|\nabb\Omega\phi|^2 \right)dVol_{t}\right)dt\\
\leq &C\int_{t_0}^{t_*} e^{-ct}\left(\int_{-\infty}^{-\frac{t}{2}} J^Z_\mu\left( \psi\right) n^{\mu}_t dVol_{t}\right)\left(\sum_{k=0}^1 \int_{-\infty}^{-\frac{t}{2}}J^T_\mu\left( \Omega^k\phi\right) n^{\mu}_t dVol_{t}\right)dt\\
\leq &C\sup_{t_0 \leq t \leq t_*} \left(\int J^{Z,w^Z}_{\mu}\left(\psi \right) n_{t}^{\mu} dVol_{t}\right)^{\frac{1}{2}}\left( E_0\left(\phi \right) \right)^{\frac{1}{2}}.
\end{split}
\end{equation*}
Therefore,
\begin{equation*}
\begin{split}
&\int J^{Z,w^Z}_{\mu}\left( \psi\right) n_{t_*}^{\mu} dVol_{\{ t=t_*\} }\\
\leq &C\int J^{Z,w^Z}_\mu\left( \psi\right) n^\mu_{t_0} dVol_{t_0}+\displaystyle\sum_{k=0}^1\int_{t_0}^{t_*}\int_{r_1^*}^{r_2^*}t\displaystyle\sum_l K^{X_l,w^{X_l}}\left( \Omega^k\psi_l \right) dVol\\
&+ C\left(\int_{t_0}^{t_*}\int_{-\frac{t}{2}}^{\frac{t}{2}} t^{2+2\delta}\left( |K^{\tilde{X},w^{\tilde{X}}}\left( \psi \right)|+\displaystyle\sum_l K^{X_l,w^{X_l}}\left(\psi_l \right)\right) dVol\right)^{\frac{1}{2}}\left(\sum_{m=0}^1 \sum_{k=0}^2 E_0\left(\partial_t^m \Omega^k \phi \right)\right)^{\frac{1}{2}}\\
& + C\sup_{t_0 \leq t \leq t_*} t_*^{\frac{\delta}{2}}\left(\int J^{Z,w^Z}_{\mu}\left(\psi \right) n_{t}^{\mu} dVol_{t}\right)^{\frac{1}{2}}\left(\sum_{k=0}^2 E_0\left( \Omega^k\phi \right) \right)^{\frac{1}{2}}.
\end{split}
\end{equation*}
The proof concludes with Lemma 12, taking
\begin{equation*}
\begin{split}
h_1\left(t_*\right)=&\displaystyle\sum_{k=0}^1\int_{t_0}^{t_*}\int_{r_1^*}^{r_2^*}t\displaystyle\sum_l K^{X_l,w^{X_l}}\left(\Omega^k\psi_l \right) dVol\\
&+ \left(\int_{t_0}^{t_*}\int_{-\frac{t}{2}}^{\frac{t}{2}} t^{2+2\delta}\left(|K^{\tilde{X},w^{\tilde{X}}}\left(\psi \right)|+\displaystyle\sum_l K^{X_l,w^{X_l}}\left(\psi_l \right)\right) dVol\right)^{\frac{1}{2}}\left(\sum_{m=0}^1 \sum_{k=0}^2 E_0\left(\partial_t^m \Omega^k \phi \right)\right)^{\frac{1}{2}},\\
h_2\left(t_*\right)=&t_*^{\delta}\left(\sum_{k=0}^2 E_0\left(\Omega^k\phi \right) \right).
\end{split}
\end{equation*}
We notice that $h_1\left(t\right)$ and $h_2\left(t\right)$ are increasing.
\end{proof}
We now combine Propositions 11, 13, 14 and 15 to prove Theorem 10.2. This will then imply the other parts of Theorem 10.
\begin{proposition}
$$\int J^{Z,w^Z}_\mu\left(\psi\right) n^{\mu}_{t_*} dVol_{t_*} \leq CE_1\left(\phi\right) t_*^{\delta}$$
\end{proposition}
\begin{proof}
We first show that $\int J^{Z,w^Z}_\mu\left(\psi\right) n^{\mu}_{t_*} dVol_{t_*}$ grows only like $t_*^{1+\delta}$.
Using Propositions 13 and 11,
\begin{equation*}
\begin{split}
&\int_{t_0}^{t_*} |K^{\tilde{X},w^{\tilde{X}}}\left(\psi\right)| + \sum_l K^{X_l,w^X}\left(\psi_l\right)dVol\\
\leq &C\left(\int J^T_{\mu}\left(\psi \right)n_{t_*}^{\mu}dVol_{t_*}+\int J^T_{\mu}\left(\psi \right)n_{t_0}^{\mu}dVol_{t_0}\right)+C t_0^{-2+\delta}\displaystyle\sum_{k=0}^2 E_0\left(\Omega^k \phi \right)\\
\leq &C\int J^T_{\mu}\left(\psi \right)n_{t_0}^{\mu}dVol_{t_0}+C \displaystyle\sum_{k=0}^2 E_0\left(\Omega^k \phi \right).
\end{split}
\end{equation*} 
Similarly, 
$$\int_{t_0}^{t_*} |K^{\tilde{X},w^{\tilde{X}}}\left(\Omega\psi\right)| + \sum_l K^{X_l,w^X}\left(\Omega\psi_l\right)dVol \leq C\int J^T_{\mu}\left(\Omega\psi \right)n_{t_0}^{\mu}dVol_{t_0}+C \displaystyle\sum_{k=0}^3 E_0\left(\Omega^k \phi \right).$$
Apply Proposition 15 to get 
\begin{equation*}
\begin{split}
&\int J^{Z,w^Z}_\mu\left( \psi\right) n^\mu_{t_*} dVol_{t_*}\\
\leq &C\int J^{Z,w^Z}_\mu\left( \psi\right) n^\mu_{t_0} dVol_{t_0}+Ct_*\displaystyle\sum_{k=0}^1\int_{t_0}^{t_*}\int_{r_1^*}^{r_2^*}\displaystyle\sum_l K^{X_l,w^{X_l}}\left( \Omega^k\psi_l \right) dVol\\
&+ Ct_*^{1+\delta}\left(\int_{t_0}^{t_*}\int_{-\frac{t}{2}}^{\frac{t}{2}} |K^{\tilde{X},w^{\tilde{X}}}\left(\psi \right)|+\displaystyle\sum_l K^{X_l,w^{X_l}}\left(\psi_l \right) dVol\right)^{\frac{1}{2}}\left(\sum_{m=0}^1 \sum_{k=0}^2 E_0\left(\partial_t^m \Omega^k \phi \right)\right)^{\frac{1}{2}}\\
&+ Ct_*^\delta \sum_{k=0}^2 E_0\left(\Omega^k\phi \right)\\
\leq &C\int J^{Z,w^Z}_\mu\left(\psi\right) n^\mu_{t_0} dVol_{t_0}+Ct_*\sum_{k=0}^1 \int J^T_{\mu}\left(\Omega^k\psi \right)n_{t_0}^{\mu}dVol_{t_0}+ Ct_*\sum_{k=0}^3 E_0\left(\Omega^k \phi \right)\\
&+ Ct_*^{1+\delta} \left(\int J^T_{\mu}\left(\psi \right)n_{t_0}^{\mu}dVol_{t_0}\right)^{\frac{1}{2}}\left(\sum_{m=0}^1 \sum_{k=0}^2 E_0\left(\partial_t^m \Omega^k \phi \right) \right)^{\frac{1}{2}}+ Ct_*^{1+\delta}\sum_{m=0}^1 \sum_{k=0}^2 E_0\left(\partial_t^m \Omega^k \phi \right)\\
\leq &C\left(\int J^{Z,w^Z}_\mu\left( \psi\right) n^\mu_{t_0} dVol_{t_0}+t_*^{1+\delta} \left(\sum_{k=0}^1\int J^T_{\mu}\left(\Omega^k\psi \right)n_{t_0}^{\mu}dVol_{t_0}+\sum_{m=0}^1 \sum_{k=0}^{3-m} E_0\left(\partial_t^m \Omega^k \phi \right)\right)\right).
\end{split}
\end{equation*}
We now have some control over $\int J^{Z,w^Z}_\mu\left(\psi\right) n^\mu_{t_*} dVol_{t_*}$ and we will use Proposition 14 to estimate the spacetime integral terms by integrating dyadically.\\
By Proposition 14 and 11,
\begin{equation*}
\begin{split}
&\int_{t_i}^{t_{i+1}}\int_{r_1^*}^{r_2^*}\displaystyle\sum_l K^{X_l,w^{X_l}}\left(\psi_l \right) dVol\\
\leq &C\left(t_i^{-2}\int J^Z_\mu\left(\psi\right) n^\mu_{t_i} dVol_{t_i}+t_i^{-2+\delta}\left(\sum_{m=0}^1 \sum_{k=0}^2 E_0\left(\partial_t^m \Omega^k \phi \right)\right)\right)\\
\leq &C\left(t_i^{-2}\int J^Z_\mu\left(\psi\right) n^\mu_{t_0} dVol_{t_0}+t_i^{-1+\delta}\left(\sum_{k=0}^1\int J^T_{\mu}\left(\Omega^k\psi \right)n_{t_0}^{\mu}dVol_{t_0}+\sum_{m=0}^1 \sum_{k=0}^{3-m} E_0\left(\partial_t^m \Omega^k \phi \right)\right)\right),
\end{split}
\end{equation*}
where here we have not kept track of the constant factor in front of $\delta$, but just note that it can be chosen to be arbitrarily small.\\
We can apply the same argument to $\int_{t_i}^{t_{i+1}}\int_{r_1^*}^{r_2^*}\displaystyle\sum_l K^{X_l,w^{X_l}}\left(\Omega\psi_l \right) dVol$ to get
\begin{equation*}
\begin{split}
&\int_{t_i}^{t_{i+1}}\int_{r_1^*}^{r_2^*}\displaystyle\sum_l K^{X_l,w^{X_l}}\left(\Omega\psi_l \right) dVol\\
\leq &C\left(t_i^{-2}\int J^Z_\mu\left(\Omega\psi\right) n^\mu_{t_0} dVol_{t_0}+t_i^{-1+\delta}\left(\sum_{k=0}^2\int J^T_{\mu}\left(\Omega^k\psi \right)n_{t_0}^{\mu}dVol_{t_0}+\sum_{m=0}^1 \sum_{k=0}^{4-m} E_0\left(\partial_t^m \Omega^k \phi \right)\right)\right).
\end{split}
\end{equation*}
This in turn provides more control on $\int J^{Z,w^Z}_\mu\left(\psi\right) n^\mu_{t_*} dVol_{t_*}$ by Proposition 15:
\begin{equation*}
\begin{split}
&\int J^{Z,w^Z}_\mu\left(\psi\right) n^\mu_{t_*} dVol_{t_*}\\
\leq &C\int J^{Z,w^Z}_\mu\left(\psi\right) n^\mu_{t_0} dVol_{t_0}+C\displaystyle\sum_{k=0}^1\sum_{i=0}^{n-1}t_i\int_{t_i}^{t_{i+1}}\int_{r_1^*}^{r_2^*}\displaystyle\sum_l K^{X_l,w^{X_l}}\left(\Omega^k\psi_l \right) dVol\\
&+ C\sum_{i=0}^{n-1}t_i^{1+\delta}\left( \int_{t_i}^{t_{i+1}}\int_{-\frac{t}{2}}^{\frac{t}{2}} |K^{\tilde{X},w^{\tilde{X}}}\left(\psi \right)|+\displaystyle\sum_l K^{X_l,w^{X_l}}\left(\psi_l \right) dVol\right)^{\frac{1}{2}}\left(\sum_{m=0}^1 \sum_{k=0}^2 E_0\left(\partial_t^m \Omega^k \phi \right)\right)^{\frac{1}{2}}\\
&+ Ct_*^\delta \sum_{k=0}^2 E_0\left(\Omega^k\phi \right)\\
\leq &C\sum_{k=0}^1\int J^{Z,w^Z}_\mu\left(\Omega^k\psi\right) n^\mu_{t_0} dVol_{t_0}+Ct_*^{\delta}\left(E_0\left(\psi\right)+\sum_{m=0}^1 \sum_{k=0}^{4-m} E_0\left(\partial_t^m \Omega^k \phi \right)\right)\\
& +C\sum_{i=0}^{n-1}t_i^{1+\delta}\left(\int_{t_i}^{t_{i+1}}\int_{-\frac{t}{2}}^{\frac{t}{2}} |K^{\tilde{X},w^{\tilde{X}}}\left(\psi \right)|+\displaystyle\sum_l K^{X_l,w^{X_l}}\left(\psi_l \right) dVol\right)^{\frac{1}{2}}\left(\sum_{m=0}^1 \sum_{k=0}^2 E_0\left(\partial_t^m \Omega^k \phi \right)\right)^{\frac{1}{2}}\\
\leq &C E_1\left(\phi\right)t_*^{\delta} +C\sum_{i=0}^{n-1}t_i^{1+\delta}\left( \int_{t_i}^{t_{i+1}}\int_{-\frac{t}{2}}^{\frac{t}{2}} |K^{\tilde{X},w^{\tilde{X}}}\left(\psi \right)|+\displaystyle\sum_l K^{X_l,w^{X_l}}\left(\psi_l \right) dVol\right)^{\frac{1}{2}}\left( \sum_{m=0}^1 \sum_{k=0}^2 E_0\left(\partial_t^m \Omega^k \phi \right)\right)^{\frac{1}{2}}.\\
\end{split}
\end{equation*}
Here, we recall that we have defined $E_1\left(\phi\right)=E_0\left(\psi\right)+\displaystyle\sum_{m=0}^1 \sum_{k=0}^{4-m} E_0\left(\partial_t^m \Omega^k \phi \right)$ in Section 1.3.\\
Clearly, we can replace $\delta$ by $\epsilon$ with a different constant $C$ which depends only on $\epsilon$:
\begin{equation}\label{Iorig}
\begin{split}
&\int J^{Z,w^Z}_\mu\left(\psi\right) n^\mu_{t_*} dVol_{t_*}\\
\leq &C E_1\left(\phi\right)t_*^{\epsilon}\\
&+C\sum_{i=0}^{n-1}t_i^{1+\epsilon}\left( \int_{t_i}^{t_{i+1}}\int_{-\frac{t}{2}}^{\frac{t}{2}} |K^{\tilde{X},w^{\tilde{X}}}\left(\psi \right)|+\displaystyle\sum_l K^{X_l,w^{X_l}}\left(\psi_l \right) dVol\right)^{\frac{1}{2}}\left( \sum_{m=0}^1 \sum_{k=0}^2 E_0\left(\partial_t^m \Omega^k \phi \right)\right)^{\frac{1}{2}}.\\
\end{split}
\end{equation}
Notice that at this point, the only term that exhibits more growth than expected is $$\sum_{i=0}^{n-1}t_i^{1+\epsilon}\left(\int_{t_i}^{t_{i+1}}\int_{-\frac{t}{2}}^{\frac{t}{2}} |K^{\tilde{X},w^{\tilde{X}}}\left(\psi \right)|+\displaystyle\sum_l K^{X_l,w^{X_l}}\left(\psi_l \right) dVol\right)^{\frac{1}{2}}\left(\sum_{m=0}^1 \sum_{k=0}^2 E_0\left(\partial_t^m \Omega^k \phi \right)\right)^{\frac{1}{2}}.$$ We will close the argument with a bootstrap.\\
For notational purposes, we define
\begin{equation*}
\begin{split}
I_{t_*}&=\int J^{Z,w^Z}_\mu\left(\psi\right) n^\mu_{t_*} dVol_{t_*},\\
II_{t_i}&=\int_{t_i}^{t_{i+1}}\int_{-\frac{t}{2}}^{\frac{t}{2}} |K^{\tilde{X},w^{\tilde{X}}}\left(\psi \right)|+\displaystyle\sum_l K^{X_l,w^{X_l}}\left(\psi_l \right) dVol,\\
\end{split}
\end{equation*}
(\ref{Iorig}) is equivalent to 
\begin{equation}\label{I}
I_{t_*} \leq C \left(t_*^\epsilon E_1\left(\phi\right)+\sum_{i=0}^{n-1}t_i^{1+\epsilon}\left(II_{t_i}\right)^{\frac{1}{2}}E_1^{\frac{1}{2}}\left(\phi\right)\right).
\end{equation}
On the other hand, Proposition 14 gives
\begin{equation}\label{II}
II_{t_i} \leq C \left(t_i^{-2} I_{t_i} + t_i^{-2+\delta}E_1\left(\phi\right)\right).
\end{equation}
Assume $I_t \leq A t^\delta E_1\left(\phi\right)$, where $A \geq 4C$. We want to show that $I_t \leq \frac{A}{2} t^\delta E_1\left(\phi\right)$ for all $t \geq \left(\frac{400 C^{\frac{1}{2}}\left(A^{\frac{1}{2}}+1\right)}{A}\right)^{\frac{4}{\delta}}$.
>From the assumption and (\ref{II}) we have 
$$II_{t_i} \leq C \left(A t_i^{-2+\delta} E_1\left(\phi\right) + t_i^{-2+\delta}E_1\left(\phi\right)\right).$$
Hence, by picking $\epsilon=\frac{\delta}{4}$ in (\ref{I}),
\begin{equation*}
\begin{split}
I_{t} \leq &C t^\frac{\delta}{4} E_1\left(\phi\right) + 100 C^{\frac{1}{2}}\left(A^{\frac{1}{2}}+1\right) t^{\frac{3\delta}{4}} E_1\left(\phi\right)\\
\leq &C t^\delta E_1\left(\phi\right) + 100 C^{\frac{1}{2}}\left(A^{\frac{1}{2}}+1\right) t^{\frac{3\delta}{4}} E_1\left(\phi\right)\\
\leq &\frac{A}{2} t^\delta E_1\left(\phi\right),
\end{split}
\end{equation*}
since $A \geq 4 C$ and $t\geq \left(\frac{400 C^{\frac{1}{2}}\left(A^{\frac{1}{2}}+1\right)}{A}\right)^{\frac{4}{\delta}}$.
\end{proof}
\begin{remark}
We would like to note that the number of derivatives used in the above argument is highly wasteful (we used a total of 8 derivatives!). Blue-Soffer constructed a vector field to control trapping with only $\epsilon$ derivatives \cite{BSo}. Therefore, we can, in principle, repeat the above argument noting the unnecessary loss of derivatives. The details, however, have not been pursued. It is known that with this vector field, Theorem 5 holds with $E_0\left(\phi\right)$ only having $1+\epsilon$ derivatives. Moreover, in Proposition 11-15, instead of having two $\Omega$ derivatives on $\phi$, we only need $1+\epsilon$ of them. We then go to the proof of Proposition 16 and note that it can be reproved assuming only that $\phi$ is in $H^{2+\epsilon}$ initially with suitable decay. 
\end{remark}
Now Theorem 10 follows directly from Propositions 16 and 14.
\section{Estimates near the Event Horizon}
In this section, we will use the vector field $Y$ to prove that any decay estimates that can be proved on a suitable compact set holds also along the horizon. We will also show that these estimates control enough derivatives to give pointwise decay estimate.
\begin{proposition}
Suppose $\int_{r_0^*}^{2\left(\left(1.2\right)r_0\right)^*-r_0^*} J^T_\mu\left(\phi\right) n^\mu_{t} dVol_{t} \leq B t^{-\alpha}$ for all $t$, for some $\alpha \geq 0$. Then 
$$\int_{\{r \leq r_0\}} J^Y_\mu\left(\phi\right) n^\mu_{\frac{1}{2}\left(t_*+r_0^*\right)} dVol_{\{v=\frac{1}{2}\left(t_*+r_0^*\right)\}} \leq C \left(B+\int J^N_\mu\left(\phi\right)n^\mu_{t_0}dVol_{t_0}\right) t_*^{-\alpha}.$$
\end{proposition}
\begin{remark}
The reader should think of $B$ as some energy quantity of the initial data. For example, as we will show later, the hypothesis of this Proposition holds for $B=C\displaystyle\sum_{m+k\leq1}E_1\left(\partial_t^m\Omega^k\phi\right)$.
\end{remark}
\begin{proof}
Apply Proposition 1 for $Y$, on the region $\mathcal{R}=\{\frac{1}{2}\left(t_1+r_0^*\right) \leq v \leq \frac{1}{2}\left(t_*+r_0^*\right), t \geq t_1\}$ as in the figure, we get
\begin{equation*}
\begin{split}
&\int_{\{t \geq t_1\}} J^Y_\mu\left(\phi\right) n^\mu_{\frac{1}{2}\left(t_*+r_0^*\right)} dVol_{\{v=\frac{1}{2}\left(t_*+r_0^*\right)\}}\\
&+ \int_{\{\frac{1}{2}\left(t_1+r_0^*\right) \leq v \leq \frac{1}{2}\left(t_*+r_0^*\right)\}} J^Y_\mu\left(\phi\right) n^\mu_{\infty} dVol_{\{u=\infty\}}+ \int_\mathcal{R} K^Y\left(\phi\right) dVol\\
 = &\int_{\{t \geq t_1\}} J^Y_\mu\left(\phi\right) n^\mu_{\frac{1}{2}\left(t_1+r_0^*\right)} dVol_{\{v=\frac{1}{2}\left(t_1+r_0^*\right)\}} + \int_{\{\frac{1}{2}\left(t_1+r_0^*\right) \leq v \leq \frac{1}{2}\left(t_*+r_0^*\right)\}} J^Y_\mu\left(\phi\right) n^\mu_{t_1} dVol_{t_1}.
\end{split}
\end{equation*}
\begin{figure}[htbp]
\begin{center}

\input{prop16.pstex_t}
 
\caption{The region $\mathcal{R}$}
\end{center}
\end{figure}\\

We split up the integrals into $r\leq r_0$ and $r>r_0$ parts.\\ Notice that the domain of integration of $\int_{\{t \geq t_1\}} J^Y_\mu\left(\phi\right) n^\mu_{\frac{1}{2}\left(t_1+r_0^*\right)} dVol_{\{v=\frac{1}{2}\left(t_1+r_0^*\right)\}}$ lies inside $\{r\leq r_0\}$. Moreover, we note that $\int_{\{\frac{1}{2}\left(t_1+r_0^*\right) \leq v \leq \frac{1}{2}\left(t_*+r_0^*\right)\}} J^Y_\mu\left(\phi\right) n^\mu_{\infty} dVol_{\{u=\infty\}}\geq 0$. Hence
\begin{equation*}
\begin{split}
&\int_{\{r \leq r_0\}} J^Y_\mu\left(\phi\right) n^\mu_{\frac{1}{2}\left(t_*+r_0^*\right)} dVol_{\{v=\frac{1}{2}\left(t_*+r_0^*\right)\}} + \int_{\mathcal{R}\cap \{r \leq r_0\}} K^Y\left(\phi\right) dVol\\
\leq &\int_{\{r \leq r_0\}} J^Y_\mu\left(\phi\right) n^\mu_{\frac{1}{2}\left(t_1+r_0^*\right)} dVol_{\{v=\frac{1}{2}\left(t_1+r_0^*\right)\}} + \int_{\{\frac{1}{2}\left(t_1+r_0^*\right) \leq v \leq \frac{1}{2}\left(t_*+r_0^*\right)\}} J^Y_\mu\left(\phi\right) n^\mu_{t_1} dVol_{t_1}\\
&+\int_{\{r \geq r_0\}\cap\{t \geq t_1\}} J^Y_\mu\left(\phi\right) n^\mu_{\frac{1}{2}\left(t_*+r_0^*\right)} dVol_{\{v=\frac{1}{2}\left(t_*+r_0^*\right)\}}+\int_{\mathcal{R}\cap \{r \geq r_0\}} |K^Y\left(\phi\right)| dVol.
\end{split}
\end{equation*}
We estimate three terms on the right hand side. Notice that $Y$ is constructed to be supported in $\{r \leq \left(1.2\right)r_0\}$.
\begin{equation*}
\begin{split}
&\int_{\{\frac{1}{2}\left(t_1+r_0^*\right) \leq v \leq \frac{1}{2}\left(t_*+r_0^*\right)\}} J^Y_\mu\left(\phi\right) n^\mu_{t_1} dVol_{t_1}\\
\leq &C \int_{r_0^*}^{\left(\left(1.2\right)r_0\right)^*} J^T_\mu\left(\phi\right) n^\mu_{t_1} dVol_{t_1}\\
\leq &CBt_1^{-\alpha},
\end{split}
\end{equation*}
For the second and third term, we first use the compact support of $Y$ and then apply the conservation law associated to the Killing vector $T$. 
\begin{equation*}
\begin{split}
&\int_{\{r \geq r_0\}\cap\{u \geq \frac{1}{2}\left( t_1-r_0^*\right)\}} J^Y_\mu \left( \phi\right) n^\mu_{\frac{1}{2}\left( t_*+r_0^*\right)} dVol_{\{ v=\frac{1}{2}\left( t_*+r_0^*\right) \}}\\
\leq &C \int_{ \{ r_0 \leq r \leq \left(1.2\right)r_0 \}\cap\{u \geq \frac{1}{2}\left(t_1-r_0^*\right)\}} J^T_\mu\left(\phi\right) n^\mu_{\frac{1}{2}\left(t_*+r_0^*\right)} dVol_{\{v=\frac{1}{2}\left(t_*+r_0^*\right)\}}\\
\leq &C \int_{\{r_0^* \leq r^* \leq 2\left(\left(1.2\right)r_0\right)^*-r_0^* \}} J^T_\mu\left( \phi\right) n^\mu_{t_1} dVol_{t_*}\\
\leq &CBt_*^{-\alpha}\\
\leq &CBt_1^{-\alpha},
\end{split}
\end{equation*}
\begin{equation*}
\begin{split}
&\int_{\mathcal{R}\cap \{r \geq r_0\}} |K^Y\left( \phi\right)| dVol\\
\leq &C \int_{t_1}^{t_*} \int_{r_0^*}^{\left(\left(1.2\right)r_0\right)*} |K^Y\left(\phi\right)| dVol\\
\leq &C \int_{t_1}^{t_*} \int_{r_0^*}^{\left(\left(1.2\right)r_0\right)*} J^T_{\mu}\left( \phi\right)n^{\mu}_{t} dVol_{t} dt\\
\leq &CB \int_{t_1}^{t_*} t^{-\alpha} dt\\
\leq &CB \left( t_*-t_1\right)t_1^{-\alpha},
\end{split}
\end{equation*}
since $\alpha\geq0$.\\
Write $f\left(t\right)=\int_{\{r\leq r_0\}} J^Y_{\mu}\left(\phi\right) n^{\mu }_{\frac{1}{2}\left(t+r_0^*\right)} dVol_{\{v=\frac{1}{2}\left(t+r_0^*\right)\}}$.
Then we have
$$f\left(t_*\right)+\int_{t_1}^{t_*} f\left(\tau\right) d\tau \leq C\left(f\left(t_1\right)+B \max\{t_*-t_1, 1\} t_1^{-\alpha}\right).$$
We take $C$ to be fixed from this point on. We clearly can assume the $C>1$.\\
>From this, we will prove the Proposition by a bootstrap argument. Assume $f\left(t\right)\leq At^{-\alpha}$ for some large $A$ that is to be determined. We want to show that $f\left(t\right) \leq \frac{A}{2}t^{-\alpha}$.\\
Let $t_1=t_*-8C^2$. Since we are only concerned with $t_*$ large, we assume without loss of generality that $t_* > 8\left(1-2^{-\frac{1}{\alpha}}\right)^{-1}C^2$ so that $t_*<2^{\frac{1}{\alpha}}t_1$. Then
\begin{equation*}
\begin{split}
f\left(t_*\right)+\int_{t_*-8C^2}^{t_*} f\left(\tau\right) d\tau \leq &C\left(At_1^{-\alpha}+8C^2Bt_1^{-\alpha}\right)\\
\leq &2C\left(A+8C^2B\right)t_*^{-\alpha}.
\end{split}
\end{equation*}
There exists $\tilde{t}$ with $t_*-8C^2 \leq \tilde{t} \leq t_*$ such that 
\begin{equation*}
\begin{split}
f\left(\tilde{t}\right) \leq &\frac{1}{8C^2}\int_{t_*-8C^2}^{t_*} f\left(\tau\right) d\tau\\
\leq &\frac{\left(A+8C^2B\right)}{4C}t_*^{-\alpha}.
\end{split}
\end{equation*}
Now we let $t_1=\tilde{t}$. Notice that $t_*<2^{\frac{1}{\alpha}}\tilde{t}$. Then
\begin{equation*}
\begin{split}
f\left(t_*\right)+\int_{\tilde{t}}^{t_*} f\left(\tau\right) d\tau \leq &C\left(f\left(\tilde{t}\right)+8C^2B\tilde{t}^{-\alpha}\right)\\
\leq &\frac{A}{4}t_*^{-\alpha}+2C^2Bt_*^{-\alpha}+16C^3Bt_*^{-\alpha}\\
\leq &\frac{A}{2}t_*^{-\alpha},
\end{split}
\end{equation*}
if $A\geq 72C^3B$.\\
Of course to have $f\left(t\right)\leq At^{-\alpha}$ for all $t$, we also need it to hold initially, i.e., $A\geq f\left(t_0\right)$. Therefore, we have
$$\int_{\{r \leq r_0\}} J^Y_\mu\left(\phi\right) n^\mu_{\frac{1}{2}\left(t_*+r_0^*\right)} dVol_{\{v=\frac{1}{2}\left(t_*+r_0^*\right)\}} \leq C \left(B+\int J^N_\mu\left(\phi\right)n^\mu_{t_0}dVol_{t_0}\right) t_*^{-\alpha},$$
where $C$ is a universal constant different from the one above.
\end{proof}
Using Proposition 17, we claim that a similar estimate holds on $t$-slices.
\begin{proposition}
Suppose $\int_{r_0^*}^{2\left(\left(1.2\right)r_0\right)^*-r_0^*} J^T_\mu\left(\phi\right) n^\mu_{t} dVol_{t} \leq B t^{-\alpha}$ for all $t$, for some $\alpha >1$. Then 
$$\int_{\{v_* \leq v \leq v_*+1\}} J^Y_\mu\left(\phi\right) n^\mu_{\tau} dVol_{\tau} \leq C \left(B+\int J^N_\mu\left(\phi\right)n^\mu_{t_0}dVol_{t_0}\right) v_*^{-\alpha},$$
for $v_* \geq 1$.
\end{proposition}
\begin{proof}
We prove this using the conservation law for $Y$ on the region $\mathcal{R}=\{v_* \leq v \leq v_*+1, 2v_*-r_0^* \leq t \leq \tau\}$.
\begin{equation*}
\begin{split}
&\int_{\{t \geq 2v_*-r_0^*\}} J^Y_\mu\left(\phi\right) n^\mu_{v_*+1} dVol_{\{v=v_*+1\}}\\
&+ \int_{\{v_* \leq v \leq v_*+1\}} J^Y_\mu\left(\phi\right) n^\mu_{\tau} dVol_{\tau}+ \int_\mathcal{R} K^Y\left(\phi\right) dVol\\
 = &\int_{\{t \geq 2v_*-r_0^*\}} J^Y_\mu\left(\phi\right) n^\mu_{v_*} dVol_{v_*} + \int_{\{v_* \leq v \leq v_*+1\}} J^Y_\mu\left(\phi\right) n^\mu_{2v_*-r_0^*} dVol_{\{t=2v_*-r_0^*\}}.
\end{split}
\end{equation*}
We split up the integrals into $r\leq r_0$ and $r>r_0$ parts.\\ Notice that the domain of integration of $\int_{\{t \geq 2v_*-r_0^*\}} J^Y_\mu\left(\phi\right) n^\mu_{v_*} dVol_{v_*}$ lies inside $\{r\leq r_0\}$. Notice also that 
$$\int_{\{t \geq 2v_*-r_0^*\}\cap\{r\leq r_0\}} J^Y_\mu\left(\phi\right) n^\mu_{v_*+1} dVol_{\{v=v_*+1\}}+\int_{\mathcal{R}\cap\{r\leq r_0\}} K^Y\left(\phi\right) dVol\geq 0. $$ Hence
\begin{equation*}
\begin{split}
&\int_{\{v_* \leq v \leq v_*+1\}} J^Y_\mu\left(\phi\right) n^\mu_{\tau} dVol_{\tau}\\
\leq &\int_{\{t \geq 2v_*-r_0^*\}} J^Y_\mu\left(\phi\right) n^\mu_{v_*} dVol_{v_*} + \int_{\{v_* \leq v \leq v_*+1\}} J^Y_\mu\left(\phi\right) n^\mu_{2v_*-r_0^*} dVol_{\{t=2v_*-r_0^*\}}\\
&+\int_{\{t \geq 2v_*-r_0^*\}\cap\{r\geq r_0\}} J^Y_\mu\left(\phi\right) n^\mu_{v_*+1} dVol_{\{v=v_*+1\}}+\int_{\mathcal{R}\cap \{r \geq r_0\}} |K^Y\left(\phi\right)| dVol.
\end{split}
\end{equation*}
We show that each term has the correct bound. The first term is bounded using Proposition 17,
\begin{equation*}
\begin{split}
&\int_{\{t \geq 2v_*-r_0^*\}} J^Y_\mu\left(\phi\right) n^\mu_{v_*} dVol_{v_*}\\
= &\int_{\{r \leq r_0\}} J^Y_\mu\left(\phi\right) n^\mu_{v_*} dVol_{v_*}\\
\leq &C \left(B+\int J^N_\mu\left(\phi\right)n^\mu_{t_0}dVol_{t_0}\right)\left(2v_*-r_0^*\right)^{-\alpha}\\
\leq &C \left(B+\int J^N_\mu\left(\phi\right)n^\mu_{t_0}dVol_{t_0}\right)v_*^{-\alpha}.
\end{split}
\end{equation*}
The second term is controlled by assumption
\begin{equation*}
\begin{split}
&\int_{\{v_* \leq v \leq v_*+1\}} J^Y_\mu\left(\phi\right) n^\mu_{2v_*-r_0^*} dVol_{\{t=2v_*-r_0^*\}}\\
\leq &C\int_{\{v_* \leq v \leq v_*+1\}} J^T_\mu\left(\phi\right) n^\mu_{2v_*-r_0^*} dVol_{\{t=2v_*-r_0^*\}}\\
\leq &CB\left(2v_*-r_0^*\right)^{-\alpha}\\
\leq &CBv_*^{-\alpha}.
\end{split}
\end{equation*}
The last two terms are bounded by noting that $Y$ is supported in $r \leq \left(1.2\right)r_0$. The details are identical to the proof of Proposition 17.
Therefore,
$$\int_{\{v_* \leq v \leq v_*+1\}} J^Y_\mu\left(\phi\right) n^\mu_{\tau} dVol_{\tau}C \left(B+\int J^N_\mu\left(\phi\right)n^\mu_{t_0}dVol_{t_0}\right)v_*^{-\alpha}.$$
\end{proof}
This, and Sobolev embedding, is sufficient to show pointwise decay of the derivatives of $\phi$ along the horizon.
We show further that if on a compact set, we have both energy decay and $L^2$ decay, then we have pointwise decay along the event horizon. More precisely, we have
\begin{proposition}
There exist $\tilde{r}$ very close to $2M$ such that if 
$$\sum_{m=0}^1\sum_{k=0}^{3-m} \int_{\tilde{r}^*}^{2\left(\left(1.2\right)r_0\right)^*-r_0^*} \left(J^T_\mu\left(\partial_t^m\Omega^k\phi\right) n^\mu_{t}+\phi^2\right) dVol_{t} \leq B t^{-\alpha}$$ 
for all $t$, for some $\alpha \geq 0$, then
$$|\phi\left(v_*,r\right)|^2 \leq C \left(B+\sum_{k=0}^2\int J^N_\mu\left(\Omega^k\phi\right)n^\mu_{t_0}dVol_{t_0}\right)v_*^{-\alpha},$$
$$|\partial_{r^*}\phi\left(v_*,r\right)|^2 \leq C \left(B+\sum_{m=0}^{1}\sum_{k=0}^{3-m} \int J^N_\mu\left(\partial_t^m\Omega^k\phi\right)n^\mu_{t_0}dVol_{t_0}\right)v_*^{-\alpha},$$
for $v_* \geq 1$, $r \leq \tilde{r}$.
\end{proposition}
\begin{proof}
We first take $\tilde{r}$ to be small enough to apply $Y$, i.e., $\tilde{r} <r_0$. The exact condition on $\tilde{r}$ will be determined later.\\
For decay of $\phi\left(v_*,r\right)$, we want to show that on any time-slice, say $t=\tau$, 
$$\sum_{k=0}^2 \int_{\{v_* \leq v \leq v_*+1\}} \left(|\nabb^k\phi|^2+|\nabb^k\partial_{r^*}\phi|^2\right)dA dr^*_{\tau} \leq C v_*^{-\alpha}.$$
For decay of $\partial_{r^*}\phi\left(v_*,r\right)$, we want to show that on any time-slice, say $t=\tau$, 
$$\sum_{k=0}^2 \int_{\{v_* \leq v \leq v_*+1\}} \left(|\nabb^k\partial_{r^*}\phi|^2+|\nabb^k\partial_{r^*}^2\phi|^2\right)dA dr^*_{\tau} \leq C v_*^{-\alpha}.$$
Proposition 18 gives
$$\int_{\{v_* \leq v \leq v_*+1\}} \left(\left(\partial_{r^*}\phi\right)^2+|\nabb\phi|^2\right)dA dr^*_{\tau} \leq C \left(B+\int J^N_\mu\left(\phi\right)n^\mu_{t_0}dVol_{t_0}\right) v_*^{-\alpha}.$$
After commuting with an appropriate number of $\Omega$, Proposition 18 gives
$$\int_{\{v_* \leq v \leq v_*+1\}} \left(|\nabb\partial_{r^*}\phi|^2+|\nabb\nabb\partial_{r^*}\phi|^2\right)dA dr^*_{\tau} \leq C \left(B+\sum_{k=1}^2 \int J^N_\mu\left(\Omega^k\phi\right)n^\mu_{t_0}dVol_{t_0}\right) v_*^{-\alpha}.$$
After commuting with $\partial_t$ and using the equation, Proposition 18 gives
\begin{equation*}
 \begin{split}
&\int_{\{v_* \leq v \leq v_*+1\}} \left(\left(\partial_{r^*}^2\phi\right)^2+|\nabb\partial_{r^*}^2\phi|^2+|\nabb^2\partial_{r^*}^2\phi|^2\right)dA dr^*_{\tau}\\
\leq &C \left( B+\sum_{m=0}^1\sum_{k=0}^{3-m} \int J^N_\mu\left(\partial_t^m\Omega^k\phi\right)n^\mu_{t_0}dVol_{t_0}\right) v_*^{-\alpha}. 
 \end{split}
\end{equation*}
Therefore, it remains to show 
$$\int_{\{v_* \leq v \leq v_*+1\}} \phi^2 dA dr^*_{\tau} \leq C v_*^{-\alpha}.$$
We rewrite
$$\int_{\{v_* \leq v \leq v_*+1\}} \phi^2 dA dr^*_{\tau}=\int_{2v_*-\tau}^{2v_*-\tau+2} \phi^2 dA dr^*_{\tau}.$$
To achieve decay, we integrate in the $u$-direction and use the estimates we have on the compact set.
\begin{equation*}
\begin{split}
&\int_{2v_*-\tau}^{2v_*-\tau+2} \phi^2 dA dr^*_{\tau}\\
\leq &\int_{\tilde{r}^*}^{\tilde{r}^*+2} \phi^2 dA dr^*_{\{t=2v_*-\tilde{r}^*\}}+ \int_{2v_*-\tilde{r}^*}^{\tau} \int_{2v_*-t}^{2v_*-t+2} \phi\left(\partial_u\phi\right)dA dr^*dt\\
\leq &\int_{\tilde{r}^*}^{\tilde{r}^*+2} \phi^2 dA dr^*_{\{t=2v_*-\tilde{r}^*\}}+ \int_{2v_*-\tilde{r}^*}^{\tau} \int_{2v_*-t}^{2v_*-t+2} \frac{\left(\partial_u\phi\right)^2}{1-\mu }dA dr^*dt+ \int_{2v_*-\tilde{r}^*}^{\tau} \int_{2v_*-t}^{2v_*-t+2} \phi^2 \left(1-\mu \right)dA dr^*dt\\
\leq &B \left(2v_*-\tilde{r}^*\right)^{-\alpha} + \int_{2v_*-\tilde{r}}^{\tau} \int_{2v_*-t}^{2v_*-t+2} K^Y\left(\phi\right) dVol + \int_{2v_*-\tilde{r}^*}^{\tau} \int_{2v_*-t}^{2v_*-t+2} \phi^2 \left(1-\mu \right)dA dr^*dt.
\end{split}
\end{equation*}
Using the conservation law for $Y$, and controlling all the terms on the region $\{r \geq \tilde{r}\}$ with the assumption, we have
\begin{equation*}
\begin{split}
&\int_{2v_*-\tilde{r}^*}^{\tau} \int_{2v_*-t}^{2v_*-t+2} K^Y\left(\phi\right) dVol\\
\leq & \int_{\{r \leq \tilde{r}\}} J^N_\mu\left(\phi\right) n^\mu_{v_*} dVol_{v_*}+ C B v_*^{-\alpha}\\
\leq & C\left(B+\int J^N_\mu\left(\phi\right)n^\mu_{t_0}dVol_{t_0}\right) \left(v_*-\tilde{r}^*\right)^{-\alpha}+ C B v_*^{-\alpha},
\end{split}
\end{equation*}
where in the last step we have used Proposition 17.\\
Therefore, 
$$\int_{2v_*-\tau}^{2v_*-\tau+2} \phi^2 dA dr^*_{\tau} \leq C \left(B+\int J^N_\mu\left(\phi\right)n^\mu_{t_0}dVol_{t_0}\right) v_*^{-\alpha}+\int_{2v_*-\tilde{r}^*}^{\tau} \int_{2v_*-t}^{2v_*-t+2} \phi^2 \left(1-\mu \right)dA dr^*dt.$$
The decay from the last term comes from the exponentially decaying (towards $r^*=-\infty$) factor $\left(1-\mu \right)$. To use this decay, we use a bootstrap argument. Assume the decay $\int_{v_*-t}^{v_*-t+1} \phi^2 dA dr^*_{\{t=t\}} \leq Av_*^{-\alpha},$ independent of $t$ (Note that we can do this initially (in $v$) independent of $t$ because after we fix $v$, the region of integration is a bounded set of the manifold. The apparent infiniteness is just an artifact of the choice of coordinates). We want to show that $\int_{v_*-t}^{v_*-t+1} \phi^2 dA dr^*_{\{t=t\}} \leq \frac{A}{2}v_*^{-\alpha}$.\\
>From the above, and using that $\left(1-\mu \right) \leq Ce^{cr^*}$ we have 
\begin{equation*}
\begin{split}
&\int_{2v_*-\tau}^{2v_*-\tau+2} \phi^2 dA dr^*_{\tau}\\ 
\leq &C \left(B+\int J^N_\mu\left(\phi\right)n^\mu_{t_0}dVol_{t_0}\right) v_*^{-\alpha}+\int_{v_*-\tilde{r}^*}^{\infty} ACe^{c\left(2v_*-t+2\right)}v_*^{-\alpha} dt\\
\leq &C \left(B+\int J^N_\mu\left(\phi\right)n^\mu_{t_0}dVol_{t_0}\right) v_*^{-\alpha}+c^{-1}ACe^{c\left(\tilde{r}^*+2\right)}v_*^{-\alpha}\\
\leq & \frac{A}{2}v_*^{-\alpha},
\end{split}
\end{equation*}
if we choose $A \geq 4 C \left(B+\int J^N_\mu\left(\phi\right)n^\mu_{t_0}dVol_{t_0}\right)$ and $\tilde{r}^* \leq -2 + \frac{1}{c}\log{\frac{c}{4C}}$.
\end{proof}
\section{Proof of the Main Theorem}
\subsection{Improved Decay for $\phi$}
To prove Main Theorem 1, we proceed in two steps. First, we show that for every $t_*$, there exist $t_1<t_*$, $t \sim t_*$ such that a weighted $L^2$-norm of $\phi$ on the slice $\{t=t_1\}$ has the desired decay of $t_*^{-3+\delta}$. We then use the estimates for $\psi$ to upgrade this to decay estimates for a weighted $L^2$-norm of $\phi$ on the slice $\{t=t_*\}$.\\
We first set up some notation. Fix $r_1^*, r_2^*$. These are the $r_1^*$ and $r_2^*$ in the statement of Main Theorem 1. In other words, we would like to prove a decay estimate on the fixed compact region $r_1^*\leq r \leq r_2^*$. Let $t_* \geq 2\left(|r_1^*|+|r_2^*|\right)$ be the time slice on which we want to show the decay estimate. Let $\tilde{t}= \left(1.1\right)^{-1}t_*$ and $\mathcal{P}=\{\tilde{t} \leq t \leq t_*, r_1^*-t_*+t\leq r^* \leq r_2^*+t_*-t\}$.
\begin{proposition}
There exist a $t_1$ with $\tilde{t} \leq t_1 \leq t_*$ such that 
$$\int_{\mathcal{P}\cap \{t=t_1\}} \left(\phi^2+\left(\partial_{r^*}\phi\right)^2\right) r^{-2}\left(1-\mu\right)^2 dA dr^* \leq C t_*^{-3}E_0\left(\phi\right). $$
\begin{figure}[htbp]
\begin{center}
 
\input{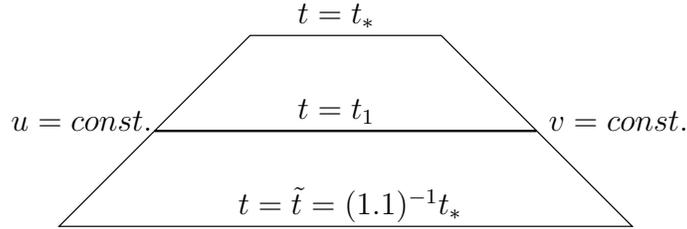}
 
\caption{The region $\mathcal{P}$}
\end{center}
\end{figure}
\end{proposition}
\begin{proof}
By Theorem 5, 
\begin{equation*}
\begin{split}
& \int_\mathcal{P} \left( \phi^2+\left( \partial_{r^*}\phi\right)^2\right) r^{-2}\left( 1-\mu\right)^2 dA dr^* dt\\
\leq &C \int_\mathcal{P} |K^{\tilde{X},w^{\tilde{X}}}\left( \phi \right)|+\displaystyle\sum_l K^{X_l,w^{X_l}}\left( \phi_l \right) dVol\\
\leq &C\tilde{t}^{-2} E_0\left( \phi\right)\\
\leq &Ct_*^{-2} E_0\left( \phi\right).
\end{split}
\end{equation*}
Now take $t_1$ such that 
$$\int_{\mathcal{P}\cap \{t=t_1\}} \left(\phi^2+\left(\partial_{r^*}\phi\right)^2\right) r^{-2}\left(1-\mu\right)^2 dA dr^*=\inf_{\tilde{t} \leq t \leq t_*} \int_{\mathcal{P}\cap \{t=t\}} \left(\phi^2+\left(\partial_{r^*}\phi\right)^2\right) r^{-2}\left(1-\mu\right)^2 dA dr^* ,$$
which exists since we are taking the infimum over a compact interval, and note that 
\begin{equation*}
\begin{split}
\inf_{\tilde{t} \leq t \leq t_*} \int_{\mathcal{P}\cap \{t=t\}} \left(\phi^2+\left(\partial_{r^*}\phi\right)^2\right) r^{-2}\left(1-\mu\right)^2 dA dr^* \leq &\left(t_*-\tilde{t}\right)^{-1}\int_\mathcal{P} \left(\phi^2+\left(\partial_{r^*}\phi\right)^2\right) r^{-2}\left(1-\mu\right)^2 dA dr^* dt\\
\leq &C t_*^{-1}\int_\mathcal{P} \left(\phi^2+\left(\partial_{r^*}\phi\right)^2\right) r^{-2}\left(1-\mu\right)^2 dA dr^* dt\\
\leq &C t_*^{-3}E_0\left(\phi\right).
\end{split}
\end{equation*}
\end{proof}
To upgrade this to an estimate for a generic $t$, we make two observations about $S$. Firstly, $S$ is timelike away from the event horizon. Secondly, $S$ has a weight $\sim t$. We can therefore integrate from the "good slice" $t=t_1$ to the slice $t=t_*$ and get the same decay estimate. This is done using integration by parts in the following Proposition. We prove a more general form but the reader should keep in mind that we will use $f=\phi^2+\left(\partial_{r^*}\phi\right)^2$, $g=r^{-2}\left(1-\mu\right)^2$.
\begin{proposition}
Let $f=f\left(r^*,t,\omega\in\mathbb S^2\right)$, $g=g\left(r^*\right)$, $\mathcal{P}=\{t_1 \leq t \leq t_*, r_1^*-t_*+t\leq r^* \leq r_2^*+t_*-t\}$. Then
\begin{equation*}
\begin{split}
&t_*\int_{\mathcal{P}\cap \{ t=t_*\}} f g dA dr^* +\int_{\mathcal{P}\cap \{v=\frac{1}{2}\left(t_*+r_2^*\right) \}} v f g dA dr^*+\int_{\mathcal{P}\cap \{u=\frac{1}{2}\left(t_*-r_1^* \right)\}} u f g dA dr^*\\
= &t_1\int_{\mathcal{P}\cap \{ t=t_1\}} f g dA dr^* + 2\int_\mathcal{P} f g dA dt dr^* +\int_\mathcal{P} r^*fg'dAdtdr^*+\int_\mathcal{P} \left(Sf\right)gdA dtdr^*.
\end{split}
\end{equation*}
\end{proposition}
\begin{proof}
We change to the variables $u,v$ and integrate by parts,
\begin{equation*}
\begin{split}
&\int_\mathcal{P} v\left(\partial_v f\right)g dA dt dr^*\\
=&\int_{\frac{1}{2}\left(t_*-r_2^*\right)}^{\frac{1}{2}\left(t_*-r_1^*\right)}\int_{t_1-u}^{t_*-u} v\left(\partial_v f\right)g dA dv du + \int_{\frac{1}{2}\left(2t_1-t_*-r_2^*\right)}^{\frac{1}{2}\left(t_*-r_2^*\right)}\int_{t_1-u}^{\frac{1}{2}\left(t_*+r_2^*\right)} v\left(\partial_v f\right)g dAdv du\\
=&-\int_\mathcal{P} fg dAdt dr^*-\int_\mathcal{P} v f \partial_{r^*}g dAdt dr^* + \int_{\mathcal{P}\cap \{ t=t_*\}} v f g dAdr^*+\int_{\mathcal{P}\cap \{v=t_*+r_2^* \}} v f g dAdr^*\\
&-\int_{\mathcal{P}\cap \{ t=t_1\}} v f g dAdr^*,
\end{split}
\end{equation*}
\begin{equation*}
\begin{split}
&\int_\mathcal{P} u\left(\partial_u f\right)g dAdt dr^*\\
=&\int_{\frac{1}{2}\left(2t_1-t_*+r_1^*\right)}^{\frac{1}{2}\left(t_*+r_1^*\right)}\int_{t_1-v}^{\frac{1}{2}\left(t_*-r_1^*\right)} u\left(\partial_u f\right)g dAdu dv + \int_{\frac{1}{2}\left(t_*+r_1^*\right)}^{\frac{1}{2}\left(t_*+r_2^*\right)}\int_{t_1-v}^{t_*-v} u\left(\partial_u f\right)g dAdu dv\\
=&-\int_\mathcal{P} fg dAdt dr^*+\int_\mathcal{P} u f \partial_{r^*}g dAdt dr^* + \int_{\mathcal{P}\cap \{ t=t_*\}} u f g dAdr^*+\int_{\mathcal{P}\cap \{u=t_*-r_1^* \}} u f g dAdr^*\\
&-\int_{\mathcal{P}\cap \{ t=t_1\}} u f g dAdr^*.\\
\end{split}
\end{equation*}
The proposition is proved by adding these two equations.
\end{proof}
To prove the main theorem, we use the above identity using $f=\phi^2+\left(\partial_{r^*}\phi\right)^2$, $g=r^{-2}\left(1-\mu\right)^2$.\\
We notice that since $f, g \geq 0$ by definition and $u, v \geq 0$ in $\mathcal{P}=\{t_1 \leq t \leq t_*, r_1^*-t_*+t\leq r^* \leq r_2^*+t_*-t\}$.\\
Therefore, 
$$\int_{\mathcal{P}\cap \{v=\frac{1}{2}\left(t_*+r_2^*\right) \}} v f g dAdr^*+\int_{\mathcal{P}\cap \{u=\frac{1}{2}\left(t_*-r_1^*\right) \}} u f g dAdr^* \geq 0.$$
Thus Proposition 21 would imply
\begin{equation*}
\begin{split}
&t_* \int_{\mathcal{P}\cap \{ t=t_*\}} \left( \phi^2+\left( \partial_{r^*}\phi \right)^2\right) r^{-2}\left( 1-\mu\right)^2 dAdr^*\\
\leq &t_1\int_{\mathcal{P}\cap \{ t=t_1\}} \left( \phi^2+\left( \partial_{r^*}\phi \right)^2\right) r^{-2}\left( 1-\mu\right)^2 dAdr^* + 2\int_\mathcal{P} \left( \phi^2+\left( \partial_{r^*}\phi \right)^2\right) r^{-2}\left( 1-\mu\right)^2 dA dr^* dt\\
&+2\int_\mathcal{P} |r^*\left( \phi^2+\left( \partial_{r^*}\phi \right)^2\right) \left( r^{-2}\left(1-\mu \right)^2\right)'|dtdr^* +\int_P \left(|\psi\phi|+|\partial_{r^*}\psi\partial_{r^*}\phi|\right) r^{-2}\left(1-\mu\right)^2dAdtdr^*\\
\leq &C t_*^{-2} E_0\left(\phi\right) + \int_\mathcal{P} \left(\psi^2+\left(\partial_{r^*}\psi\right)^2 \right) r^{-2}\left( 1-\mu\right)^2dAdtdr^* + \int_\mathcal{P} \left(\phi^2+\left(\partial_{r^*}\phi\right)^2\right) r^{-2}\left(1-\mu\right)^2dAdtdr^*\\
\leq &C t_*^{-2+\delta} E_1\left(\phi\right),
\end{split}
\end{equation*}
where we have used Proposition 20 at the second to last step and Theorems 5 and 10 at the last step.\\
Therefore,
$$\int_{r_1^*}^{r_2^*} \left(\phi\left(t_*\right)^2+\left(\partial_{r^*}\phi\left(t_*\right)\right)^2\right) dAdr^* \leq C t_*^{-3+\delta} E_1\left(\phi\right).$$
Since $\partial_t, \Omega$ are Killing, it follows immediately that 
$$\sum_{l=0}^1 \int_{r_1^*}^{r_2^*} \left(\left(\partial_t^m\partial_{r^*}^l\phi\left(t_*\right)\right)^2+|\nabb^k\partial_{r^*}^l\phi\left(t_*\right)|^2\right) dAdr^* \leq C t_*^{-3+\delta} E_1\left(\partial_t^m\Omega^k\phi\right),$$ for any $k,m$.\\
Using the equation $\Box_g \phi=0$, we get 
$$\int_{r_1^*}^{r_2^*} \left(\nabla^l \phi\left(t_*\right)\right)^2 dVol \leq C t_*^{-3+\delta} \sum_{k+m \leq l} E_1\left(\partial_t^m\Omega^k\phi\right).$$
The pointwise decay part of Main Theorem 1 follows form the standard Sobolev Embedding Theorem and Proposition 19.
\subsection{Improved Decay for $\partial_t\phi$}
To estimate the time derivatives of $\phi$, we follow an idea of Klainerman-Sideris \cite{KS}. The key observation is that the first derivatives of $\partial_t\phi$ are controlled with a weight of $\frac{1}{t-r^*}$ by a linear combination of first derivatives of $\phi$ and $\psi$. This extra weight would give extra decay to $\partial_t\phi$.\\
\begin{proposition}
Suppose $t+r^* \geq \max\{\frac{t}{2},\frac{|r^*|}{2}\}$. (This is true for example when $r^*$ is bounded below and $t$ is sufficiently large.)
\begin{enumerate}
\item $|\left(t-r^*\right)\partial_t^2\phi| \leq C\left(|\partial_t\psi|+|\partial_{r^*}\psi|+|\partial_t\phi|+|\partial_{r^*}\phi|+\left(1-\mu \right)|r^*||\lapp\phi|\right),$
\item $|\left(t-r^*\right)\partial_{r^*}\partial_t\phi| \leq C\left(|\partial_t\psi|+|\partial_{r^*}\psi|+|\partial_t\phi|+|\partial_{r^*}\phi|+\left(1-\mu \right)|r^*||\lapp\phi|\right),$
\item $|t\left(1-\mu\right)\nabb\partial_t\phi|\leq C(\left(1-\mu \right)|\nabb\psi|+|\partial_{r^*}\Omega\phi|.$
\end{enumerate}
\end{proposition}
\begin{proof}
Define $\Delta_g\phi=\left(1-\mu \right)^{-1}\partial_{r^*}\phi+\frac{2}{r}\partial_{r^*}\phi+\lapp\phi$. Then $\Box_g\phi=0$ reads $\left(1-\mu \right)\partial_t^2\phi=\Delta_g\phi$.\\
Recall that 
$$\psi=t\partial_t\phi+r^*\partial_{r^*}\phi.$$
Therefore,
$$\partial_t\psi-\partial_t\phi=t\partial_t^2\phi+r^*\partial_{r^*}\partial_t\phi,$$
$$\partial_{r^*}\psi-\partial_{r^*}\phi=r^*\partial_{r^*}^2\phi+t\partial_{r^*}\partial_t\phi.$$
Hence,
\begin{equation*}
\begin{split}
&t\left(\partial_t\psi-\partial_t\phi\right)-r^*\left(\partial_{r^*}\psi-\partial_{r^*}\phi \right)\\
=&t^2\partial_t^2\phi- \left(r^* \right)^2\partial_{r^*}^2\phi\\
=&\left( t^2-\left( r^*\right)^2\right)\partial_t^2\phi+ \left( r^*\right)^2\left( \left( 1-\mu \right)\Delta_g\phi-\partial_{r^*}^2\phi \right)\\
=&\left(t^2-\left( r^*\right)^2\right)\partial_t^2\phi+ \left( r^*\right)^2\left( \frac{2\left( 1-\mu \right)}{r}\partial_{r^*}\phi+ \left( 1-\mu \right)\lapp\phi \right).
\end{split}
\end{equation*}
Therefore, by re-arranging and dividing by $\left(t+r^*\right)$,
\begin{equation*}
\begin{split}
&|\left(t-r^*\right)\partial_t^2\phi|\\
=&|\frac{1}{t+r^*}(t\left(\partial_t\psi-\partial_t\phi \right)-r^* \left(\partial_{r^*}\psi-\partial_{r^*}\phi \right)-\left(r^*\right)^2\left(\frac{2\left(1-\mu \right)}{r}\partial_{r^*}\phi+ \left(1-\mu \right)\lapp\phi )\right)|\\
\leq &C\left(|\partial_t\psi|+|\partial_{r^*}\psi|+|\partial_t\phi|+|\partial_{r^*}\phi|+\left(1-\mu \right)|r^*||\lapp\phi|\right).
\end{split}
\end{equation*}
We have thus proved 1.
On the other hand, using again the above equality, we also have 
\begin{equation*}
\begin{split}
&\left( t-r^*\right)\partial_{r^*}\partial_t\phi\\
=&-\partial_t\psi+\partial_{r^*}\psi+\partial_t\phi-\partial_{r^*}\phi+t\partial_t^2\phi-r^*\partial_{r^*}^2\phi\\
=&-\partial_t\psi+\partial_{r^*}\psi+\partial_t\phi-\partial_{r^*}\phi+\left( t-r^* \right)\partial_t^2\phi+r^* \left(\left( 1-\mu \right)\Delta_g\phi-\partial_{r^*}^2\phi \right)\\
=&-\partial_t\psi+\partial_{r^*}\psi+\partial_t\phi-\partial_{r^*}\phi+\left( t-r^* \right)\partial_t^2\phi+\frac{2r^*\left( 1-\mu \right)}{r}\partial_{r^*}\phi+ \left( 1-\mu \right)|r^*|\lapp\phi.\\
\end{split}
\end{equation*}
This, together with 1, implies 2.\\
The proof of 3 is more direct. Using the definition of $S$, and that $\Omega$ is independent of $t$ and $r^*$,
$$\Omega\psi=r^*\partial_{r^*}\Omega\phi+t\partial_t\Omega\phi.$$
Thus, by noting that $\Omega$ and $r\nabb$ differ only by constant,
\begin{equation*}
\begin{split}
&|t\left( 1-\mu\right)\nabb\partial_t\phi|\\
\leq &|\left( 1-\mu \right)\nabb\psi|+| \left( 1-\mu \right)\frac{r^*}{r}\partial_{r^*}\Omega\phi|\\
\leq &C\left(\left( 1-\mu \right)|\nabb\psi|+|\partial_{r^*}\Omega\phi| \right).
\end{split}
\end{equation*}
\end{proof}
\begin{corollary}
$$\int_{\tilde{r}^*}^{ct_*} J^T_\mu \left(\partial_t\phi\right)n^\mu_{t_*} dVol_{t_*}\leq Ct_*^{-4+\delta}\left(\sum_{m=0}^1 \sum_{k=0}^1 E_0\left(\partial_t^m\Omega^k\phi\right)+E_1\left(\phi\right)\right), $$
for all $c <1$ and $\tilde{r}$. In particular, $\tilde{r}$ can be chosen as that given by Proposition 19.
\end{corollary}
\begin{proof}
We can consider $t_*$ large enough so that firstly, the assumption of of Proposition 22 holds and secondly, on the domain of integration, $\left(t_*-r^*\right) \sim t_*$.
\begin{equation*}
\begin{split}
&\int_{\tilde{r}^*}^{ct_*} J^T_\mu \left( \partial_t\phi\right)n^\mu_{t_*} dVol_{t_*}\\
=& \int_{\tilde{r}^*}^{ct_*} \left(\left(\partial_{r^*}\partial_t\phi \right)^2 +\left(\partial_t^2\phi\right)^2+\left(1-\mu \right)|\nabb\partial_t\phi|^2 \right) dVol_{t_*}\\
\leq & Ct_*^{-2}\int_{\tilde{r}^*}^{ct_*} \left(\left( \partial_t\psi \right)^2+ \left( \partial_{r^*}\psi \right)^2 +\left( 1-\mu \right)|\nabb\psi|^2 +\left( \partial_t\phi \right)^2+\left( \partial_{r^*}\phi \right)^2 +\left( 1-\mu \right)|\nabb\Omega\partial_t\phi|^2 +\left( \partial_{r^*}\Omega\phi \right)^2\right) dVol_{t_*}\\
\leq &Ct_*^{-2}\int_{\tilde{r}^*}^{ct_*} J^T\left( \psi \right)+J^T\left( \phi \right)+J^T\left( \Omega\phi \right)+J^T\left( \partial_t\Omega\phi\right) dVol_{t_*}.
\end{split}
\end{equation*}
The corollary follows from Theorem 5 and 10.
\end{proof}
\begin{corollary}
$$|\partial_t\phi\left(v^*\right)|^2\leq Cv_*^{-4+\delta}\left(\sum_{m=0}^2 \sum_{k=0}^{4-m} E_0\left(\partial_t^m \Omega^k\phi\right)+  \sum_{m=0}^{2} \sum_{k=0}^{2-m} E_1\left(\partial_t^m\Omega^k\phi\right)\right),$$ if $r^*\leq  \frac{t_*}{2}$.
\end{corollary}
\begin{proof}
We prove a Sobolev-type inequality. We first work on $\mathbb R^3$. We claim that for $u\in C^\infty_c\left(\mathbb R^3\right)$, $$||u||_{L^\infty\left(\mathbb R^3\right)} \leq C ||u||^{\frac{1}{2}}_{\dot{H}^1\left(\mathbb R^3\right)}||u||^{\frac{1}{2}}_{\dot{H}^2\left(\mathbb R^3\right)}.$$
We give a simple proof using Littlewood-Paley theory. Let $N \in 2^{\mathbb Z}$ be a dyadic number, $\chi\left(\xi\right)$ be a radial cutoff function which is supported in $\{|\xi| < 2\}$ and is identically 1 in $\{|\xi| < 1\}$. Define the Littlewood-Paley operators $P_N$ by $\widehat{P_N u} =\left(\chi\left(\frac{\xi}{N}\right)-\chi\left(\frac{2\xi}{N}\right)\right)\hat{u}$.\\
Since the inequality claimed is invariant under scaling $u\left(x\right) \to \lambda u\left(x\right)$ and $u\left(x\right) \to u\left(\lambda x\right)$, we can assume that $||u||_{\dot{H}^1\left(\mathbb R^3\right)}=||u||_{\dot{H}^2\left(\mathbb R^3\right)}=1$.
Then, by Bernstein inequality,
$$||P_N u||_{L^2\left(\mathbb R^3\right)}\leq \min\{CN^{-1}, CN^{-2}\}.$$
Therefore, by Bernstein inequality again,
$$||u||_{L^\infty\left(\mathbb R^3\right)} \leq C \sum_N N^{\frac{3}{2}} ||P_N u||_{L^2\left(\mathbb R^3\right)}\leq C \left(\sum_{N\geq N_0} N^{-\frac{1}{2}}+\sum_{N < N_0} N^{\frac{1}{2}}\right)\leq C.$$
We note that a variant of this is true. We have for $u\in C^\infty_c\left(\mathbb R^3\right)$, 
$$||u||_{L^\infty\left(\mathbb R^3\setminus B_r\left(0\right)\right)} \leq C ||u||^{\frac{1}{2}}_{\dot{H}^1\left(\mathbb R^3\setminus B_r\left(0\right)\right)}||u||^{\frac{1}{2}}_{\dot{H}^2\left(\mathbb R^3\setminus B_r\left(0\right)\right)}.$$
This is true because one can extend $u$ into $B_r\left(0\right)$ without increasing the $\dot{H}^1$ or $\dot{H}^2$ norm.\\
We now apply this to a cutoff version of $\partial_t\phi$.\\
Let $\chi=\left\{\begin{array}{clcr}1&|x|\le 1\\0&|x|\ge 1.1\end{array}\right.$.
On $t=t_*$, let $\tilde{\phi}=\chi\left(\frac{r^*}{0.5t_*}\right)\phi$ for $r \geq \tilde{r}$, where $\tilde{r}$ is as in Proposition 19.\\
The $\dot{H}^1$ norm is controlled with Corollary 23 and Theorem 5.4.
\begin{equation*}
\begin{split}
||\partial_t\tilde{\phi}||_{\dot{H}^1\left( \mathbb R^3 \setminus B_{\tilde{r}}\left(0\right)\right)} \leq &C\int_{\tilde{r}^*}^{\frac{0.55t_*}{2}} J^T_\mu\left(\partial_t\phi\right) n^\mu_{t_*} dVol_{t_*}+Ct_*^{-2}\int_{\tilde{r}^*}^{\frac{0.55t_*}{2}}\left(\partial_t\phi\right)^2 r^2 \left(1-\mu \right)dAdr^*\\
\leq &Ct_*^{-4+\delta}\left( \sum_{m=0}^1\sum_{k=0}^1 E_0\left( \partial_t^m\Omega^k\phi\right)+E_1\left( \phi \right)\right).
\end{split}
\end{equation*}
The $\dot{H}^2$ norm can be controlled similarly once we note that using the equation, we have for $r \geq \tilde{r}$, 
\begin{equation*}
\begin{split}
|\partial_{r^*}^2\partial_t\phi| \leq &
C\left(|\partial_t^3\phi|+\frac{1}{r}|\partial_{r^*}\partial_t\phi|+|\nabb^2\partial_t\phi|\right)\\
\leq &C\left(|\partial_t^3\phi|+|\partial_{r^*}\partial_t\phi|+|\nabb\Omega\partial_t\phi|\right).
\end{split}
\end{equation*}
Therefore, for $t=t_*$,
\begin{equation*}
\begin{split}
&||\partial_t\tilde{\phi}||_{\dot{H}^2\left(\mathbb R^3 \setminus B_{\tilde{r}}\left(0\right)\right)}\\
\leq &C\int_{\tilde{r}^*}^{\frac{0.55t_*}{2}} \left( J^T_\mu \left( \partial_t^2\phi \right)+J^T_\mu \left( \partial_t\Omega\phi \right)+J^T_\mu \left( \partial_t\phi \right)\right) n^\mu_{t_*} dVol_{t_*}\\
&+Ct_*^{-2}\int_{\tilde{r}^*}^{\frac{0.55t_*}{2}}\left( \left( \partial_t^2\phi\right)^2+\left( \partial_t\Omega\phi\right)^2+\left( \partial_t\phi \right)^2\right) r^2 \left(1-\mu \right)dAdr^*\\
\leq &Ct_*^{-4+\delta}\left(\sum_{m=0}^2 \sum_{k=0}^2 E_0\left( \partial_t^m\Omega^k\phi\right)+\sum_{m=0}^1 \sum_{k=0}^1 E_1\left(\partial_t^m\Omega^k\phi \right)\right).
\end{split}
\end{equation*}
Therefore,
\begin{equation*}
\begin{split}
||\partial_t\phi||_{L^\infty\left( \{\tilde{r}^*\leq r^*\leq \frac{t_*}{2}\}\right)}\leq &||\partial_t\tilde{\phi}||_{L^\infty\left( \mathbb R^3\setminus B_{\tilde{r}}\left( 0\right)\right)}\\
\leq &C ||\partial_t\tilde{\phi}||^{\frac{1}{2}}_{\dot{H}^1\left(\mathbb R^3\setminus B_{\tilde{r}}\left( 0\right)\right)}||\partial_t\tilde{\phi}||^{\frac{1}{2}}_{\dot{H}^2\left(\mathbb R^3\setminus B_{\tilde{r}}\left(0\right)\right)}\\
\leq &Ct_*^{-4+\delta}\left( \sum_{m=0}^2\sum_{k=0}^2 E_0\left( \partial_t^m\Omega^k\phi\right)+\sum_{m=0}^1\sum_{k=0}^1  E_1\left( \partial_t^m\Omega^k\phi \right)\right).
\end{split}
\end{equation*}
In particular, for $t$ sufficiently large, this $L^\infty$ estimate holds on sets of compact $r^*$. Noting that the $L^\infty$ norm controls the $L^2$ norm on compact sets, we have
$$\int_{\tilde{r}}^{\left(\left(1.2\right)r_0\right)^*} \left(\partial_t\phi\right)^2 dVol_{t_*} \leq Ct_*^{-4+\delta}\left(\sum_{m=0}^2 \sum_{k=0}^2 E_0\left(\partial_t^m\Omega^k\phi\right)+\sum_{m=0}^1  \sum_{k=0}^1 E_1\left( \partial_t^m\Omega^k\phi \right)\right).$$
We also have, by Corollary 23,
\begin{equation*}
\begin{split}
&\sum_{m=0}^1 \sum_{k=0}^{3-m} \int_{\tilde{r}}^{\left(\left( 1.2\right) r_0\right)^*} J^T_\mu \left( \partial_t^m\Omega^k\left( \partial_t\phi \right)\right) n^\mu_t dVol_{t_*}\\
\leq &C t_*^{-4+\delta} \left( \sum_{m=0}^2 \sum_{k=0}^{4-m} E_0\left( \partial_t^m\Omega^k\phi \right)+ \sum_{m=0}^1 \sum_{k=0}^{4-m} E_1 \left( \partial_t^m\Omega^k\phi  \right) \right).
\end{split}
\end{equation*}
The corollary then follows from the Sobolev Embedding Theorem and Proposition 19.
\end{proof}

\section{Acknowledgments}
The author thanks his advisor Igor Rodnianski for suggesting the problem and for sharing numerous insights. He thanks Mihalis Dafermos, Gustav Holzegel and Igor Rodnianski for very helpful comments on preliminary versions of the manuscript.
\bibliographystyle{hplain}
\bibliography{Decay}
\end{document}